\documentclass[a4paper,UKenglish,cleveref,autoref,thm-restate]{lipics-v2021}
\usepackage[utf8]{inputenc}
\usepackage{amsmath,amssymb}
\allowdisplaybreaks[1]
\usepackage{lmodern}
\usepackage[T1]{fontenc}
\usepackage{stmaryrd}
\usepackage{microtype}
\makeatletter
\newenvironment{compactenum}[1][]{\ifx\relax#1\relax\enumerate\else\enumerate[#1]\fi}{\endenumerate}
\newenvironment{compactitem}{\itemize}{\enditemize}
\makeatother
\usepackage{tikz}
\usetikzlibrary{shapes.misc,decorations.pathreplacing,arrows.meta}
\usepackage{bm}
\usepackage{etoolbox}
\usepackage{algorithm}
\usepackage{algpseudocode}
\usepackage{prftree}

\let\phi\varphi
\newcommand*{\Nat}{\mathbb{N}}
\renewcommand*{\Re}{{\mathbb{R}_{\geq 0}}}
\newcommand*{\Int}{\mathbb{Z}} 
\newcommand*{\Events}{\mathcal{P}}
\newcommand*{\tms}{\tau}
\newcommand*{\tmsdiff}{\delta}
\newcommand*{\ev}{\sigma}

\newcommand*{\len}[1]{|#1|}

\newcommand*{\Fintraces}{\mathcal{T}}
\newcommand*{\Netraces}{\mathcal{T}^+}

\newcommand*{\Monitor}{M}

\newcommand*{\States}{Q}
\newcommand*{\Init}{q_0}
\newcommand*{\Step}{\mu}
\newcommand*{\Out}{y}
\newcommand*{\Tuple}[1]{\langle #1\rangle}
\newcommand*{\Sem}[1]{\llbracket #1\rrbracket}
\newcommand*{\Change}{C}

\newcommand*{\Extract}[1]{\mathord{\downarrow} #1}
\newcommand*{\Build}[1]{\mathord{\uparrow} #1}
\newcommand*{\BE}[1]{\mathord{\updownarrow} #1}

\newcommand*{\bisim}{\simeq}
\newcommand*{\deepbisim}{\cong}

\newcommand*{\hypermodels}{\mathrel{\bm{|}}\mathrel{\mkern-4mu}\mathrel{\bm{=}}}
\newcommand*{\tforall}{\bm{\forall}}
\newcommand*{\texists}{\bm{\exists}}
\DeclareMathOperator{\dom}{dom}

\newcommand*{\Props}{\mathcal{P}}

\newcommand*{\HAlways}{\blacksquare}

\newcommand*{\Once}{\blacklozenge}
\newcommand*{\Prev}{\mathord{\raisebox{-0.3pt}{\tikz{\fill (0,0) circle[radius=3.3pt];}}}}

\newcommand*{\Since}{\mathbin{\mathsf{S}}}
\newcommand*{\UUntil}{\mathbin{\mathsf{U}}}

\newcommand*{\True}{\mathsf{true}}
\newcommand*{\False}{\mathsf{false}}
\newcommand*{\Imp}{\rightarrow}
\newcommand*{\Iff}{\leftrightarrow}
\newcommand*{\sattptl}{\textsc{sat1TPTL}}

\newcommand*{\HTrue}{\textsf{\bfseries true}}

\newcommand*{\HNeg}{\bm{\lnot}}
\newcommand*{\HImp}{\mathrel{\bm{\rightarrow}}}
\newcommand*{\HIff}{\mathrel{\bm{\leftrightarrow}}}
\newcommand*{\HConj}{\mathbin{\bm{\land}}}
\newcommand*{\HDisj}{\mathbin{\bm{\lor}}}
\newcommand*{\HSince}{\mathbin{\textsf{\bfseries \normalshape S}}}
\newcommand*{\HSinceS}{\mathbin{\textsf{\bfseries \normalshape S}^{\bm{<}}}}
\newcommand*{\HLift}[2]{\{#1\}_{#2}}

\newcommand*{\Embed}[1]{\llparenthesis #1\rrparenthesis}

\DeclareMathOperator{\img}{img}
\newcommand*{\Eval}{\mathit{eval}}

\definecolor{red}{RGB}{200,80,80}
\newrobustcmd{\todo}[1]{\textcolor{red}{[TODO: {#1}]}}
\newrobustcmd{\dni}[1]{\textcolor{blue}{[DN: {#1}]}}

\newif\ifshowlean \showleantrue
\definecolor{leangreen}{RGB}{0,122,92}
\newcommand*{\leanrepo}{https://github.com/francois-hublet/policy-change-lean/blob/main/PolicyChange}
\newcommand{\leanpill}[1]{%
  \tikz[baseline=(b.base)]\node[inner sep=1.3pt, outer sep=0pt, rounded corners=2pt,
    #1](b){\sffamily\fontsize{6}{6}\selectfont\bfseries\strut Lean};}
\newcommand{\leanfull}{\leanpill{fill=leangreen, text=white}}
\newcommand{\leanpartial}{\leanpill{draw=leangreen, line width=0.5pt, text=leangreen}}
\newcommand{\leanref}[3]{\href{\leanrepo/#1\#L#2}{\textcolor{leangreen}{\texttt{\detokenize{#3}}}}}
\newcommand{\leanmargin}[4]{\ifshowlean\leavevmode\marginpar{\raggedright\href{\leanrepo/#2\#L#3}{#1}\\[-2.5pt]%
  {\fontsize{5.5}{6}\selectfont\href{\leanrepo/#2\#L#3}{\textcolor{leangreen}{\texttt{\detokenize{#4}}}}}}\fi}
\newcommand{\leanok}[3]{\leanmargin{\leanfull}{#1}{#2}{#3}}
\newcommand{\leanpart}[3]{\leanmargin{\leanpartial}{#1}{#2}{#3}}

\hideLIPIcs
\nolinenumbers

\title{Policy Change for Treelike Monitors}

\author{François Hublet}{Institute of Information Security, Department of Computer Science, ETH Zurich, Switzerland \and \url{https://hblt.eu/sci}}{francois.hublet@inf.ethz.ch}{https://orcid.org/0000-0001-5419-3125}{}
\author{Dhruv Nevatia}{Institute of Information Security, Department of Computer Science, ETH Zurich, Switzerland \and \url{https://notdhruv.github.io}}{dhruv.nevatia@inf.ethz.ch}{https://orcid.org/0009-0008-0845-6754}{}
\author{Joshua Schneider}{Institute of Information Security, Department of Computer Science, ETH Zurich, Switzerland}{}{}{}

\authorrunning{F. Hublet, D. Nevatia and J. Schneider}

\Copyright{François Hublet, Dhruv Nevatia and Joshua Schneider}



\relatedversion{}
\supplementdetails[subcategory={Lean formalization}]{Software}{https://github.com/francois-hublet/policy-change-lean}

\begin{document}

\maketitle

\begin{abstract}
  We study the policy change problem that arises in the runtime verification of long-running systems.
  The online monitors typically used in this context are generally \emph{treelike}, in that they
  maintain substates that monitor subformulae of the target policy.
  We consider when and how the policy can be changed while the monitored system is running
  by only exploiting the information stored in the monitor's state.
  This is relevant, for example, to account for new system functionality or changes in
  regulatory requirements.
  We formally define the policy change problem in a general setting,
  independent of any specific (treelike) monitor implementation.
  We then show that policy change for past-time metric temporal logic (pMTL) is decidable and complete for the non-primitive recursive class $\mathbf{F}_{\omega^\omega}$,
  while with discrete-time semantics it is EXPSPACE-complete, and PSPACE-complete for past-time LTL.
  Without the treelike assumption, the discrete-time problems have the same complexity.
\end{abstract}

\section{Introduction}

Runtime verification is a lightweight approach to formal verification.
Its primary tool is the \emph{monitor}, which checks individual runs of a system for compliance with a \emph{policy}.
Policies are typically trace properties expressed using a suitable policy language.
\emph{Online} monitors run simultaneously to the monitored system without requiring that all
system events be stored in a log file.
They are thus especially useful if the system runs indefinitely, as is common for, e.g., web servers.

Most practically used monitoring algorithms and tools~\cite{conf/fm/PnueliZ06,journals/entcs/ThatiR05,journals/jacm/BasinKMZ15,journals/fmsd/BasinBKT19,conf/atva/RaszykBKT19,fmsd/HavelundPU20,conf/tacas/LimaHRTY23} operate on a treelike data structure
that maintains a monitor substate for every subformula of the target policy. At every point in time, the satisfaction
of every subformula is computed bottom-up, and substates are updated to contain enough information to continue monitoring the trace
in the future. 
These monitors, which we call \emph{treelike}, are the primary focus of
this paper.

Long-running systems may be reconfigured to a substantial degree at runtime.
For example, consider a registry, operated by some public authority,
that stores records in a database.
Initially, our registry may be required by law to notify users about certain administrative actions within
30~days.
If that legal deadline is shortened to 15~days, this should be reflected in the policy so that the monitor can detect if the system fails to implement the new requirement.
A~trivial approach to changing the policy consists in stopping the old monitor and starting a new one, initialized using the existing trace.
This has two major drawbacks:
first, it assumes that a sufficient history of events has been saved; second, it requires re-processing a potentially large fragment of the event trace.

\looseness=-1
\emph{Problem statement.}  We say that a monitor supports \emph{policy change} from policy $\phi_1$ to policy $\phi_2$ if,
after monitoring a trace prefix against $\phi_1$, it can continue monitoring the trace's remainder against $\phi_2$ using only the information in its state.
This avoids the two drawbacks mentioned above.
For example, the policy change corresponding to shortening the registry's notification deadline is always possible.
When an administrative action has not been reported yet, the monitor must keep track of the time elapsed since the action's most recent occurrence, up to the old deadline.
This information is sufficient to judge whether the new, shorter deadline is violated.
However, policy change is not always possible, as an efficient monitor's state typically does not contain enough information to reconstruct all past events.
As an extreme example, consider the case where $\phi_1$ is the trivially satisfied policy,
i.e., $\phi_1 = \True$. Clearly, a monitor for $\phi_1$ need not store any information in its
state, and hence this state could not be used to construct a valid monitor state for any non-trivial $\phi_2$.
This leads us to the following research question: under which condition is policy change possible?
In this paper, we answer %
this question for past-time metric temporal logic (pMTL) with both real-time and discrete-time semantics~\cite{journals/rts/Koymans90,journals/iandc/AlurH93} as the policy language.

\looseness=-1
\emph{Related work.} The importance of policy change has been recognized in the context of self-adaptive systems,
which use monitors to enact adaptation procedures if a violation of the system's requirements is detected~\cite{conf/iwssd/FeatherFLP98}.
Carwehl et~al.~\cite{conf/seams/CarwehlVRG23} present a policy change approach for monitors constructed from timed automata templates, which correspond to a fragment of future-time MTL.
They work with a manually crafted catalog of so-called \emph{property adaptation patterns} that directly modify the monitor's automaton and current state.
Yuan~\cite{Yuan22} presents a sound but incomplete algorithm that can perform changes between some pairs of first-order policies in the MonPoly tool~\cite{journals/jacm/BasinKMZ15}.
In contrast to these works, we consider a decision procedure that determines, for every policy pair, whether adaptation is possible, independent of the specific monitoring algorithm used.
Changing specifications at runtime has also been studied from the perspective of controller synthesis from modal sequence diagrams~\cite{conf/icse/GhezziGM12}.
There, the main problem is to delay the change until it can be done safely, whereas monitors are expected to respond to whatever changes occur in their environment.

\emph{Formalization}. All definitions and the reductions underlying our main results are formalized in Lean~4 (\url{https://github.com/francois-hublet/policy-change-lean}); results are marked with \leanfull{} or \leanpartial{} in the margin, linking to the formalization (see Appendix~\ref{apx:lean}).

\emph{Contributions and outline}. After reviewing runtime monitoring with treelike monitors and pMTL (Section~\ref{sec:prelim}), we make the following contributions:
\pagebreak[3]
\begin{compactenum}
  \item We give a generalised definition of policy change for treelike monitors independent of any specific monitor implementation (Section~\ref{sec:change}).
  \item We solve the decision problem associated with policy change for real-time pMTL by characterizing it as an inclusion between equivalence relations on trace prefixes, expressing these relations in a new hyperlogic, and reducing the resulting entailment problem to one-variable past TPTL (1-pTPTL) satisfiability over finite timed words~\cite{book/Haase10}. We prove that the problem is $\mathbf{F}_{\omega^\omega}$-complete; the lower bound follows via a reduction from pMTL satisfiability on finite traces (Section~\ref{sec:continuous}).
  \item We prove that policy change for discrete-time pMTL is EXPSPACE-complete
    and policy change for pLTL is PSPACE-complete (Section~\ref{sec:discrete}).
  \item We show that if we remove the assumption that monitors are treelike, policy change reduces to a refinement check between Nerode congruences. For discrete-time pMTL, policy change for non-treelike monitors is EXPSPACE-complete. For pLTL, it is PSPACE-complete (Section~\ref{sec:general}).
\end{compactenum}
The following table summarizes the results proven in Sections~\ref{sec:continuous} and~\ref{sec:discrete}. Proofs of all lemmata and theorems are provided in Appendix~\ref{apx:proofs}.

\vspace{2mm}
\resizebox{0.95\textwidth}{!}{
  \begin{tabular}{|c|c|c|}
    \cline{2-3}
    \multicolumn{1}{c|}{} & \textbf{Discrete time} & \textbf{Real time} \\
    \cline{1-3}
    \textbf{pMTL} & EXPSPACE-complete (Theorems~\ref{thm:pc ub discrete}, \ref{thm:pc lb discrete}) & $\mathbf{F}_{\omega^\omega}$-complete (Theorems~\ref{thm:upper}, \ref{theorem:pc pmtl real-time lb}) \\
    \cline{1-3}
    \textbf{pLTL} & PSPACE-complete (Theorems~\ref{thm:pc ub discrete}, \ref{thm:pc lb discrete}) & $-$ \\
    \hline
  \end{tabular}}
\vspace{2mm}

\section{Preliminaries}
\label{sec:prelim}

\subsection{Past-time Metric Temporal Logic}

Traces represent sequences of actions or states that occur during a system's execution.
We consider discrete systems that can be described by a finite set of instantaneous actions or state changes modeled as events.
In runtime verification, traces are obtained by observing the running system and are thus finite.
We therefore model traces as finite timed words, using non-negative real numbers for the \emph{real} time domain and natural numbers for the \emph{discrete} time domain. Henceforth, we shall only consider real-time semantics unless stated otherwise.

Let $\Events$ be an alphabet of \emph{predicates}, $\Nat = \{0,1,...\}$, and $\mathbb{R}$ denote the set of reals. A~trace $s$ of length $\len{s}$ is a pair $(\ev,\tmsdiff)$, where $\ev$ is a sequence of predicate sets $(\ev_i)_{0 \leq i < \len{s}} \in (2^{\Events})^{|\sigma|}$, and $\tmsdiff$ is a sequence of \emph{delays} $(\tmsdiff_i)_{0 \leq i < \len{s}} \in \Re^{|\sigma|}$. We call the index $i$ the \emph{time-point}, and the corresponding tuple $(\ev_i,\tmsdiff_i)$ an \emph{event}.
Each delay measures the amount of time that has passed since the preceding event or since the start of time for the first event.%
\footnote{The logics we study (pMTL and pLTL) cannot express properties relating to absolute time.
  Therefore, using delays simplifies policy change slightly since monitors need not store the most recent absolute timestamp.}
If the sequence $(\tmsdiff_i)$ is clear from context, the \emph{(absolute) timestamp} $\tms_i$ is defined by $\sum_{j\leq i} \tmsdiff_j$.
We write $\Fintraces$ ($\Netraces$, resp.) for the set of all (non-empty, resp.) traces and $s \cdot s'$ for the concatenation of two traces $s$ and $s'$ (as event sequences, so that the delays of $s'$ are relative to the end of $s$).

A~\emph{policy} is a set of traces.
A~trace in the set \emph{satisfies} the policy, otherwise it \emph{violates} it.
We consider policies definable using \emph{past-time metric temporal logic} (pMTL).
The pMTL formulae over predicates $\Props = \{P,Q,\dots\}$ are defined as
\begin{align*}
  \phi ::= \True \mid P \mid \lnot \phi \mid \phi_1 \lor \phi_2 \mid \Prev_I \phi \mid \phi_1 \Since_I \phi_2,
\end{align*}
where $P$ ranges over $\Props$ and $I$ ranges over intervals with endpoints in $\Nat\cup \{\infty\}$.
An~interval is either empty ($\emptyset$), represented by its endpoints in the bounded
case ($[a,b]$ with $a \leq b$, $(a, b)$, $[a, b)$ or $(a, b]$ with $a < b$, $a,b \in \Nat$),
or by its lower bound in the unbounded case ($[a,\infty)$ or $(a,\infty)$, $a \in \Nat$).
An~omitted interval subscript defaults to $[0,\infty)$.
The operator $\Prev_I$ (``previous'') states that its operand held at the preceding time-point, which lies a time distance in $I$ in the past.
The operator $\Since_I$ (``since'') states that its right operand held at some earlier (or the current) time-point whose distance from the present lies in $I$, and that its left operand has held at all time-points after that one, up to and including the current one.
We define other operators as abbreviations: $\False \equiv \lnot \True$, $(\phi_1 \land \phi_2) \equiv \lnot (\lnot \phi_1 \lor \lnot \phi_2)$, $(\phi_1 \Imp \phi_2) \equiv (\lnot \phi_1 \lor \phi_2)$, $\Once_I \phi \equiv (\True \Since_I \phi)$, and $\HAlways_I \phi \equiv \lnot (\True \Since_I \lnot \phi)$.
We use the convention that unary operators ($\lnot,\Prev,\Once,\HAlways$) bind the strongest, followed by $\Since$, conjunction, disjunction, and implication. The size $|\phi|$ of a formula is its total number of symbols, where each interval bound counts with the length of its binary encoding.

\begin{example}
  Consider the policy ``whenever the system processes data, it must notify the user within 30 days.''
  Compliance with this policy can be checked by monitoring the policy
  $\phi_A = \neg (\neg \mathsf{notify} \Since_{(30,\infty)} \mathsf{process})$, which reads
  ``it is not the case that there has been a $\mathsf{process}$ event more than 30 days ago and that there has not been any $\mathsf{notify}$ event since then.'' If the delay is reduced to 15 days, the formula becomes
  $\phi_B = \neg (\neg \mathsf{notify} \Since_{(15,\infty)} \mathsf{process})$.
  To specify ``once a notification has been sent, the next notification must be sent within 30 days,''
  we can write $\phi_C = \neg (\neg \mathsf{notify} \Since_{(30,\infty)} \mathsf{notify})$.
  Note that $\phi_C$ does not require a notification in every window of 30 days: it is satisfied by traces without any $\mathsf{notify}$ event.
\end{example}

Let $s$ be a trace over $\Props$ and $i < \len{s}$. 
The relation $\Tuple{s,i} \models \phi$, specifying the \emph{satisfaction} of $\phi$ on $s = (\ev,\tmsdiff)$ at $i$,
is the least relation closed under the following rules:
\begin{align*}
  \Tuple{s,i} &\models \True\\
  \Tuple{s,i} &\models P &&\text{if } P \in \ev_i \\
  \Tuple{s,i} &\models \lnot \phi &&\text{if } \Tuple{s,i} \not\models \phi\\
  \Tuple{s,i} &\models \phi_1 \lor \phi_2 &&\text{if } \Tuple{s,i} \models \phi_1 \text{ or } \Tuple{s,i} \models \phi_2\\
  \Tuple{s,i} &\models \Prev_I \phi &&\text{if $i > 0$, $\tau_i - \tau_{i-1} \in I$, and $\Tuple{s,i-1} \models \phi$} \\
  \Tuple{s,i} &\models \phi_1 \Since_I \phi_2 &&\text{if there exists $j \leq i$ such that $\tau_i - \tau_j \in I$,} \\[-\jot]
      &&&\text{$\Tuple{s,j} \models \phi_2$, and $\Tuple{s,k} \models \phi_1$ for all $k$ with $j < k \leq i$}
\end{align*}
The policy $\Sem{\phi}$ consists of all non-empty traces $s$ such that $\Tuple{s,\len{s}-1} \models \phi$.

Past-time linear temporal logic (pLTL) can be viewed as the fragment of pMTL where all intervals of $\Prev$ and $\Since$ operators
are $[0,\infty)$. On this fragment, the truth value of a formula does not depend on the relative or absolute timestamps.

\subsection{Online Monitors}
\label{sec:online monitors}

An~\emph{online monitor} for a policy receives a (possibly unbounded) trace incrementally and determines whether the current prefix satisfies the policy.
We formalize an online monitor $\Monitor^\phi$ for a policy $\phi$ as a deterministic transducer $M^\phi = \Tuple{\States^\phi,\Init^\phi,\Step^\phi,\Out^\phi}$, where $\States^\phi$ is the monitor's state space, $\Init^\phi \in \States^\phi$ is the initial state, $\Step^\phi : \States^\phi \times (2^\Props \times \Re) \to \States^\phi$ is the step function, and $\Out^\phi : \States^\phi \times (2^\Props \times \Re) \to \{0,1\}$ specifies the output. %
Any monitor that computes its verdict from the successor state $\Step^\phi(q,x)$ and $x$ is of this form, since $\Step^\phi$ is deterministic. Conversely, a Mealy machine can be brought to that form by also storing the last verdict in its state, which does not affect bisimilarity. We use the Mealy formulation because it does not force monitors to store this verdict; e.g., the natural tester for $\Prev X$ only keeps whether $X$ held at the last time-point.
In the discrete-time setting, $\Re$ is replaced by $\Nat$ throughout.
We canonically lift $\Step^\phi$ and $\Out^\phi$ to traces such that $\Out^\phi$ is defined only for non-empty traces and returns the last output, i.e., $\Out^\phi(q,s\cdot x) = \Out^\phi(\Step^\phi(q,s),x)$.
In the following, we refer to ``online monitors'' as just ``monitors''.
A~\emph{generic monitor} $\Monitor$ is a mapping from policies $\phi$ to monitors $\Monitor^\phi$.
The generic monitor is correct iff it satisfies
$\Out^\phi(\Init^\phi,s) = (\text{if } s \in \Sem{\phi}\text{ then }1\text{ else } 0)$
for all pMTL formulae $\phi$ and non-empty traces $s$ over $\Props$.
In this paper, all (generic) monitors are assumed to be correct.

\subsection{Treelike Monitors}
\label{sec:treelike}

Most monitors found in existing work rely on trees of \emph{temporal testers}~\cite{conf/fm/PnueliZ06} associated with
the subformulae of the target policy.
Temporal testers compute a sequence of truth values that indicate whether the subformula they monitor is satisfied or not
at each time-point.
Thati and Ro\c{s}u's dynamic programming algorithm~\cite{journals/entcs/ThatiR05} is constructed in this way;
when restricted to pMTL, it is an online monitor as in Section~\ref{sec:online monitors}. %
Similarly, most state-of-the-art tools such as Aerial~\cite{journals/fmsd/BasinBKT19}, Hydra~\cite{conf/atva/RaszykBKT19}, DejaVu~\cite{fmsd/HavelundPU20}, Explanator~\cite{conf/tacas/LimaHRTY23}, and MonPoly~\cite{journals/jacm/BasinKMZ15}, a monitor for metric first-order temporal logic (MFOTL), use
temporal testers internally when operating on pMTL specifications.

Next, we define temporal testers and substates.
Let $X,Y$ be two \emph{fresh} predicates, which stand for the operands of unary and binary operators.
The set of all pMTL operators is $\Omega = \{\True, \lnot X, X \lor Y\} \cup \Props \cup \bigcup_I \{\Prev_I X, X \Since_I Y\}$.
Each $o \in \Omega$ is itself a pMTL formula over $\Props \cup \{X,Y\}$.
A~\emph{temporal tester}~\cite{conf/fm/PnueliZ06} for $o$ is a monitor $t = \Tuple{\States^t,\Init^t,\Step^t,\Out^t}$ for $o$ in the sense of Section~\ref{sec:online monitors}, over the extended alphabet $2^{\Props \cup \{X,Y\}} \times \Re$. At each time-point, the temporal temporal tester reads the delay, the current event, and the truth values of the operands (encoded by the presence of $X$ and $Y$), and outputs the truth value of $o$ at that time-point.
Temporal testers exist for every operator in $\Omega$; e.g., a (very inefficient) tester for $X \Since_I Y$ can store the whole prefix of $X$s and $Y$s read so far and
recompute the truth value of the since at every time-point.
Practical testers, such as those of Hydra~\cite{conf/atva/RaszykBKT19} or of Thati and Ro\c{s}u's algorithm~\cite{journals/entcs/ThatiR05}, store only the timestamps of the relevant past time-points; see Example~\ref{ex:testers}.

The syntax tree of every pMTL formula can be seen as a finite, rooted tree labeled by operators in $\Omega$.
A~\emph{monitor tree $T^\phi$ for $\phi$} is a tree isomorphic to $\phi$'s syntax tree, where each vertex $v$ is labeled by a temporal tester $t_v$ for the operator $o_v$ labeling~$v$.
We write $\phi_v$ for the subformula of $\phi$ rooted at $v$.
The monitor tree $T^\phi$ gives rise to a monitor $\Tuple{\States^{T^\phi},\Init^{T^\phi},\Step^{T^\phi},\Out^{T^\phi}}$.
Its state space $\States^{T^\phi} = \prod_v \States^{t_v}$ consists of trees of the same shape as $T^\phi$ whose vertices $v$ are labeled by states of the testers $t_v$, and $\Init^{T^\phi}$ labels each $v$ with $\Init^{t_v}$.
Given a vertex $v$ and $q \in \States^{T^\phi}$, the label $q(v) \in \States^{t_v}$ is called the \emph{substate} of $q$ at $v$.
Each new event $(\ev_i,\tmsdiff_i)$ received by the monitor is processed bottom-up:
every vertex $v$ with children $v_1$ (and $v_2$) feeds its tester $t_v$ with the extended event $(\ev_i \cup \{X \mid b_{v_1} = 1\} \cup \{Y \mid b_{v_2} = 1\},\tmsdiff_i)$, where $b_{v_1}$ and $b_{v_2}$ are the outputs just produced by the children; $t_v$ updates its substate using $\Step^{t_v}$ and produces its own output $b_v$ using $\Out^{t_v}$.
The monitor's output is the output $b_{\mathrm{root}}$ of the root.
By induction on the tree, $b_v = 1$ iff $\phi_v$ holds at the current time-point, so $T^\phi$ is a correct monitor for $\phi$.
Intuitively, a monitor tree is a monitor whose state reflects the syntax of the monitored formula.

\begin{example}\label{ex:testers}
  Figure~\ref{fig:tree} shows monitor trees for  $\phi_A$,  $\phi_B$, and  $\phi_C$. These formulae contain atoms $\mathsf{notify}$ and  $\mathsf{process}$, 
  unary $\neg$ operators,
  and binary $\Since_{(x,\infty)}$ operators.
  The temporal tester for an atom $\mathsf{e}$ $\in \{\mathsf{notify},\mathsf{process}\}$ is stateless
  and its output function is $\Out^{\mathsf{e}}(q,(\ev,\tmsdiff)) = (\mathsf{e} \in \ev)$.
  Similarly, 
  the temporal tester for $\neg$ is stateless and outputs the negation of the
  Boolean $b$ it receives from its child (i.e., $\Out^{\lnot}(q,(\ev,\tmsdiff)) = (X \notin \ev)$).
  The temporal tester for $\Since_{(x,\infty)}$ is more complex. Its state is a pair $(c,w)$, where $c$ is a real number and $w$ is a real number or the special value $\bot$. The value $c$ stores the latest timestamp read so far. The value
  $w$ stores the timestamp of the \emph{earliest} time-point at which the right operand ($Y$) held and such that the left operand ($X$) has held at all later time-points until the present.
  If such a time-point does not exist, then $w$ is $\bot$. +The step function first computes the new timestamp $c + \tmsdiff$ and then uses the
  function $\mathsf{upd}$ shown in Figure~\ref{fig:tree} to update $w$ for
  every new time-point. The corresponding output function checks whether, after this update, $w \neq \bot$ and $(c+\tmsdiff) - w > x$.
\end{example}

We do not restrict our analysis to monitor trees since this would be overly restrictive, excluding,
e.g., implementations that reuse state across identical subformulae.
Instead, we only require that there be a \emph{homomorphism} between the monitor's states and the states of some suitable monitor tree.
\begin{definition}\label{def:treelike}\leanok{Spec.lean}{274}{TreelikeMonitor}
  A~\emph{homomorphism} between a monitor $M^\phi =\Tuple{\States^\phi,\Init^\phi,\Step^\phi,\Out^\phi}$ and a monitor tree $T^\phi$ for $\phi$ is a function $f^\phi : Q^\phi \to Q^{T^\phi}$ satisfying
  $f^\phi(\Init^\phi) = \Init^{T^\phi}$,
  $f^\phi(\Step^\phi(q,x)) = \Step^{T^\phi}(f^\phi(q),x)$, and
  $\Out^\phi(q,x) = \Out^{T^\phi}(f^\phi(q),x)$
  for all $q \in Q^\phi$ and $x \in 2^\Props \times \Re$.
  The corresponding~\emph{treelike monitor} is the triple $\Tuple{\Monitor^\phi,T^\phi,f^\phi}$.
\end{definition}

A \emph{generic treelike monitor} is a triple of mappings $\Tuple{\Monitor,T,f}$ such that $\Monitor$ is a generic monitor
and for all $\phi$, the triple $\Tuple{\Monitor^\phi,T^\phi,f^\phi}$ is a treelike monitor. %

\begin{figure}[t]
  \centering
  \setlength{\belowcaptionskip}{-12pt}
  \scalebox{.9}{
  \begin{tikzpicture}
    \draw[draw=none,fill=black!10] (-1.5,-2.8) rectangle (1.5,0.5);
    \node[anchor=north west] at (-1.5,0.5)  {$M^{\phi_A}$};
    \node (phi) at (0,0) {$t_{\neg}$};
    \node (since) at (0,-0.8) {$t_{\Since_{(30,\infty)}}$};
    \node (neg) at (-1,-1.6) {$t_{\neg}$};
    \node (notify) at (-1,-2.4) {$t_{\mathsf{notify}}$};
    \node (process) at (1,-1.6) {$t_{\mathsf{process}}$};
    \draw[<-] (phi) to node[left] {\scriptsize $b$} (since);
    \draw[<-] (since) to node[left] {\scriptsize $b_l$} (neg);
    \draw[<-] (neg) to node[left] {\scriptsize $b$} (notify);
    \draw[<-] (since) to node[right] {\scriptsize $b_r$} (process);
  \end{tikzpicture}}
  \scalebox{.9}{
  \begin{tikzpicture}
    \draw[draw=none,fill=black!10] (-1.5,-2.8) rectangle (1.5,0.5);
    \node[anchor=north west] at (-1.5,0.5)  {$M^{\phi_B}$};
    \node (phi) at (0,0) {$t_{\neg}$};
    \node (since) at (0,-0.8) {$t_{\Since_{(15,\infty)}}$};
    \node (neg) at (-1,-1.6) {$t_{\neg}$};
    \node (notify) at (-1,-2.4) {$t_{\mathsf{notify}}$};
    \node (process) at (1,-1.6) {$t_{\mathsf{process}}$};
    \draw[<-] (phi) to node[left] {\scriptsize $b$} (since);
    \draw[<-] (since) to node[left] {\scriptsize $b_l$} (neg);
    \draw[<-] (neg) to node[left] {\scriptsize $b$} (notify);
    \draw[<-] (since) to node[right] {\scriptsize $b_r$} (process);
  \end{tikzpicture}}
  \scalebox{.9}{
  \begin{tikzpicture}
  \draw[draw=none,fill=black!10] (-1.5,-2.8) rectangle (1.5,0.5);
    \node[anchor=north west] at (-1.5,0.5)  {$M^{\phi_C}$};
    \node (phi) at (0,0) {$t_{\neg}$};
    \node (since) at (0,-0.8) {$t_{\Since_{(30,\infty)}}$};
    \node (neg) at (-1,-1.6) {$t_{\neg}$};
    \node (notify) at (-1,-2.4) {$t_{\mathsf{notify}}$};
    \node (process) at (1,-1.6) {$t_{\mathsf{notify}}$};
    \draw[<-] (phi) to node[left] {\scriptsize $b$} (since);
    \draw[<-] (since) to node[left] {\scriptsize $b_l$} (neg);
    \draw[<-] (neg) to node[left] {\scriptsize $b$} (notify);
    \draw[<-] (since) to node[right] {\scriptsize $b_r$} (process);
  \end{tikzpicture}}
  \vspace{1mm}

  \resizebox{\textwidth}{!}{\begin{tikzpicture}
    \node[anchor=west] at (1.2,0) {\scriptsize $t_{\neg} = \langle \lbrace \bot\rbrace, \bot, (\lambda (q,x).~q), (\lambda (q,(\ev,\tmsdiff)).~X \notin \ev) \rangle $};
    \node[anchor=west] at (7.8,0) {\scriptsize$t_{\Since_{(x,\infty)}} = \langle \Re \times (\Re \cup \{\bot\}), (0, \bot), \Step, \Out\rangle$, where};
    \node[anchor=west] at (8.0,-0.35) {\scriptsize$\Step((c,w),(\ev,\tmsdiff)) = (c+\tmsdiff,~\mathsf{upd}(w,c+\tmsdiff,X \in \ev,Y\in \ev))$};
    \node[anchor=west] at (8.0,-0.7) {\scriptsize$\Out((c,w),(\ev,\tmsdiff)) = \bigl(w' \neq \bot \land c + \tmsdiff - w' > x\bigr)$, $w' = \mathsf{upd}(w,c+\tmsdiff,X \in \ev,Y\in \ev)$};
    \node[anchor=west] at (8.0,-1.45) {\scriptsize$\mathsf{upd}(w,\tau,b_l,b_r) = \begin{cases} w & \text{if } w \neq \bot \land b_l = 1 \\ \tau & \text{otherwise, if } b_r = 1 \\ \bot & \text{otherwise} \end{cases}$};
    \node[anchor=west] at (1.2,-0.6) {\scriptsize$\forall \mathsf{e} \in \{ \mathsf{notify}, \mathsf{process} \}.$};
    \node[anchor=west] at (1.5,-0.9) {\scriptsize$t_{\mathsf{e}} = \langle \lbrace \bot\rbrace, \bot, (\lambda (q,x).~q), (\lambda (q,(\ev,\tmsdiff)).~\mathsf{e} \in \ev) \rangle$};
  \end{tikzpicture}}
  \vspace{-6mm}
  \caption{Monitor trees $M^{\phi_A}$, $M^{\phi_B}$, and $M^{\phi_C}$\label{fig:tree}}
  \vspace{-2mm}
\end{figure}

\section{Policy Change}
\label{sec:change}

The \emph{policy change problem} for a generic monitor $\Monitor$ consists of transforming the state of the monitor after monitoring a policy $\phi_1$ on an arbitrary trace to a state for a different policy $\phi_2$.
The new state must be such that all future verdicts are the same as if the monitor had been checking $\phi_2$ since the beginning of the trace.
In this section, we formally define the associated decision problem, which asks whether it is always possible to change from $\phi_1$ to $\phi_2$.
We first outline the problem in a general setting (Section~\ref{sec:change approach}), and then specialize the problem statement
to the case of treelike monitors (Section~\ref{sec:treelike pc}).
\looseness=-1

\subsection{First Approach}
\label{sec:change approach}

For any monitor $\Monitor^\phi$, define \emph{bisimilarity} ($\bisim_{\Monitor^\phi}$) between two states of $\Monitor^\phi$ as
\begin{align*}
  q \bisim_{\Monitor^\phi} q' \iff \forall s \in \Netraces.\; \Out^\phi(q,s) = \Out^\phi(q',s).
\end{align*}
Observe that $\bisim_{\Monitor^\phi}$ is an equivalence relation.
Intuitively, bisimilar states are indistinguishable under any possible future input to the monitor.
We now say that an algorithm $\Change: Q^{\phi_1} \rightarrow Q^{\phi_2}$ implements policy change from $\phi_1$ to $\phi_2$ iff
\begin{equation}\label{eq:pc1}
  q \bisim_{\Monitor^{\phi_1}} \Step^{\phi_1}(\Init^{\phi_1},s) \implies \Change(q) \bisim_{\Monitor^{\phi_2}} \Step^{\phi_2}(\Init^{\phi_2},s)
\end{equation}
for all finite traces $s$ and states $q \in \States^{\phi_1}$.
Suppose that we started monitoring $\phi_1$ on the trace prefix $s_1$ using monitor $\Monitor^{\phi_1}$.
Its state at the end of the prefix is $q = \Step^{\phi_1}(\Init^{\phi_1},s_1)$, which trivially satisfies the precondition of \eqref{eq:pc1}.
We then compute $\Change(q)$ and replace the running monitor with $\Monitor^{\phi_2}$, using $\Change(q)$ as its initial state.
Monitoring some continuation $s_2$ of the trace yields
\[
  \Out^{\phi_2}(\Change(q),s_2) \stackrel{\text{(\ref{eq:pc1}), def. $\bisim$}}{=} \Out^{\phi_2}(\Step^{\phi_2}(\Init^{\phi_2},s_1),s_2) \stackrel{\text{Prop.~\ref{prop:monitor:simps}}}{=} \Out^{\phi_2}(\Init^{\phi_2},s_1\cdot s_2),
\]
which is the output we would have obtained by running $\Monitor^{\phi_1}$ from the trace's beginning.
The first equality follows from \eqref{eq:pc1} and the definition of bisimilarity.

One may wonder whether the condition on $\Change$ could be strengthened by requiring that $q$ be equal, rather than just bisimilar, to $\Step^{\phi_1}(\Init^{\phi_1},s_1)$, as in our example.
However, this would prevent us from performing multiple policy changes on the same trace.
E.g., the state $q' = \Step^{\phi_2}(\Change(q),s_2)$ obtained after monitoring the trace's continuation $s_2$ generally only satisfies $q' \bisim_{\Monitor^{\phi_2}} \Step^{\phi_2}(\Init^{\phi_2},s_1 \cdot s_2)$, but the states are not necessarily equal.
Hence, we could not perform a further policy change, say to $\phi_3$.
In contrast, our definition allows us to continue with the state $\Change'(q')$ whenever $\Change'$ implements policy change from $\phi_2$ to $\phi_3$.

The above definition is still limited as it relies on the syntactic structure of the monitor rather than its semantics. For instance,
let $\phi_1 \equiv (P \Since Q) \land \lnot \Once R$ and $\phi_2 \equiv P \Since Q$.
We can construct $\Monitor^{\phi_1}$ such that it enters and remains in a single unique state whenever it encounters an event satisfying $R$, since $\phi_1$ is necessarily violated from that point on.
Yet there are two classes of traces, characterized by the truth value of $\phi_2$ at the last position, that are distinguishable by $\Monitor^{\phi_2}$.
Hence, policy change from $\phi_1$ to $\phi_2$ is not possible according to our current definition.

\subsection{Policy Change for Treelike Monitors}
\label{sec:treelike pc}

If it is treelike, the monitor $\Monitor^{\phi_1}$ from our previous example must maintain a correct \emph{substate}
for the subformula $P \Since Q$. Using this substate, policy change to $\phi_2$ is possible.
Next, we generalize this observation by introducing a notion of deep bisimilarity between monitor trees and, by extension,
between treelike monitors.
Intuitively, deep bisimilarity is bisimilarity for each pair of matching substates.

\begin{definition}\leanok{Spec.lean}{255}{MonTree.DeepBisim}
  Fix a treelike monitor $\Tuple{\Monitor^\phi,T^\phi,f^\phi}$.
  Two states $q,q' \in \States^{T^\phi}$ are \emph{deeply bisimilar}, denoted $q \deepbisim_{T^\phi} q'$, iff for each vertex $v$ in the syntax tree of $\phi$, the substates $q(v)$ and $q'(v)$ are bisimilar under $\bisim_{t_v}$, i.e., as states of the temporal tester $t_v$.%
  Two states $q,q' \in \States^\phi$ are \emph{deeply bisimilar under $f^\phi$}, denoted $q \deepbisim_{f^\phi} q'$, iff $f^\phi(q) \deepbisim_{T^\phi} f^\phi(q')$.
\end{definition}

Bisimilarity between substates is defined with respect to the tester alone, quantifying over \emph{all} future truth values of the operands $X$ and $Y$.
Deep bisimilarity is in general finer than the Nerode congruence of the language of the subformula $\phi_v$.
For instance, for $\phi = (\lnot\True) \Since Q$, which is equivalent to $Q$, the future truth values of $\phi$ do not depend on the past, yet the substate of the tester for $X \Since Y$ must record whether $Q$ held at the last time-point since this determines the output on a future event satisfying $X$.

Furthermore, we note that policy change for treelike monitors depends on the syntax of the formulae, not only on their semantics.
For example, the substates of a monitor for $\Once P \lor \Once Q$ contain the truth value of $\Once P$, so policy change from $\Once P \lor \Once Q$ to $\Once P$ is possible (Section~\ref{sec:continuous}); in contrast, the equivalent formula $\Once (P \lor Q)$ only maintains the truth value of $\Once(P \lor Q)$, hence policy change from $\Once(P\lor Q)$ to $\Once P$ is impossible.
The treelike assumption captures what information a syntax-directed monitor must maintain.

In general, vertices implementing the non-temporal operators $\{\True, \lnot X, X \lor Y\} \cup \Props$ need not be stateful.
Therefore, assuming an efficient implementation of the monitors, the relation $\deepbisim_{f^\phi}$ essentially enforces that the substates corresponding to $\phi$'s \emph{temporal subformulae} be bisimilar.
We now redefine policy change in terms of deep bisimilarity. In our new definition, the policy change algorithm has access to possibly more information (from substates), but it is also required to output a deeply bisimilar state, which is a stronger requirement than in Equation~\eqref{eq:pc1}.

\begin{definition}\label{def:pc}\leanok{Spec.lean}{287}{ImplementsPC}
  Let $\Tuple{\Monitor,T,f}$ be a generic treelike monitor, where $\Monitor = \Tuple{\States,\Init,\Step,\Out}$.
  A~function $\Change : Q^{\phi_1} \rightarrow Q^{\phi_2}$ implements \emph{policy change} from $\phi_1$ to $\phi_2$ iff
  \begin{equation*}\label{eq:pc2}
    q \deepbisim_{f^{\phi_1}} \Step^{\phi_1}(\Init^{\phi_1},s) \implies \Change(q) \deepbisim_{f^{\phi_2}} \Step^{\phi_2}(\Init^{\phi_2},s)
  \end{equation*}
  for all finite traces $s$ and states $q \in \States^{\phi_1}$.
\end{definition}

\begin{example}
  Policy change is possible between $M^{\phi_A}$ and $M^{\phi_B}$ from Figure~\ref{fig:tree}:
  the only temporal tester with non-trivial state in both $M^{\phi_A}$ and $M^{\phi_B}$
  is $t_{\Since_{(x,\infty)}}$, whose state does not depend on $x$. Hence, the identity function
  implements policy change from $\phi_A$ to $\phi_B$ for any $M$ with $M^{\phi_A}$ and $M^{\phi_B}$
  as in Figure~\ref{fig:tree}.
\end{example}

\looseness=-1
In the following, our focus will be on the associated decision problem:
\begin{definition}[\textsc{PolicyChange}]~\leanok{Spec.lean}{295}{PolicyChange}
  
  \textbf{Input}: a pair of pMTL formulae $\Tuple{\phi_1,\phi_2}$.
  
  \textbf{Output}: 1 if, for every generic treelike monitor $M$, there exists a function $C$
implementing policy change from $\phi_1$ to $\phi_2$, and 0 otherwise.
\end{definition}

\begin{example}
  \label{ex:impossible}
  The output of \textsc{PolicyChange} on the input $\langle \phi_A, \phi_C\rangle$ is 0.
  Consider a generic treelike monitor whose monitor trees for $\phi_A$ and $\phi_C$ are $M^{\phi_A}$ and $M^{\phi_C}$ in Figure~\ref{fig:tree} (with $f$ the identity), and the trace prefixes
  \begin{gather*}
    s_1 = (\lbrace \mathsf{notify}\rbrace, 0)\,(\emptyset, 30)\qquad
    s_2 = (\emptyset, 0)\,(\emptyset,30),
  \end{gather*}
  written as sequences of events $(\ev_i,\tmsdiff_i)$.
  The final state of the monitor tree $M^{\phi_A}$ is the same when reading $s_1$ or $s_2$:
  the state of the tester for $\Since_{(30,\infty)}$ is $(30,\bot)$, as $\mathsf{process}$ never occurs, and all other testers are stateless. Any function $C$ implementing
  policy change from $\phi_A$ to $\phi_C$ would thus map these two states to the same state $q$ of $M^{\phi_C}$, which would have to be deeply bisimilar to both states reached by $M^{\phi_C}$ on $s_1$ and on $s_2$.
  Let $x = (\emptyset,1)$, so that the extended traces end at timestamp $31$.
  In $s_1 \cdot x$, the $\mathsf{notify}$ event at time $0$ lies more than 30 time units in the past and has not been followed by another $\mathsf{notify}$; hence $\Tuple{s_1 \cdot x, 2} \not\models \phi_C$.
  In contrast, $\Tuple{s_2 \cdot x, 2} \models \phi_C$.
  Hence the root outputs of $M^{\phi_C}$ on $x$ from the two states differ, so the two states are not bisimilar, and no state $q$ can be bisimilar to both, a contradiction.
\end{example}

\section{Deciding Policy Change with Real-Time Semantics}

\label{sec:continuous}

In this section, we decide \textsc{PolicyChange} in three steps.
First, given the old policy $\phi_1$, we compute a formula $\Extract{\phi_1}$, expressed in a suitable \emph{hyperlogic}, that encodes the information extractable from any monitor state.
Similarly, we compute a formula $\Build{\phi_2}$ encoding the information needed to assemble the new state.
Then, policy change is possible iff $\Extract{\phi_1}$  entails $\Build{\phi_2}$, i.e., iff $\Extract{\phi_1} \Imp \Build{\phi_2}$ is valid. We use a decision procedure for the hyperlogic to decide whether $\Extract{\phi_1} \Imp \Build{\phi_2}$ is valid.

This section is organized as follows. We first develop a general framework for reducing policy change to hyperlogic entailment
(Section~\ref{sec:entailment}). Then, we specialize our framework to the case of pMTL, introducing a new hyperlogic called HyperpTPTL
(Section~\ref{sec:information}). Then, we show how entailment in this hyperlogic can be reduced to 1-TPTL satisfiability to decide \textsc{PolicyChange} (Section~\ref{sec:algorithm}).

\subsection{Policy Change as Hyperlogic Entailment}
\label{sec:entailment}

\begin{figure}[t]
  \centering
  \setlength{\belowcaptionskip}{-8pt}
  \begin{subfigure}{0.48\textwidth}
    \centering
    \scalebox{.9}{
    \begin{tikzpicture}[scale=0.8]
      \draw[red, fill=red!10] (-1,0.5) circle (1cm);
      \draw[red, fill=red!10] (2.2,0) circle (1.8cm);

      \draw[blue, fill=blue!20, opacity=0.5] (-1,0.5) circle (0.8cm);
      \draw[blue, fill=blue!20, opacity=0.5] (1.5,-0.5) circle (0.7cm);
      \draw[blue, fill=blue!20, opacity=0.5] (2.5,0.7) circle (0.6cm);
      \draw[blue, fill=blue!20, opacity=0.5] (2.8,-0.5) circle (0.4cm);

      \fill (-1.5,0.5) circle (2pt);
      \fill (-0.8,0.8) circle (2pt);
      \fill (-1.2,0.2) circle (2pt);
      \fill (1.5,-0.3) circle (2pt);
      \fill (1.7,-0.7) circle (2pt);
      \fill (2.3,0.9) circle (2pt);
      \fill (2.7,0.5) circle (2pt);
      \fill (2.8,-0.3) circle (2pt);
    \end{tikzpicture}}
    \caption{Policy change $\langle {\color{blue}\phi_1},{\color{red}\phi_2}\rangle$ is possible}
  \end{subfigure}
  \begin{subfigure}{0.48\textwidth}
    \centering
    \scalebox{.9}{
    \begin{tikzpicture}[scale=0.8]
      \draw[red, fill=red!10] (-1,0.5) circle (1cm);
      \draw[red, fill=red!10] (2.2,0) circle (1.8cm);
      \draw[red, fill=red!10] (0.05,-0.7) circle (0.4cm);

      \draw[blue, fill=blue!20, opacity=0.5] (-1,0.5) circle (0.8cm);
      \draw[blue, fill=blue!20, opacity=0.5] (0.5,-0.5) circle (0.7cm);
      \draw[blue, fill=blue!20, opacity=0.5] (2.5,0.7) circle (0.6cm);
      \draw[blue, fill=blue!20, opacity=0.5] (2.8,-0.5) circle (0.4cm);

      \fill (-1.5,0.5) circle (2pt);
      \fill (-0.8,0.8) circle (2pt);
      \fill (-1.2,0.2) circle (2pt);
      \fill (0.6,-0.3) circle (2pt);
      \node[anchor=west] at (0.6,-0.3) {$s_1$};
      \node[anchor=west] at (0.3,-0.7) {$s_2$};
      \fill (0.3,-0.7) circle (2pt);
      \fill (2.3,0.9) circle (2pt);
      \fill (2.7,0.5) circle (2pt);
      \fill (2.8,-0.3) circle (2pt);
    \end{tikzpicture}}
    \caption{Policy change $\langle {\color{blue}\phi_1},{\color{red}\phi_2}\rangle$ is impossible}
  \end{subfigure}
  \caption{Configurations of equivalence classes of trace prefixes under $\deepbisim_{{\color{blue}\phi_1}}$ and $\deepbisim_{{\color{red}\phi_2}}$\label{fig:potatoes}}
\end{figure}

We begin by lifting the deep bisimilarity relation to finite traces.
\begin{definition}\leanok{Spec.lean}{266}{MonTree.Indist}
  Let $T^\phi$ be a monitor tree. 
  Two traces $s,s' \in \Fintraces$ are \emph{indistinguishable by $T^\phi$}, written $s \deepbisim_{T^\phi} s'$, iff
  $\Step^{T^\phi}(\Init^{T^\phi},s) \deepbisim_{T^\phi} \Step^{T^\phi}(\Init^{T^\phi},s')$.
\end{definition}
This is again an equivalence relation.
Moreover, it does not depend on the implementation of $T^\phi$ (i.e., the transducers' states and their transition functions), but only on the tree structure,
which we have fixed for every pMTL formula.
\begin{lemma}\label{lemma:deepbisim:indep}\leanok{Main.lean}{17}{main_deep_indep}
  For every pair $T^\phi,U^\phi$ of monitor trees for $\phi$, $s \deepbisim_{T^\phi} s'$ iff $s \deepbisim_{U^\phi} s'$.
\end{lemma}
Therefore, we can simply write $s \deepbisim_\phi s'$ without specifying a concrete monitor tree.
Let $[s]_{\phi}$ denote the equivalence class under $\deepbisim_{\phi}$ that contains $s$.
These equivalence classes are the first step towards characterizing the possibility of policy change:

\begin{lemma}\label{lemma:pc:iff}\leanok{Main.lean}{20}{main_pc_iff}
	For every generic treelike monitor $\Tuple{\Monitor,T,f}$, there exists a function implementing policy change from $\phi_1$ to $\phi_2$ iff $[s]_{{\phi_1}} \subseteq [s]_{{\phi_2}}$ for all finite traces $s$.
\end{lemma}

\looseness=-1
The condition $\forall s.~[s]_{\phi_1} \subseteq [s]_{\phi_2}$ from the above lemma and its negation $\exists s.~[s]_{\phi_1} \not\subseteq [s]_{\phi_2}$ are illustrated in Figure~\ref{fig:potatoes}, where each dot represents a trace prefix. Policy change is impossible iff, as in Example~\ref{ex:impossible}, one can find prefixes $s_1 \deepbisim_{\phi_1} s_2$ such that $s_1 \not\deepbisim_{\phi_2} s_2$.

The set of all classes of trace prefixes which are indistinguishable with respect to a given
formula is a \emph{hyperproperty}~\cite{journals/jcs/ClarksonS10}, i.e., a set of sets of traces.
To model these classes, we construct \emph{information formulae} $\Extract{\phi_1}$ and $\Build{\phi_2}$, in Section~\ref{sec:information} using an appropriate hyperlogic.
For reasons that will become apparent there, we assume that the hyperlogic semantics is only defined for sets of \emph{aligned} traces, i.e., the last absolute timestamp in all non-empty traces is the same.
The general idea behind the information formulae is that $\Extract{\phi_1}$ (over-)approximates the equivalence classes $[s]_{\phi_1}$.
Similarly, $\Build{\phi_2}$ \mbox{(under-)}approximates the equivalence classes $[s]_{\phi_2}$.
Now, if the formula $\Extract{\phi_1} \Imp \Build{\phi_2}$ is valid, the condition of Lemma~\ref{lemma:pc:iff} follows, and policy change is possible.
Moreover, reducing the policy change problem to validity checking in the hyperlogic is complete if $\Extract{\phi} \equiv \Build{\phi}$ for all formulae.

\begin{lemma}\label{lemma:build:extract}\leanok{Main.lean}{26}{main_lemma_12}
  Suppose that for all finite traces $s$, non-empty sets of aligned finite traces $S$, and formulae $\phi$:
  \begin{compactenum}[(1)]
    \item $S \subseteq [s]_{\phi}$ implies $S \hypermodels \Extract{\phi}$.
    \item $s \in S$ and $S \hypermodels \Build{\phi}$ imply $S \subseteq [s]_{\phi}$.
    \item $s \deepbisim_\phi s^{+d}$ for every $d \in \Re$, where $s^{+d}$ denotes the copy of $s$ where the first delay has been incremented by $d$.
  \end{compactenum}
  Then the validity of $\Extract{\phi_1} \Imp \Build{\phi_2}$ implies $[s]_{\phi_1} \subseteq [s]_{\phi_2}$ for all $s$.
  Moreover, if $\Extract{\phi} \equiv \Build{\phi} =: \BE{\phi}$ for both $\phi \in \{\phi_1,\phi_2\}$, the converse holds, too.
\end{lemma}

Conditions~(1) and~(2) formalize the intuition behind the two information formulae.
A~set $S$ of prefixes satisfies $\Extract{\phi}$ whenever the monitor for $\phi$ cannot tell its elements apart, and hence $\Extract{\phi}$ provides a necessary condition for states to be indistinguishable.
Conversely, a set satisfying $\Build{\phi}$ only contains prefixes that the monitor for $\phi$ cannot tell apart: $\Build{\phi}$ gives a sufficient condition for indistinguishability.
Condition~(3) states that shifting all timestamps does not affect the monitor state; it lets us restrict attention to aligned sets of traces.

\subsection{Information Formulae for pMTL}

\label{sec:information}

\looseness=-1
Our construction of information formulae will require time constraints spanning across
multiple temporal operators.
Such bounds are difficult to express in logics like MTL, where the interval constraints are local to each operator.
Therefore, we construct a hyperlogic based on TPTL~\cite{journals/jacm/AlurH94}, which introduces a freeze quantifier for time variables.
Specifically, we augment the past-time variant of TPTL with trace quantification, similarly to how 
HyperLTL~\cite{conf/post/ClarksonFKMRS14},
HyperMTL~\cite{conf/isola/BonakdarpourDP18},
and HyperMITL~\cite{journals/ic/HoZJ21} extend the respective logics.
Our augmented logic, which we call Hyper past-Time Propositional Temporal Logic or \emph{HyperpTPTL}, is intractable.
However, we carve out a fragment for which satisfiability can be reduced to 1-TPTL, which is decidable over finite words
with non-primitive recursive complexity~\cite{book/Haase10}.
We find this fragment to be expressive enough to encode information formulae for pMTL.

\subsubsection{pTPTL and (1-)HyperpTPTL.}
\looseness=-1
Past-time TPTL (pTPTL) uses a single sort of variables representing time ($t$, $t'$, \dots).
HyperpTPTL adds a sort for trace variables ($\pi$, $\pi'$, \dots).
Since we use a point-based time model, traces may not be \emph{synchronized}, i.e., their timestamps may not agree at a given time-point; in particular, the equivalence classes $[s]_\phi$ that we aim to capture are typically unsynchronized.
Therefore, we construct HyperpTPTL using two levels.
The inner level embeds pTPTL with the usual pointwise semantics, whereas the outer level follows (Hyper)MITL, defining the semantics at the level of timestamps instead of time-points.
For example, consider the traces $s = (\{P\},0)\,(\emptyset,2)$ and $s' = (\{P\},1)\,(\emptyset,1)$, which both end at timestamp~2.
The inner formula $\Once_{[2,2]} P$ is evaluated on each trace separately:
here, it holds at the last time-point of $s$ but not of $s'$.
The outer formula $\HLift{P}{\pi} \HConj \HLift{P}{\pi'}$ is instead evaluated at a common \emph{instant} $u \in \Re$ shared by both traces. Therefore, it is false at every instant for the assignment $\pi \mapsto s$, $\pi' \mapsto s'$, since $s$ and $s'$ contain their $P$-events at different timestamps ($0$ and~$1$).
The outer temporal operators thus range over instants, which allows them to relate time-points of different traces.
The syntax of pTPTL formulae $\phi$ and HyperpTPTL formulae $\chi$ is given by the following grammar:%
\begin{align*}
  \phi ::={} &\True \mid P \mid \lnot \phi \mid \phi_1 \lor \phi_2 \mid \Prev \phi \mid \phi_1 \Since \phi_2 \mid {t.\;\phi} \mid {t_1 - t_2 \leq c}\\
  \chi ::={} &\HTrue \mid \HLift{\phi}{\pi} \mid \HNeg \chi \mid \chi_1 \HDisj \chi_2 \mid \chi_1 \HSince \chi_2 \mid \chi_1 \HSinceS \chi_2 \mid {t.\;\chi} \mid {t_1 - t_2 \leq c}\mid \texists \pi. \chi
\end{align*}
where $c \in \Int$, and $t_1$ and $t_2$ are time variables.
The freeze quantifier $t.\;\phi$ binds $t$ to the timestamp of the current time-point (or to the current instant, at the outer level), and constraints compare the values of two frozen variables.
  Time constraints relating a variable $t$ to the \emph{current} time are written by freezing the current time in a second variable around the comparison, as in $t'.\;(t - t' \leq c)$.
For both logics, we define $t_1 - t_2 \geq c \equiv t_2 - t_1 \leq -c$, $t_1 - t_2 \in [a,b] \equiv t_1 - t_2 \geq a \land t_1 - t_2 \leq b$ and similarly for other intervals, $\tforall \pi. \phi \equiv \HNeg (\texists \pi. \HNeg \phi)$,
as well as the other usual logical operators.
The \emph{strict since} outer operator $\HSinceS$ is similar to $\HSince$ but
  requires the right-hand side to occur at a strictly earlier timestamp and lifts the
  obligation for the left-hand side to hold at present.
The scope of the quantifiers $t.$ (the quantifier for time variables), $\texists$, and $\tforall$ extends to the right as far as possible.

pTPTL's semantics is defined by a relation $\Tuple{s,i,v} \models \phi$ where $v$ maps time variables to timestamps.
We only show the cases that differ from pMTL:
\begin{align*}
  \Tuple{s,i,v} &\models \Prev \phi &&\text{if $i > 0$  and $\Tuple{s,i-1,v} \models \phi$} \\
  \Tuple{s,i,v} &\models \phi_1 \Since \phi_2 &&\text{if there exists $j \leq i$ such that $\Tuple{s,j,v} \models \phi_2$ and} \\[-\jot]
      &&&\text{$\Tuple{s,k,v} \models \phi_1$ for all $k$ with $j < k \leq i$} \\
  \Tuple{s,i,v} &\models t.\;\phi &&\text{if $\Tuple{s,i,v(t \mapsto \tms_i)} \models \phi$}\\
  \Tuple{s,i,v} &\models t_1 - t_2 \leq c &&\text{if $v(t_1) - v(t_2) \leq c$}
\end{align*}
HyperpTPTL's semantics is defined by $\Tuple{S,\Pi,u,v} \hypermodels \phi$
where $S$ is a non-empty set of aligned traces, $\Pi$ maps trace variables to traces in $S$,
$u \in \Re$ is a timestamp, $v$ once again maps time variables to timestamps, and:
  \pagebreak[3]
\begin{alignat*}{2}\,
  \Tuple{S,\Pi,u,v} &\hypermodels \HTrue \\
  \Tuple{S,\Pi,u,v} &\hypermodels \HLift{\phi}{\pi} &&\text{if there exists a greatest $i$ such that $\tms_i = u$}\\[-\jot]
  &&&\text{where $\Pi(\pi) = (\ev,\tmsdiff)$, and $\Tuple{\Pi(\pi),i,v} \models \phi$} \\
  \Tuple{S,\Pi,u,v} &\hypermodels \HNeg\chi &&\text{if $\Tuple{S,\Pi,u,v} \not\hypermodels \chi$} \\
  \Tuple{S,\Pi,u,v} &\hypermodels \chi_1 \HDisj \chi_2 &&\text{if $\Tuple{S,\Pi,u,v} \hypermodels \chi_1$ or $\Tuple{S,\Pi,u,v} \hypermodels \chi_2$} \\
  \Tuple{S,\Pi,u,v} &\hypermodels \chi_1 \HSince \chi_2 &&\begin{aligned}[t]
    &\text{if $\exists u' \leq u$ such that $\Tuple{S,\Pi,u',v} \hypermodels \chi_2$, then}\\[-\jot]
    &\text{$\forall u''$ with $u' < u'' \leq u$ we have $\Tuple{S,\Pi,u'',v} \hypermodels \chi_1$}
  \end{aligned}\\
  \Tuple{S,\Pi,u,v} &\hypermodels \chi_1 \HSinceS \chi_2 &&\text{like $\chi_1 \HSince \chi_2$, but with $u' < u$ and $u' < u'' < u$} \\
  \Tuple{S,\Pi,u,v} &\hypermodels t.\;\chi &&\text{if $\Tuple{S,\Pi,u,v(t \mapsto u)} \hypermodels \chi$} \\
  \Tuple{S,\Pi,u,v} &\hypermodels t_1 - t_2 \leq c~ &\quad&\text{if $v(t_1) - v(t_2) \leq c$} \\
  \Tuple{S,\Pi,u,v} &\hypermodels \texists \pi. \chi &&\text{if $\exists s \in S$ such that $\Tuple{S,\Pi(\pi \mapsto s),u,v} \hypermodels \chi$}
\end{alignat*}

\looseness=-1
If $\chi$ is a closed formula (no free trace or time variables), we write $S \hypermodels \chi$ for $\Tuple{S,\emptyset,\tms(S),\emptyset} \hypermodels \chi$, where $\tms(S)$ is the unique last timestamp of the traces in $S$, or $0$ if all these traces are empty.
This is the reason why we assumed traces to be aligned: it provides us with a unique timestamp to anchor the interpretation of hyperformulae.
We say that $\chi$ is satisfiable iff there exists $S$ such that $S \hypermodels \chi$.
We denote by 1-pTPTL (1-HyperpTPTL, resp.) the syntactic fragment of pTPTL (HyperpTPTL, resp.) that uses only the time variables $t$ and $t'$ with the following restrictions:
\begin{compactenum}[(a)]
\item every occurrence of $t'$ lies in a subformula of the form $t'.\;(C \land \psi)$ (or $t'.\;C$), where $C$ is a Boolean combination of constraints of the form $t - t' \leq c$ or $t' - t \leq c$ and $\psi$ has no free time variables;
  \item every constraint occurs in such a $C$.
\end{compactenum}
Under these restrictions, $t'.\;(C\land\psi)$ is equivalent to $C^\ast \land \psi$, where $C^\ast$ is obtained from $C$ by reading each constraint $t - t' \leq c$ as ``the current timestamp minus the value of $t$ is at least $-c$''.
Hence 1-pTPTL is, up to this linear-size notation change, exactly the one-variable fragment of past TPTL considered in~\cite{book/Haase10}. %
We call \emph{HyperpLTL} the fragment of HyperpTPTL that uses neither freeze quantifiers nor constraints, at either level.
The following, recursively defined function translates pMTL formulae to equivalent 1-pTPTL formulae, using the standard idea that a metric constraint on an operator can be checked by freezing the time at which the operator is evaluated ($t$) and comparing it with the time at which the operand is witnessed ($t'$). With nested temporal operators, at most one time variable is free in each subformula, and hence the converted formula remains within 1-pTPTL:\begin{align*}
  \Embed{\True} &\equiv \True &
  \Embed{\phi_1 \lor \phi_2} &\equiv \Embed{\phi_1} \lor \Embed{\phi_2} \\
  \Embed{P} &\equiv P &
  \Embed{\Prev_I \phi} &\equiv t.\; \Prev ((t'.\; t - t' \in I) \land \Embed{\phi}) \\
  \Embed{\lnot \phi} &\equiv \lnot \Embed{\phi} &
  \Embed{\phi_1 \Since_I \phi_2} &\equiv t.\; \Embed{\phi_1} \Since ((t'.\; t - t' \in I) \land \Embed{\phi_2})
\end{align*}
\vspace{-4ex}
\begin{lemma}\label{lemma:pmtl to ptptl}\leanok{Hyper.lean}{48}{sat_emb}
  For all $\phi$, $s$, $i$, and $v$,
  $\Tuple{s,i,v} \models \Embed{\phi}$ iff $\Tuple{s,i} \models \phi$.
\end{lemma}

With this translation in place, we abuse notation and write $\HLift{\phi}{\pi} :\equiv \HLift{\Embed{\phi}}{\pi}$ for pMTL formulae $\phi$.
The operator $\HSince_I$ can be desugared in a similar fashion: $\chi_1 \HSince_I \chi_2 \equiv t.\;\chi_1 \HSince (t'.\;(t - t' \in I \land \chi_2))$, where $\chi_2$ is closed.

\newcommand*{\Red}{\rho}
\newcommand*{\TP}{\mathit{tp}}
\newcommand*{\Tick}{\mathit{ck}}
\subsubsection{Anchored Formulae}
The outer temporal operator $\HSince$ quantifies over all instants $u' \in \Re$, including instants at which none of the traces under consideration has a time-point.
In practice, it is often convenient to restrict ourselves to
hyperformulae in which the truth values of the outer $\HSince$ operators do not depend
on instants without time-points.
This will prove useful in defining a reduction from a fragment of HyperpTPTL to
1-TPTL.
For an assignment $\Pi$, let $T(\Pi)$ be the set of all timestamps of time-points of the traces $\Pi(\pi)$, $\pi \in \dom \Pi$.
We call a quantifier-free HyperpTPTL formula \emph{strict} if it is a lift $\HLift{\phi}{\pi}$, a conjunction with at least one strict conjunct, a disjunction of strict formulae, or of the form $t.\;\chi$ or $t'.\;(C \land \chi)$ with $\chi$ strict.
We call such a formula \emph{lax} if it is $\HTrue$, the negation of a strict formula, a conjunction of lax formulae, or a disjunction with at least one lax disjunct.
The key semantic property here is that a strict formula is false and a lax formula is true at every instant $u \notin T(\Pi)$ (for every $v$), because lifts require a time-point at the current instant to hold. Based on these, we define:
\begin{definition}\label{def:anchored}\leanok{Hyper.lean}{139}{Anchored}
  A~quantifier-free HyperpTPTL formula is \emph{anchored} if, for each of its subformulae $\chi_1 \HSince \chi_2$ (after desugaring $\HSince_I$) or $\chi_1 \HSinceS \chi_2$, $\chi_1$ is lax and $\chi_2$ is strict.
  A~HyperpTPTL formula in prenex normal form is anchored if its quantifier-free part is.
\end{definition}
For an anchored formula $\chi_1 \HSince \chi_2$ (or $\chi_1 \HSinceS \chi_2$) and any set of instants $U \supseteq T(\Pi)$, the definition of $\HSince$ may equivalently quantify over $u',u'' \in U$ only.
Anchoring is preserved by conjunction, disjunction, and trace variable substitution.

\subsubsection{Information Formulae} For pMTL, we can choose $\Extract{\phi}$ and $\Build{\phi}$ to be the same formula $\BE{\phi}$, which describes exactly the sets of prefixes that are pairwise indistinguishable by every tester in the monitor tree for $\phi$ (Lemma~\ref{lemma:information formulae}); by Lemma~\ref{lemma:build:extract}, the reduction to entailment is then sound and complete. We define the information formulae for pMTL recursively as follows, where $\pi$ and $\pi'$ are two arbitrary, distinct trace variables.
\begingroup
\allowdisplaybreaks
  \begin{align*}
    \BE{\True} &= \BE{P} = \True
    & \BE{(\Prev_{\emptyset}\phi)} &= \BE{\phi} \\
    \BE{(\lnot \phi)} &= \BE{\phi}
    & \BE{(\Prev_I\phi)} &= \BE{\phi} \HConj \tforall \pi,\pi'.~\{\phi\}_{\pi} \HImp \{\phi\}_{\pi'} \qquad \text{if }I \neq \emptyset\\
    \BE{(\phi_1 \lor \phi_2)} &= \BE{\phi_1} \HConj \BE{\phi_2}
    & \BE{(\phi_1 \Since_{\emptyset} \phi_2)} &= \BE{\phi_1} \HConj \BE{\phi_2}
  \end{align*}
  \vspace{-25pt}
  \begin{align*}
    \BE{(\phi_1 \Since_{I} \phi_2)} &= \begin{aligned}[t]&\BE{\phi_1} \HConj \BE{\phi_2} \HConj{}\tforall \pi,\pi'.\bm{\bigwedge}_{d \in I \cup [0,\inf I]} \left(\{\phi_1 \Since_{I-d} \phi_2\}_\pi \HImp \{\phi_1 \Since_{I-d} \phi_2\}_{\pi'}\right)\end{aligned}
  \end{align*}
\endgroup
Here we write $I - d$ for $\{x - d \mid x \in I, x \geq d\}$, which is again an interval.
Intuitively, the information formula for $\phi$ is a conjunction of formulae for each temporal operator in $\phi$, which correspond to stateful monitor tree vertices.
\pagebreak[3]
These formulae characterize the equivalence classes of traces that are indistinguishable by the corresponding transducer.
For example, for a subformula $\Prev_I \phi'$ where $I \neq \emptyset$, there are two classes of (non-empty) traces that are distinguished by the truth of $\phi'$ at the last time-point. The formula $\tforall \pi,\pi'.~\{\phi'\}_{\pi} \HImp \{\phi'\}_{\pi'}$ enforces that all traces in the set under consideration agree with respect to $\phi'$. The approach for operators of the form $\phi_1 \Since_I \phi_2$ is similar, except that there are more classes to consider.
We can distinguish them by hypothetically appending an event satisfying $\phi_1$.
Then $\phi_1 \Since_I \phi_2$ is satisfied at the new time-point iff $\phi_1 \Since_{I - d} \phi_2$ was satisfied at the last time-point of the original trace, where $d$ is the timestamp difference between the new and the previous time-point.
Accordingly, the information formula asserts that all traces must agree on $\phi_1 \Since_{I - d} \phi_2$ for all $d \in \Re$.
In practice, taking $d$ in $I \cup [0,\inf I]$ is enough, denoted by an \emph{infinitary conjunction},\footnote{As an abuse of notation, the semantics of $\HConj$ has been straightforwardly extended to (possibly infinite) families of HyperpTPTL formulae $(\chi_j)_{j \in J}$ by defining that $\Tuple{S,\Pi,u,v} \hypermodels \bm{\bigwedge}_{j \in J} \chi_j$ holds iff for all $j \in J$, $\Tuple{S,\Pi,u,v} \hypermodels \chi_j$ holds.} since other $d$ give rise to empty shifted intervals.

\begin{example}
  \label{ex:phi_A_phi_B}
  The information formulae for $\phi_A$, $\phi_B$, and $\phi_C$ are
  \begin{gather*}
    \BE{\phi_A}=\gamma_{30,\neg\mathsf{notify},\mathsf{process}}, \quad
    \BE{\phi_B}=\gamma_{15,\neg\mathsf{notify},\mathsf{process}}, \text{and} \quad
    \BE{\phi_C}=\gamma_{30,\neg\mathsf{notify},\mathsf{notify}},
  \end{gather*}
  respectively, where $\gamma_{b,\alpha,\beta} = \tforall \pi,\pi'. \bm{\bigwedge}_{d\in\Re} \left(\{\alpha \Since_{(b,\infty)-d} \beta\}_\pi \HImp \{\alpha \Since_{(b,\infty)-d} \beta\}_{\pi'}\right)$.
  Now, observe that
{\small
  \begin{align*}
    \gamma_{b,\alpha,\beta} \equiv \tforall \pi,\pi'. \left(\{\alpha \Since \beta\}_\pi \HIff \{\alpha \Since \beta\}_{\pi'}\right) \HConj \bm{\bigwedge}_{x \in [0,b]} \left(\{\alpha \Since_{(x,\infty)} \beta\}_\pi \HIff \{\alpha \Since_{(x,\infty)} \beta\}_{\pi'}\right)%
  \end{align*}}
  Hence,  $\gamma_{b,\alpha,\beta} \HImp \gamma_{b',\alpha,\beta}$ is a tautology whenever $b \geq b'$.
  Therefore, $\BE{\phi_A} \HImp \BE{\phi_B}$ is a tautology, and thus policy change between
  $\phi_A$ and $\phi_B$ is possible.

  Noting that $\{ \alpha \Since_{(x',\infty)} \beta\}_\pi$ implies
  $\{ \alpha \Since_{(x,\infty)} \beta\}_\pi$ when $x \leq x'$, 
   the expression of $\gamma_{b,\alpha,\beta}$ obtained above shows that a temporal tester
   for $\alpha \Since_{(b,\infty)} \beta$ needs to store \emph{at least} the following information:
   (i) the maximum $x \in [0,b]$ such that $\alpha \Since_{[x,x]} \beta$ holds (if any),
   (ii) the truth value of $\alpha \Since \beta$, and
   (iii) the truth value of $\alpha \Since_{(b,\infty)} \beta$. This is exactly
  what some state-of-the-art monitors such as Hydra~\cite{conf/atva/RaszykBKT19} store, and
  hence the `optimal' equivalence classes described by the information formulae are actually
  realized in practice.
\end{example}

\begin{lemma}\label{lemma:information formulae}\leanok{Main.lean}{36}{main_lemma_16}
  The 1-HyperpTPTL formulae $\BE{\phi}$ fulfill the conditions of Lemma~\ref{lemma:build:extract}.
\end{lemma}

Observe that for any formula $\varphi$ in pLTL, the information formula $\BE{\varphi}$ is in HyperpLTL. Namely, when $I = [0,\infty)$,
the formula for $\BE{(\phi_1 \Since_I \phi_2)}$ simplifies to
\begin{align*}
  \BE{(\phi_1 \Since\phi_2)} &= \begin{aligned}[t]&\BE{\phi_1} \HConj \BE{\phi_2} \HConj{} \tforall \pi,\pi'. \left(\{\phi_1 \Since \phi_2\}_\pi \HImp \{\phi_1 \Since \phi_2\}_{\pi'}\right),\end{aligned}
\end{align*}
which does not use any non-LTL constructs. Thus, we have the following corollary:

\begin{corollary}\label{cor:information formulae HyperpLTL}\leanok{Main.lean}{43}{main_lemma_17}
  Given any pLTL formula $\varphi$, the HyperpLTL formula of the form $\BE{\varphi}$ does not use any $\HSince$ operators (or infinitary conjunctions), and the size of $\BE{\varphi}$ is at most quadratic in $|\varphi|$.
\end{corollary}

The information formulae associated to pMTL formulae are, in general, infinite, due to the presence of infinite conjunctions in
$\BE{(\phi_1 \Since_I \phi_2)}$. However, we show:

\begin{lemma}\label{lemma:rewriting}\leanok{Size.lean}{230}{lemma_rewriting}
  For any non-empty interval $I$ and pMTL formulae $\phi_1$ and $\phi_2$, there exists an anchored 1-HyperpTPTL formula $\Phi_{\phi_1,\phi_2,I}$
  such that 
  \begin{align*}
    \Phi_{\phi_1,\phi_2,I} && \text{and} && \tforall \pi,\pi'.\bm{\bigwedge}_{d \in I \cup [0,\inf I]} \left(\{\phi_1 \Since_{I-d} \phi_2\}_\pi \HImp \{\phi_1 \Since_{I-d} \phi_2\}_{\pi'}\right)
  \end{align*}
  are semantically equivalent and $|\Phi_{\phi_1,\phi_2,I}| = O(|\phi_1| + |\phi_2| + |I|)$, where $|I|$ is the size of the binary encoding of $I$'s bounds.
\end{lemma}

\noindent The proof of this lemma is quite technical; we refer the reader to 
Appendix~\ref{sec:proof:lemma:rewriting} for details.

\looseness=-1
In the following, we denote by $\BE{\phi}$ the information formulae obtained by first applying
the rewriting from Lemma~\ref{lemma:rewriting} to the information formula described above, %
and then converting the rewritten formulae to prenex normal form, using the fact that $(\chi \HConj \forall \pi.\;\psi) \HIff (\forall \pi.\; \chi \HConj \psi)$, where $\pi$ is not free in $\chi$. By Lemma~\ref{lemma:rewriting}, these formulae are still in 1-HyperpTPTL and anchored.
Moreover, they are of the form $\tforall \pi, \pi'.~\psi$
where $|\psi|=O(|\phi|^2)$ and $\psi$ does not contain trace quantifiers, and are equivalent to those in
Lemma~\ref{lemma:information formulae}. The quadratic size of the rewritten information formulae
follows from the original definition of $\BE$ and the linear size of the
rewritten formulae. Namely, the original definition of $\BE$ generates $O(|\phi|)$ disjuncts
encoding the information needed to monitor each of the subformulae of $\phi$, each of size  $O(|\phi|)$.

\subsection{The Decision Procedure}
\label{sec:algorithm}

Our algorithm for \textsc{PolicyChange} is based on a reduction of satisfiability in the $\texists^*\tforall^*$ fragment of 1-HyperpTPTL
to 1-pTPTL satisfiability. We first describe this reduction, and then present the full decision procedure.

\subsubsection{Reducing ${\exists}^*{\forall}^*$-HyperpTPTL to pTPTL Satisfiability}

The satisfiability of general HyperLTL and HyperMITL is not decidable.
However, it has been shown that these problems are decidable when formulae do not contain any $\tforall\texists$ quantifier alternation~\cite{conf/concur/FinkbeinerH16,journals/ic/HoZJ21}.
In particular, for formulae in the decidable $\texists^*\tforall^*$ fragment, HyperMITL
satisfiability can be reduced to MITL satisfiability with exponential blow-up~\cite{conf/concur/FinkbeinerH16}.
More precisely, the blow-up factor is $O(n^m)$, where $n$ ($m$, resp.) is the number of $\texists$ ($\tforall$, resp.) quantifiers.
Hence, with a bounded number of trace quantifiers, the reduction will lead to a constant blow-up factor.

We adapt this reduction to the anchored fragment of HyperpTPTL.
The target of the reduction is pTPTL satisfiability.
In the following, let $\chi$ be an anchored HyperpTPTL formula of the form $\texists \pi_1,\dots,\pi_n.\; \tforall \pi'_1,\dots,\pi'_m.\; \chi'$, where $\chi'$ is free of trace quantifiers.
Without loss of generality, assume $n,m \geq 1$. Define
\[
  \tilde\chi \equiv \bigwedge_{i_1 \in \{1,\dots,n\}} \cdots \bigwedge_{i_m \in \{1,\dots,n\}} \chi'[\pi'_1 \mapsto \pi_{i_1}]\cdots[\pi'_m \mapsto \pi_{i_m}],
\]
where $\chi'[\pi' \mapsto \pi]$ denotes trace variable substitution.
The size of $\tilde\chi$ is $O(n^m \cdot |\chi'|)$.

While $\tilde\chi$ does not contain trace quantifiers, it is not a pTPTL formula yet.
To obtain the latter, we apply the following recursive translation $\Red$:
\begingroup
\allowdisplaybreaks
\begin{alignat*}{2}
  \Red(\HTrue) &\equiv \True
  & \Red(\chi_1 \HDisj \chi_2) &\equiv \Red(\chi_1) \lor \Red(\chi_2) \\
  \Red(\HNeg \chi) &\equiv \lnot \Red(\chi)
  & \Red(\HLift{\phi}{\pi}) &\equiv \Prev\bigl((\lnot \TP_\pi \land \lnot \Tick) \Since (\TP_\pi \land \Red_\pi(\phi))\bigr) \\
  \Red(t_1 - t_2 \leq c) &\equiv t_1 - t_2 \leq c \qquad
  & \Red(\chi_1 \HSince \chi_2) &\equiv (\Tick \Imp \Red(\chi_1)) \Since (\Tick \land \Red(\chi_2)) \\
  \Red(t.\;\chi) &\equiv t.\;\Red(\chi)\qquad
  & \Red(\chi_1 \HSinceS \chi_2) &\equiv \Prev\bigl((\Tick \Imp \Red(\chi_1)) \Since (\Tick \land \Red(\chi_2))\bigr) \\
  \Red_\pi(\True) &\equiv \True \qquad
  & \Red_\pi(P) &\equiv P_\pi \\
  \Red_\pi(\lnot \phi) &\equiv \lnot \Red_\pi(\phi)
  & \Red_\pi(\phi_1 \lor \phi_2) &\equiv \Red_\pi(\phi_1) \lor \Red_\pi(\phi_2) \\
  \Red_\pi(t_1 - t_2 \leq c) &\equiv t_1 - t_2 \leq c \qquad
  & \Red_\pi(\Prev \phi) &\equiv \Prev \left( \lnot \TP_\pi \Since (\TP_\pi \land \Red_\pi(\phi))\right) \\
  \Red_\pi(t.\;\phi) &\equiv t.\;\Red_\pi(\phi)
  & \Red_\pi(\phi_1 \Since \phi_2) &\equiv (\TP_\pi \Imp \Red_\pi(\phi_1)) \Since (\TP_\pi \land \Red_\pi(\phi_2))
\end{alignat*}
\endgroup
In the output of these translations, predicates with trace variable subscripts $P_\pi$ are treated symbolically as distinct propositions.
Moreover, we assume that the predicates $\TP_\pi$ and $\Tick$ are distinct from all other predicates.
Intuitively, the translated formula is interpreted over a single trace that interleaves a \emph{finite} set of aligned traces: those used to interpret the trace variables occurring in $\tilde\chi$.
The predicate $\TP_\pi$ marks time-points that exist in the trace for $\pi$, and $\Tick$ marks the last time-point in a sequence of time-points with the same timestamp.
The interleaving contains, for every timestamp occurring in one of the traces, a block of events with that timestamp, ending with a dedicated event that only carries $\Tick$.
The outer level is evaluated at these $\Tick$ events, which represent the instants in $T(\Pi)$.
A~lift $\HLift{\phi}{\pi}$ is evaluated by stepping back from the $\Tick$ event (hence the $\Prev$) to the last time-point of $\pi$ in the same block, if any.
Timestamps are arbitrary non-negative reals (or naturals, in the discrete-time setting); no discretization is involved.
Figure~\ref{fig:interleaving} illustrates this encoding.

\begin{figure}[t]
  \centering
  \tikzset{
    ev1/.style={draw=blue!70!black, fill=blue!10, rounded corners=1.5pt, minimum width=0.75cm, minimum height=0.8cm, inner sep=1pt, font=\scriptsize},
    ev2/.style={draw=orange!85!black, fill=orange!12, rounded corners=1.5pt, minimum width=0.75cm, minimum height=0.8cm, inner sep=1pt, font=\scriptsize},
    evc/.style={draw=green!45!black, fill=green!15, rounded corners=1.5pt, minimum width=0.6cm, minimum height=0.8cm, inner sep=1pt, font=\scriptsize},
    pt1/.style={fill=blue!70!black, circle, inner sep=1.6pt},
    pt2/.style={fill=orange!85!black, circle, inner sep=1.6pt}}
  \begin{tikzpicture}[x=1cm, y=1cm]
    \draw[->, gray!70!black] (0,3.2) node[left, font=\scriptsize, text=black]{$\Pi(\pi_1)$} -- (9,3.2);
    \draw[->, gray!70!black] (0,2.5) node[left, font=\scriptsize, text=black]{$\Pi(\pi_2)$} -- (9,2.5);
    \foreach \t/\l in {0.7/0, 3.7/1, 6.95/2.5} {
      \draw[dotted, gray] (\t,2.3) -- (\t,3.6);
      \node[above, font=\scriptsize] at (\t,3.55) {$\l$};
    }
    \node[pt1, label={[font=\tiny]above:$P$}] at (0.7,3.2) {};
    \node[pt1, label={[font=\tiny]above left:$Q$}] at (6.8,3.2) {};
    \node[pt1, label={[font=\tiny]above right:$P$}] at (7.1,3.2) {};
    \node[pt2, label={[font=\tiny]below:$Q$}] at (3.7,2.5) {};
    \node[pt2, label={[font=\tiny]below:$P$}] at (6.95,2.5) {};
    \node[ev1] (a) at (0.3,0.5) {\shortstack{$\TP_{\pi_1}$\\$P_{\pi_1}$}};
    \node[evc] (b) at (1.15,0.5) {$\Tick$};
    \node[ev2] (c) at (3.3,0.5) {\shortstack{$\TP_{\pi_2}$\\$Q_{\pi_2}$}};
    \node[evc] (d) at (4.15,0.5) {$\Tick$};
    \node[ev1] (e) at (5.6,0.5) {\shortstack{$\TP_{\pi_1}$\\$Q_{\pi_1}$}};
    \node[ev1] (f) at (6.5,0.5) {\shortstack{$\TP_{\pi_1}$\\$P_{\pi_1}$}};
    \node[ev2] (g) at (7.4,0.5) {\shortstack{$\TP_{\pi_2}$\\$P_{\pi_2}$}};
    \node[evc] (h) at (8.25,0.5) {$\Tick$};
    \node[left, font=\scriptsize] at (-0.2,0.5) {$s$};
    \foreach \x/\y/\l in {-0.15/1.5/0, 2.85/4.5/1, 5.15/8.6/2.5}
      \draw[decorate, decoration={brace, mirror, amplitude=3pt}] (\x,-0.05) -- (\y,-0.05)
        node[midway, below=3pt, font=\scriptsize]{block of timestamp $\l$};
    \draw[->, thick, blue!60!black] (h.north) to[bend right=40] (f.north);
    \node[font=\scriptsize, text=blue!60!black, align=left, anchor=west] at (8.0,1.55)
      {$\HLift{\phi}{\pi_1}$: last $\TP_{\pi_1}$\\of the block};
  \end{tikzpicture}
  \caption{Interleaving of two aligned traces (top) into a single trace $s$ (bottom).
  Each timestamp yields a block of events of $s$ containing the time-points of both traces (marked by $\TP_{\pi_1}$, $\TP_{\pi_2}$, with predicates renamed to $P_{\pi_1}, P_{\pi_2}, \dots$) and ending with a $\Tick$ event.}
  \label{fig:interleaving}
\end{figure}

Finally, we characterize proper interleavings satisfying the above properties:
\[
  \zeta \equiv \Tick \land \HAlways\bigl(\Prev \Tick \Iff t.\;\Prev(t'.\;\lnot(t-t' \leq 0))\bigr) \land \HAlways\bigl(\Tick \Imp \textstyle\bigwedge_{k} \lnot \TP_{\pi_k}\bigr) \land
  \bigwedge_{i\in\{1,\dots,n\}} (\Once \TP_{\pi_i} \Imp \Once_{[0,0]} \TP_{\pi_i}).
\]
The second conjunct states that a time-point carries $\Tick$ iff the next time-point has a strictly larger timestamp, the third that $\Tick$ events carry no other marker, and the last conjunct ensures that all non-empty encoded traces end at the same timestamp, i.e., are aligned.

\begin{lemma}\label{lemma:HyperpTPTL}\leanok{Main.lean}{85}{main_lemma_19}
  Suppose that $\chi'$ is anchored.
  Then $\chi$ is satisfiable iff $\zeta \land \Red(\tilde\chi)$ is, where satisfiability is understood with respect to the same time domain (real or discrete) on both sides.
  The size of $\zeta \land \Red(\tilde\chi)$ is $O(n^m \cdot |\chi'|)$.
\end{lemma}

\begin{lemma}\label{lemma:1-HyperpTPTL}\leanok{Main.lean}{91}{main_lemma_20}
  If $\chi$ is in 1-HyperpTPTL, then $\zeta \land \Red(\tilde\chi)$ is in 1-pTPTL.
\end{lemma}

We now prove an analogue of Lemma~\ref{lemma:HyperpTPTL} for the HyperpLTL fragment. It follows from \Cref{cor:information formulae HyperpLTL} that the rewriting step described in Lemma~\ref{lemma:rewriting} is not necessary to ensure that $\BE{\phi}$ is of size polynomial in $|\phi|$.
Hence, for this special case, we can show that $\Red'(\tilde\chi)$ is again a pLTL formula, where $\Red'(\tilde\chi)$ is defined like $\Red(\tilde\chi)$ except that $\Red'(\HLift{\phi}{\pi}) \equiv \lnot \TP_\pi \Since (\TP_\pi \land \Red_\pi(\phi))$. Moreover, the additional constraints $\zeta$ are not required to check satisfiability.

\begin{corollary}\label{corollary:HyperpLTL sat}\leanok{Main.lean}{97}{main_corollary_21}
  Suppose that $\chi$ is a HyperpLTL formula without $\HSince$ operators.
  Then $\chi$ is satisfiable iff $\Red'(\tilde\chi)$, which is a pLTL formula of size $O(n^m \cdot |\chi'|)$, is satisfiable.
\end{corollary}

\subsubsection{Deciding Policy Change}

With this reduction, the non-primitive recursive algorithm by Haase et al. for deciding 1-TPTL
satisfiability over finite timed words~\cite{book/Haase10,journal/lmcs/ouaknine2007decidability}
can be used to check the satisfiability of 1-HyperpTPTL formulae in the $\texists^*\tforall^*$ fragment. Assume given a routine  $\sattptl$, that when given as input a 1-pTPTL formula $\phi$, returns $\True$ iff $\phi$ is satisfiable, and $\False$ otherwise. Our decision procedure is shown in Algorithm~\ref{alg:upper-bound}, where $\pi_1,\pi'_1,\pi_2,\pi'_2$ are four distinct trace variables.
\begin{algorithm}[H]
	\caption{Deciding policy change\label{alg:upper-bound}}
	\begin{algorithmic}
		\Function{decidePolicyChange}{$\phi_1,\phi_2$}
		\State $(\tforall \pi,\pi'.~\chi_1) \gets \BE{\phi_1}$; $(\tforall \pi,\pi'.~\chi_2) \gets \BE{\phi_2}$
		\State $\chi \gets (\texists \pi_2,\pi'_2.\tforall \pi_1,\pi'_1. \chi_1[\pi \mapsto \pi_1, \pi' \mapsto \pi'_1] \land \lnot \chi_2[\pi \mapsto \pi_2, \pi' \mapsto \pi'_2])$
		\State \Return $\neg \sattptl(\zeta \land \Red(\tilde{\chi}))$
		\EndFunction
  \end{algorithmic}
\end{algorithm}
Combining Lemmata~\ref{lemma:pc:iff}, \ref{lemma:build:extract}, \ref{lemma:information formulae}, \ref{lemma:rewriting}, \ref{lemma:HyperpTPTL} and \ref{lemma:1-HyperpTPTL}, we now obtain:
\begin{theorem}\label{thm:upper}\leanpart{Main.lean}{52}{main_upper_bound}
  Algorithm~\ref{alg:upper-bound} decides \textsc{PolicyChange} for real-time semantics with non-primitive recursive complexity; more precisely, \textsc{PolicyChange} is in $\mathbf{F}_{\omega^\omega}$.
\end{theorem}
\begin{proof}
  \emph{Correctness.} Combining Lemmata~\ref{lemma:pc:iff}, \ref{lemma:build:extract}, \ref{lemma:information formulae} and \ref{lemma:rewriting},
	we know that policy change from $\phi_1$ to $\phi_2$ is possible iff $\BE{\phi_1} \Imp \BE{\phi_2}$ is valid. This is true iff
	\begin{align*}
		\neg(\BE{\phi_1} \Imp \BE{\phi_2}) \equiv \neg((\tforall \pi,\pi'.~\chi_1) \Imp (\tforall \pi,\pi'.~\chi_2))
	\end{align*}
	is not satisfiable. It is easy to check that $\chi \equiv \neg((\tforall \pi,\pi'.~\chi_1) \Imp (\tforall \pi,\pi'.~\chi_2))$.
	Since $\chi$ is in the $\texists^*\tforall^*$ fragment of HyperpTPTL and anchored, by Lemma~\ref{lemma:HyperpTPTL},
	it is satisfiable iff $\zeta \land \Red(\tilde\chi)$ is satisfiable in pTPTL. Therefore, policy change from $\phi_1$ to $\phi_2$ is possible
	iff $\zeta \land \Red(\tilde\chi)$ is not satisfiable. 
	
	\emph{Complexity.} The formula $\zeta \land \Red(\tilde\chi)$ is computable in polynomial time and has size polynomial in $|\phi_1| + |\phi_2|$ (Lemmas~\ref{lemma:rewriting} and~\ref{lemma:HyperpTPTL}, with $n = m = 2$).
	Since $\chi$ is in 1-HyperpTPTL, using Lemma~\ref{lemma:1-HyperpTPTL} and~\cite{book/Haase10,journal/lmcs/ouaknine2007decidability}, satisfiability of $\zeta \land \Red(\tilde\chi)$ can be decided with non-primitive recursive complexity.
	These results are stated for future-time one-variable TPTL over finite timed words. A~past-time formula evaluated at the last position of a finite timed word $(\ev_i,\tms_i)_{i<k}$ is equivalent to its mirror-image future-time formula evaluated at the first position of the reversed word $(\ev_{k-1-i},\tms_{k-1}-\tms_{k-1-i})_{i<k}$, since TPTL constraints only depend on timestamp differences.
	The decision procedure of~\cite{book/Haase10,journal/lmcs/ouaknine2007decidability} goes through emptiness of one-clock alternating timed automata, which is in $\mathbf{F}_{\omega^\omega}$~\cite{journals/toct/Schmitz16}.
\end{proof}

We further obtain a lower bound for the complexity of \textsc{PolicyChange} via a reduction from pMTL satisfiability over finite timed words, which is known to have a non-primitive recursive lower bound \cite{journal/lmcs/ouaknine2007decidability}:

\begin{theorem}\label{theorem:pc pmtl real-time lb}\leanpart{Main.lean}{68}{main_lower_bound}
  \textsc{PolicyChange} for pMTL over real-time semantics has a non-primitive recursive lower bound; more precisely, it is $\mathbf{F}_{\omega^\omega}$-hard.
\end{theorem}

Theorems~\ref{thm:upper} and~\ref{theorem:pc pmtl real-time lb} thus show
that \textsc{PolicyChange} for pMTL over real-time semantics is
$\mathbf{F}_{\omega^\omega}$-complete, with hardness under $\mathbf{F}_{<\omega^\omega}$
reductions as usual for this class.

\section{Deciding Policy Change with Discrete-Time Semantics}\label{sec:discrete}

We now consider the \textsc{PolicyChange} problem over discrete-time semantics.
We show that \textsc{PolicyChange} is EXPSPACE-complete for pMTL and PSPACE-complete for pLTL.

A~tentative approach would be to translate pMTL into pLTL and apply a decision procedure for pLTL.
However, policy change for treelike monitors depends on the syntax tree of the formulae (Section~\ref{sec:treelike pc}) and such translations introduce new subformulae,
which changes the information stored by the monitor.
Instead, we reuse the proof schemas of Section~\ref{sec:continuous}; details are
given in Appendix~\ref{apx:discrete}.

Regarding the reduction targets, satisfiability of TPTL over discrete time is EXPSPACE-complete with binary-encoded constants~\cite{journals/jacm/AlurH94}, while satisfiability of LTL over finite traces is PSPACE-complete~\cite{conf/ijcai/DeGiacomo13}. Both results are stated for future-time logics but transfer to past-time logics over finite words by the time-reversal argument of Theorem~\ref{thm:upper}.
For pLTL, the information formulae need neither $\HSince$ nor the rewriting step (Lemma~\ref{cor:information formulae HyperpLTL}), and Corollary~\ref{corollary:HyperpLTL sat} yields a pLTL formula of polynomial size.
Thus, we obtain:

\begin{theorem}\label{thm:pc ub discrete}\leanpart{Main.lean}{60}{main_upper_bound_pltl}\leanpart{Main.lean}{52}{main_upper_bound}
  The following statements hold:
  \begin{compactenum}[(i)]
    \item \textsc{PolicyChange} for pLTL is in PSPACE.
    \item \textsc{PolicyChange} for pMTL with discrete-time semantics is in EXPSPACE.
  \end{compactenum}
\end{theorem}

Satisfiability of MTL over discrete time is EXPSPACE-hard~\cite{journals/iandc/AlurH93,laroussinie2006efficient}; the hardness proof encodes the (finite) accepting runs of exponential-space Turing machines and therefore also applies to finite words.
Satisfiability of pLTL over finite traces is PSPACE-hard~\cite{conf/ijcai/DeGiacomo13}. We get:

\begin{theorem}\label{thm:pc lb discrete}\leanpart{Main.lean}{68}{main_lower_bound}
  The following statements hold:
  \begin{compactenum}[(i)]
    \item \textsc{PolicyChange} for pLTL is PSPACE-hard.
    \item \textsc{PolicyChange} for pMTL with discrete-time semantics is EXPSPACE-hard.
  \end{compactenum}
\end{theorem}
Full proofs of Theorems~\ref{thm:pc ub discrete} and~\ref{thm:pc lb discrete} are given in Appendix~\ref{apx:discrete}.

\section{Policy Change for Non-Treelike Monitors}\label{sec:general}

In this section, we study the following generalization of the \textsc{PolicyChange} problem,
which lifts the treelike assumption on monitors and goes back to the definition of Section~\ref{sec:change approach}:
\begin{definition}\label{def:pcstar}\leanok{Spec.lean}{304}{ImplementsPCStar}
  Let $\Monitor = \Tuple{\States,\Init,\Step,\Out}$ be a generic monitor.
  A~function $\Change : Q^{\phi_1} \rightarrow Q^{\phi_2}$ implements \emph{policy change} from $\phi_1$ to $\phi_2$ iff
  \begin{equation*}
    q \bisim_{\Monitor^{\phi_1}} \Step^{\phi_1}(\Init^{\phi_1},s) \implies \Change(q) \bisim_{\Monitor^{\phi_2}} \Step^{\phi_2}(\Init^{\phi_2},s)
  \end{equation*}
  for all finite traces $s$ and states $q \in \States^{\phi_1}$.
\end{definition}

\begin{definition}[\textsc{PolicyChange}*]~\leanok{Spec.lean}{311}{PolicyChangeStar}
  
  \textbf{Input}: a pair of pMTL formulae $\Tuple{\phi_1,\phi_2}$.
  
  \textbf{Output}: 1 if, for every generic monitor $M$ (not necessarily treelike), there exists a function $C$
  implementing policy change from $\phi_1$ to $\phi_2$, and 0 otherwise.
\end{definition}

For any pMTL formula $\phi$, let ${\sim_\phi}$ be the \emph{Nerode right congruence} of $\Sem{\phi}$: $s \sim_\phi s'$ iff $s \cdot s'' \in \Sem{\phi} \iff s' \cdot s'' \in \Sem{\phi}$ for all $s'' \in \Netraces$.
For every correct monitor $\Monitor^\phi$, the states reached on $s$ and $s'$ are bisimilar iff $s \sim_\phi s'$, by Proposition~\ref{prop:monitor:simps}.
The proof of Lemma~\ref{lemma:pc:iff} thus carries over, replacing deep bisimilarity by bisimilarity:
\begin{lemma}\label{lemma:pcstar:iff}\leanok{Main.lean}{71}{main_pcstar_iff}
  For every generic monitor $\Monitor$, there exists a function implementing policy change from $\phi_1$ to $\phi_2$ in the sense of Definition~\ref{def:pcstar} iff ${\sim_{\phi_1}} \subseteq {\sim_{\phi_2}}$.
\end{lemma}
\textsc{PolicyChange}* is a purely semantic problem: unlike \textsc{PolicyChange}, it only depends on the languages $\Sem{\phi_1}$ and $\Sem{\phi_2}$.
The two problems are incomparable:
the example of Section~\ref{sec:change approach} shows that policy change from $(P \Since Q) \land \lnot \Once R$ to $P \Since Q$ is possible for treelike monitors but not in general, while policy change from $\Once(P \lor Q)$ to the equivalent formula $\Once P \lor \Once Q$ is trivially possible in general, but not for treelike monitors.

For pLTL and discrete-time pMTL, there exist deterministic monitors whose states are described by polynomially (exponentially, resp.) many bits: for pLTL, it suffices to store the truth values of all subformulae at the last time-point; for discrete-time pMTL, the state must additionally store the set of distances $\leq M$ of the witnesses of operators $\alpha \Since_I \beta$, where $M$
is the largest finite bound of any interval in the formula.
Hence, by Lemma~\ref{lemma:pcstar:iff}, policy change is impossible iff there are two prefixes $s,s'$ that lead to bisimilar states of the monitor for $\phi_1$ and to non-bisimilar states of the monitor for $\phi_2$.
Since the monitors's state is finite, such prefixes can be found by a terminating search
procedure, and we prove:
\begin{theorem}\label{thm:pcstar pltl}\leanpart{Main.lean}{74}{main_pcstar_lower_bound}
  \textsc{PolicyChange}* is PSPACE-complete for pLTL and EXPSPACE-complete for pMTL with discrete-time semantics.
\end{theorem}
Thus, for discrete time, dropping the treelike assumption does not change the complexity of policy change.
Whether \textsc{PolicyChange}* is decidable for real-time semantics remains open; in this case, the Nerode congruence of a pMTL-definable timed language generally has infinitely many classes, and the techniques above are not directly applicable.

\section{Conclusion}\label{sec:concl}

In this paper, we have introduced the \emph{policy change} problem that arises in the context of
online monitoring of changing specifications, and studied it for \emph{treelike} monitors. We have shown that over real-time semantics our problem is decidable and $\mathbf{F}_{\omega^\omega}$-complete, while over discrete-time semantics it is EXPSPACE-complete, and PSPACE-complete for pLTL.
For general, non-treelike monitors, policy change  has the same complexity as in the treelike case for discrete time (Section~\ref{sec:general}).

Several questions remain open.
First, the decidability of policy change for general monitors with real-time semantics
is still open.
Second, our results concern the \emph{existence} of a policy change function $\Change$; the proof of Lemma~\ref{lemma:pc:iff} is not constructive, since it picks, for each state, a trace leading to an equivalent state.
Developing efficient procedures for computing $\Change$ (cf.\ Example~\ref{ex:phi_A_phi_B}) for concrete monitor implementations would significantly support practical use.
Third, implementations in existing monitoring tools, as well as extensions to richer specification languages such as first-order temporal logics and temporal logics with
past and future, are left for future work.

\appendix
\counterwithin{lemma}{section}
\counterwithin{definition}{section}

\newtheorem{innercustomlemma}{Lemma}
\newenvironment{customlemma}[1]
{\renewcommand\theinnercustomlemma{#1}\innercustomlemma}
{\endinnercustomlemma}
\newtheorem{innercustomtheorem}{Theorem}
\newenvironment{customtheorem}[1]
{\renewcommand\theinnercustomtheorem{#1}\innercustomtheorem}
{\endinnercustomtheorem}
\newtheorem{innercustomcorollary}{Corollary}
\newenvironment{customcorollary}[1]
{\renewcommand\theinnercustomcorollary{#1}\innercustomcorollary}
{\endinnercustomcorollary}

\newcommand*{\verteq}{\raisebox{-2pt}{\rotatebox{90}{$=\,$}}}

\section{Detailed Proofs}
\label{apx:proofs}

\subsection{Policy Change as Hyperlogic Entailment}

We note some straightforward facts:
\begin{proposition}\label{prop:monitor:simps}
  \begin{compactenum}[(1)]
  \item $\forall q \in \States^\phi.\;\forall s,s' \in \Fintraces.\;\Step^\phi(\Step^\phi(q,s),s') = \Step^\phi(q,s \cdot s')$.
  \item $\forall q \in \States^\phi.\;\forall s \in \Fintraces.\;\forall s' \in \Netraces.\;\Out^\phi(\Step^\phi(q,s),s') = \Out^\phi(q,s \cdot s')$.
  \end{compactenum}
\end{proposition}

\begin{customlemma}{\ref{lemma:deepbisim:indep}}\leanok{Main.lean}{17}{main_deep_indep}
  For every pair $T^\phi,U^\phi$ of monitor trees for $\phi$, $s \deepbisim_{T^\phi} s'$ iff $s \deepbisim_{U^\phi} s'$.
\end{customlemma}
\begin{proof}
  Let $T^\phi$ be a monitor tree for $\phi$ and let $v$ be a vertex of $\phi$'s syntax tree, labeled by the tester $t_v$ for the operator $o_v$.
  As explained in Section~\ref{sec:treelike}, the substate of $\Step^{T^\phi}(\Init^{T^\phi},s)$ at $v$ is $\Step^{t_v}(\Init^{t_v},\Eval_v(s))$.
  Since $t_v$ is a correct monitor for $o_v$, Proposition~\ref{prop:monitor:simps} yields, for all $s'' \in \Netraces$ over $2^{\Props\cup\{X,Y\}} \times \Re$,
  \[
    \Out^{t_v}(\Step^{t_v}(\Init^{t_v},\Eval_v(s)),s'') = \Out^{t_v}(\Init^{t_v},\Eval_v(s)\cdot s'') = (\text{if } \Eval_v(s)\cdot s'' \in \Sem{o_v} \text{ then } 1 \text{ else } 0).
  \]
  Hence the substates at $v$ after reading $s$ and $s'$ are bisimilar iff $\Eval_v(s) \bisim_{o_v} \Eval_v(s')$, where for a pMTL formula $\psi$ we write $r \bisim_\psi r'$ iff $\forall s'' \in \Netraces.\; r \cdot s'' \in \Sem{\psi} \iff r'\cdot s'' \in \Sem{\psi}$.
  This condition only depends on $\phi$ and $v$, not on $t_v$.
  As $s \deepbisim_{T^\phi} s'$ is the conjunction of these conditions over all vertices $v$, it does not depend on $T^\phi$, and the same holds for $U^\phi$.
\end{proof}

\begin{lemma}\label{lemma:deepbisim:conv}\leanok{Trees.lean}{376}{f_run_init}
  Let $\Tuple{\Monitor^\phi,T^\phi,f^\phi}$ be a treelike monitor.
  Then $s \deepbisim_{\phi} s'$ iff \[\Step^{\Monitor^{\phi}}(\Init^{\Monitor^{\phi}},s) \deepbisim_{f^{\phi}}
    \Step^{\Monitor^{\phi}}(\Init^{\Monitor^{\phi}},s').\]
\end{lemma}
\begin{proof}
  \begin{align*}
    s \deepbisim_{\phi} s' \iff
    & \Step^{T^{\phi}}(\Init^{T^{\phi}},s) \deepbisim_{T^{\phi}} \Step^{T^{\phi}}(\Init^{T^{\phi}},s') \\
    \iff & \Step^{\Monitor^{\phi}}(\Init^{\Monitor^{\phi}},s) \deepbisim_{f^{\phi}} \Step^{\Monitor^{\phi}}(\Init^{\Monitor^{\phi}},s')
  \end{align*}
  where the last step is justified
  because $f^{\phi_1}(\Step^{\Monitor^{\phi_1}}(\Init^{\Monitor^{\phi_1}},s)) = \Step^{T^{\phi_1}}(\Init^{T^{\phi_1}},s)$ using Definition~\ref{def:treelike}, and similarly for $s'$.
\end{proof}

\begin{customlemma}{\ref{lemma:pc:iff}}\leanok{Main.lean}{20}{main_pc_iff}
  There exists, for every generic treelike monitor $\Tuple{\Monitor,T,f}$, a function implementing policy change from $\phi_1$ to $\phi_2$ if and only if $[s]_{{\phi_1}} \subseteq [s]_{{\phi_2}}$ for all finite traces $s$.
\end{customlemma}
\begin{proof}
  \fbox{$\Longrightarrow$}
  Choose an arbitrary generic treelike monitor $\Tuple{\Monitor,T,f}$.
  (Such a monitor always exists. E.g., one can construct the monitor trees by ``implementing'' the transducers using trace equivalence classes for the states, and then choose $f^\phi$ to be the identity function.)
  Let $\Change$ be its policy change function.
  Consider arbitrary traces $s$ and $s'$ such that $s' \in [s]_{\phi_1}$, i.e., $s \deepbisim_{\phi_1} s'$.
  Thus
  \begin{equation}\label{eq:pc:iff:a}
    \Step^{\Monitor^{\phi_1}}(\Init^{\Monitor^{\phi_1}},s) \deepbisim_{f^{\phi_1}}
    \Step^{\Monitor^{\phi_1}}(\Init^{\Monitor^{\phi_1}},s')
  \end{equation}
  by Lemma~\ref{lemma:deepbisim:conv}.
  Applying policy change (Definition~\ref{def:pc}) twice to \eqref{eq:pc:iff:a},
  we obtain
  \[
    \Step^{\Monitor^{\phi_2}}(\Init^{\Monitor^{\phi_2}},s) \deepbisim_{f^{\phi_2}}
    \Change(\Step^{\Monitor^{\phi_1}}(\Init^{\Monitor^{\phi_1}},s)) \deepbisim_{f^{\phi_2}}
    \Step^{\Monitor^{\phi_2}}(\Init^{\Monitor^{\phi_2}},s').
  \]
  Using Lemma~\ref{lemma:deepbisim:conv} again, it follows that $s \deepbisim_{\phi_2} s'$ and hence $s' \in [s]_{\phi_2}$.

  \fbox{$\Longleftarrow$}
  Fix an arbitrary generic treelike monitor $\Tuple{\Monitor,T,f}$.
  We define $\Change(q)$ nondeterministically as the state $\Step^{\Monitor^{\phi_2}}(\Init^{\Monitor^{\phi_2}},s(q))$ for some finite trace $s(q)$ such that $q \deepbisim_{f^{\phi_1}} \Step^{\phi_1}(\Init^{\phi_1},s(q))$.
  The new state can be chosen arbitrarily if there is no such $s(q)$.

  Consider an arbitrary state $q \in \States^{\Monitor^{\phi_1}}$ and a trace $s$ satisfying $q \deepbisim_{f^{\phi_1}} \Step^{\phi_1}(\Init^{\phi_1},s)$.
  Thus
  \[
    \Step^{\phi_1}(\Init^{\phi_1},s) \deepbisim_{f^{\phi_1}}
    q \deepbisim_{f^{\phi_1}} \Step^{\phi_1}(\Init^{\phi_1},s(q)).
  \]
  Using Lemma~\ref{lemma:deepbisim:conv}, it follows that $s \deepbisim_{\phi_1} s(q)$ and hence $s \deepbisim_{\phi_2} s(q)$ using the assumption $[s]_{{\phi_1}} \subseteq [s]_{{\phi_2}}$.
  Therefore
  \[
    \Change(q) = \Step^{\Monitor^{\phi_2}}(\Init^{\Monitor^{\phi_2}},s(q)) \deepbisim_{f^{\phi_2}} \Step^{\Monitor^{\phi_2}}(\Init^{\Monitor^{\phi_2}},s)
  \]
  by Lemma~\ref{lemma:deepbisim:conv}, which proves that $\Change$ implements policy change correctly.
\end{proof}

\begin{customlemma}{\ref{lemma:build:extract}}\leanok{Main.lean}{26}{main_lemma_12}
  Suppose that for all finite traces $s$, non-empty sets of aligned, finite traces $S$, and formulae $\phi$:
  \begin{enumerate}[(1)]
    \item $S \subseteq [s]_{\phi}$ implies $S \hypermodels \Extract{\phi}$.
    \item $s \in S$ and $S \hypermodels \Build{\phi}$ imply $S \subseteq [s]_{\phi}$.
    \item $s \deepbisim_\phi s^{+d}$ for every $d \in \Re$, where $s^{+d}$ denotes the copy of $s$ where the first delay has been incremented by $d$.
  \end{enumerate}
  Then the validity of $\Extract{\phi_1} \Imp \Build{\phi_2}$ implies $[s]_{\phi_1} \subseteq [s]_{\phi_2}$ for all $s$.
  Moreover, if $\Extract{\phi} \equiv \Build{\phi} \equiv: \BE{\phi}$ for both $\phi \in \{\phi_1,\phi_2\}$, the converse holds, too.
\end{customlemma}
\begin{proof}
  \fbox{$\Longrightarrow$}
  Assume that the hyperformula $\Extract{\phi_1} \Imp \Build{\phi_2}$ is valid and fix an arbitrary finite trace $s$.
  Denote the last absolute timestamp of $s$ by $\tms(s)$.
  For every $x \geq \tms(s)$, let $[s]_{\phi_1}^x$ be the subset of $[s]_{\phi_1}$ that contains all traces whose last absolute timestamp is $x$, such that $[s]_{\phi_1}^x$ is aligned.
  By assumption (3), $[s]_{\phi_1}^x$ contains $s^{+(x-\tms(s))}$ and is thus non-empty.
  We have $[s]_{\phi_1}^x \hypermodels \Extract{\phi_1}$ by assumption (1), setting $S = [s]_{\phi_1}^x$.
  It follows that $[s]_{\phi_1}^x \hypermodels \Build{\phi_2}$ by modus ponens.
  Hence we obtain $[s]_{\phi_1}^x \subseteq [s^{+(x-\tms(s))}]_{\phi_2} = [s]_{\phi_2}$ by assumption (2) instantiated with $s^{+(x-\tms(s))}$ and $[s]_{\phi_1}^x$, and assumption (3).
  It remains to show that every trace $u$ in $[s]_{\phi_1}$ whose last timestamp $\tms(u)$ is less than $\tms(s)$ is also contained in $[s]_{\phi_2}$:
  \begin{align*}
    \hspace{-5pt}u \in [s]_{\phi_1} \stackrel{(3)}{\implies} u^{+\tms(s)} \in [s]_{\phi_1} \implies u^{+\tms(s)} \in [s]_{\phi_1}^{\tms(u)+\tms(s)} \implies{}   u^{+\tms(s)} \in [s]_{\phi_2} \stackrel{(3)}{\implies} u \in [s]_{\phi_2}.
  \end{align*}

  \fbox{$\Longleftarrow$}
  Suppose that $[s]_{\phi_1} \subseteq [s]_{\phi_2}$ for all $s$ and fix an arbitrary, non-empty set of aligned, finite traces $S$ satisfying $S \hypermodels \BE{\phi_1} \equiv \Build{\phi_1}$.
  Let $s$ be some element of $S$.
  By assumption (2), we have $S \subseteq [s]_{\phi_1}$.
  It follows that $S \subseteq [s]_{\phi_2}$ and hence $S \hypermodels \BE{\phi_2} \equiv \Extract{\phi_2}$.
  Because $S$ was arbitrary, we have shown that $\BE{\phi_1} \Imp \BE{\phi_2}$ is valid.
\end{proof}

\subsection{Information Formulae for pMTL}
\label{sec:upper proofs}

\begin{customlemma}{\ref{lemma:pmtl to ptptl}}\leanok{Hyper.lean}{48}{sat_emb}
  For all $\phi$, $s$, $i$, and $v$,
  $\Tuple{s,i,v} \models \Embed{\phi}$ iff $\Tuple{s,i} \models \phi$.
\end{customlemma}
\begin{proof}
  By induction on the structure of $\phi$.
  The cases follow immediately from the definitions.
\end{proof}

\newcommand*{\Equiv}{\mathit{Ind}}

Before we prove Lemma~\ref{lemma:information formulae}, we provide an alternative characterization of the equivalence relation $\deepbisim_\phi$.
For every pMTL formula $\phi_1$, define $\Eval(\phi_1;s)$ to be the trace derived from $s \in \Fintraces$ as follows.
The number of events and their timestamps are the same as in $s$, whereas the event set $\ev_i$ is extended with $X$ iff $\Tuple{s,i} \models \phi_1$, for all $i < |s|$.
The trace $\Eval(\phi_1,\phi_2;s)$ is defined similarly, except that the event set $\ev_i$ additionally contains $Y$ iff $\Tuple{s,i} \models \phi_2$.
Thus, $\Eval(\phi_1;s)$ and $\Eval(\phi_1,\phi_2;s)$ are the traces $\Eval_v(s)$ read by the tester at a vertex $v$ whose children correspond to $\phi_1$ (and $\phi_2$), see Section~\ref{sec:treelike}.
We define
\begin{align*}
  s \bisim_\phi s' \iff \left(\forall s'' \in \Netraces.\; s \cdot s'' \in \Sem{\phi} \iff s' \cdot s'' \in \Sem{\phi}\right),
\end{align*}
where $s, s', s''$ range over traces over $2^{\Props \cup \{X,Y\}}$; in particular, the continuation $s''$ may assign arbitrary truth values to $X$ and $Y$.
Finally, the relation $\Equiv(\phi;s,s')$ is defined by recursion on $\phi$.
\begin{align*}
  \Equiv(\True;s,s') &\iff \Equiv(P;s,s') \iff \mathit{true} \\
  \Equiv(\lnot \phi;s,s') &\iff \Equiv(\phi;s,s') \\
  \Equiv(\phi_1 \lor \phi_2;s,s') &\iff \Equiv(\phi_1;s,s') \land \Equiv(\phi_2;s,s') \\
  \Equiv(\Prev_I \phi;s,s') &\iff \Eval(\phi;s) \bisim_{\Prev_I X} \Eval(\phi;s') \land \Equiv(\phi;s,s') \\
  \Equiv(\phi_1 \Since_I \phi_2;s,s') &\iff \Eval(\phi_1,\phi_2;s) \bisim_{X \Since_I Y} \Eval(\phi_1,\phi_2;s') \land{}\\[-\jot]
  &\phantom{{}\iff{}} \quad\Equiv(\phi_1;s,s') \land \Equiv(\phi_2;s,s')
\end{align*}

\begin{lemma}\label{lemma:deepbisim:alt}\leanok{Info.lean}{300}{deepEquiv_iff_infoSem}
  $s \deepbisim_\phi s'$ iff $\Equiv(\phi;s,s')$.
\end{lemma}
\begin{proof}
  Note that $s \deepbisim_\phi s'$ is equivalent to $\Step^{T^{\phi}}(\Init^{T^{\phi}},s) \deepbisim_{T^{\phi}} \Step^{T^{\phi}}(\Init^{T^{\phi}},s')$ for all monitor trees $T^\phi$.
  By the proof of Lemma~\ref{lemma:deepbisim:indep}, this holds iff $\Eval_v(s) \bisim_{o_v} \Eval_v(s')$ for all vertices $v$; for leaves and non-temporal operators $o_v$, this condition is always true, since the truth value of $o_v$ at a time-point only depends on the event at that time-point.
  The cases below make this explicit.
  We proceed by structural induction on $\phi$.
  \begin{itemize}
    \item $\phi \equiv \True$:
      The monitor tree $T^\True$ consists of a single vertex.
      Since the output of the transducer labeling it does not depend on past time-points, $s \deepbisim_\True s'$ is true for all pairs $s,s'$ of traces.
      This matches the definition of $\Equiv(\True;s,s')$.
    \item $\phi \equiv P$ for some predicate $P$: Similar to the previous case.
    \item $\phi \equiv \lnot \phi'$:
      The root of the monitor tree $T^\phi$ is labeled by a transducer $t$ implementing the operator $\lnot X$.
      It has a single child, namely $T^{\phi'}$.
      Now observe that $s \deepbisim_\phi s'$ iff $s \deepbisim_{\phi'} s'$ and
      $\Step^{t}(\Init^{t},\Eval(\phi';s)) \bisim_{t} \Step^{t}(\Init^{t},\Eval(\phi';s'))$.
      By the induction hypothesis, $s \deepbisim_{\phi'} s'$ iff $\Equiv(\phi';s,s')$, whereas 
      $\Step^{t}(\Init^{t},r) \bisim_{t} \Step^{t}(\Init^{t},r')$ holds for all pairs $r,r'$ of traces, as $t$'s output does not depend on past time-points.
    \item $\phi \equiv \phi_1 \lor \phi_2$:
      The root of the monitor tree $T^\phi$ is labeled by a transducer $t$ implementing the operator $X \lor Y$.
      It has two children, namely $T^{\phi_1}$ and $T^{\phi_2}$.
      Now observe that $s \deepbisim_\phi s'$ iff $s \deepbisim_{\phi_1} s'$, $s \deepbisim_{\phi_2} s'$, and
      $\Step^{t}(\Init^{t},\Eval(\phi_1,\phi_2;s)) \bisim_{t} \Step^{t}(\Init^{t},\Eval(\phi_1,\phi_2;s'))$.
      By the induction hypotheses for $\phi_1$ and $\phi_2$, we have
      \[
        s \deepbisim_{\phi_1} s' \land s \deepbisim_{\phi_2} s' \iff \Equiv(\phi_1;s,s') \land \Equiv(\phi_2;s,s'),
      \]
      whereas $\Step^{t}(\Init^{t},r) \bisim_{t} \Step^{t}(\Init^{t},r')$ holds again for all pairs $r,r'$.
    \item $\phi \equiv \Prev_I \phi'$:
      This case is similar to the one for negation, except that we must show that $\Eval(\phi';s) \bisim_{\Prev_I X} \Eval(\phi';s')$ iff \[\Step^{t}(\Init^{t},\Eval(\phi';s)) \bisim_{t} \Step^{t}(\Init^{t},\Eval(\phi';s')).\]
      This follows immediately from the definition of $\bisim_t$ and the fact that $t$ implements $\Prev_I X$.
    \item $\phi \equiv \phi_1 \Since_I \phi_2$:
      This case is similar to the one for disjunction, except that we must show that $\Eval(\phi_1,\phi_2;s) \bisim_{X \Since_I Y} \Eval(\phi_1,\phi_2;s')$ iff \[\Step^{t}(\Init^{t},\Eval(\phi_1,\phi_2;s)) \bisim_{t} \Step^{t}(\Init^{t},\Eval(\phi_1,\phi_2;s')).\]
      This follows immediately from the definition of $\bisim_t$ and the fact that $t$ implements $X \Since_I Y$.
  \end{itemize}
\end{proof}

\begin{lemma}\label{lemma:foo:empty}\leanok{Info.lean}{247}{opEquiv_empty}
  For all traces $s$ and $s'$, $s \bisim_{\Prev_\emptyset X} s'$ and $s \bisim_{X \Since_\emptyset Y} s'$.
\end{lemma}
\begin{proof}
  Trivial since $\Sem{\Prev_\emptyset X} = \Sem{X \Since_\emptyset Y} = \emptyset$.
\end{proof}

\begin{lemma}\label{lemma:foo:prev}\leanok{Info.lean}{198}{opEquiv_prev}
  Let $I \neq \emptyset$.
  For all non-empty traces $s$ and $s'$, $s \bisim_{\Prev_I X} s'$ iff $s \in \Sem{X} \iff s' \in \Sem{X}$.
\end{lemma}
\begin{proof}
  \fbox{$\Longrightarrow$}
  Suppose that $\forall s'' \in \Netraces.\; s \cdot s'' \in \Sem{\Prev_I X} \iff s' \cdot s'' \in \Sem{\Prev_I X}$.
  Choose any $d \in I$.
  Instantiate $s''$ such that it consists of a single event $(\emptyset,d)$.
  This ensures that the timestamp differences between the last and the second-to-last time-point in $s \cdot s''$ and $s' \cdot s''$ are both in $I$.
  Hence
  \[
    s \in \Sem{X} \iff s \cdot s'' \in \Sem{\Prev_I X} \iff s' \cdot s'' \in \Sem{\Prev_I X} \iff s' \in \Sem{X}.
  \]

  \fbox{$\Longleftarrow$}
  Suppose that $s \in \Sem{X} \iff s' \in \Sem{X}$ and fix an arbitrary non-empty trace $s''$.
  If $|s''| > 1$, note that the inclusion of $s \cdot s''$ and $s' \cdot s''$ in $\Sem{\Prev_I X}$ depends only on events in $s''$.
  Hence assume that $|s''| = 1$ and let $\tmsdiff$ be the delay of the only event in $s''$.
  If $\tmsdiff \not\in I$, neither $s \cdot s''$ nor $s' \cdot s''$ is included in $\Sem{\Prev_I X}$.
  If $\tmsdiff \in I$,
  \[
    s \cdot s'' \in \Sem{\Prev_I X} \iff s \in \Sem{X} \iff s' \in \Sem{X} \iff s' \cdot s'' \in \Sem{\Prev_I X}.
  \]
  Having covered all cases, we have shown the left-hand side.
\end{proof}

\begin{lemma}\label{lemma:foo:since}\leanok{Info.lean}{218}{opEquiv_since}
  For all non-empty traces $s$ and $s'$, $s \bisim_{X \Since_I Y} s'$ iff, for all $d \in \Re$, $s \in \Sem{X \Since_{I-d} Y} \iff s' \in \Sem{X \Since_{I-d} Y}$.
\end{lemma}
\begin{proof}
  \fbox{$\Longrightarrow$}
  Suppose that $\forall s'' \in \Netraces.\; s \cdot s'' \in \Sem{X \Since_I Y} \iff s' \cdot s'' \in \Sem{X \Since_I Y}$.
  Without loss of generality, fix an arbitrary $d \in \Re$ such that $I - d$ is non-empty.
  Instantiate $s''$ such that it consists of a single event $(\{X\},d)$.
  It is not difficult to see that
  \[
    s \in \Sem{X \Since_{I-d} Y} \iff s \cdot s'' \in \Sem{X \Since_I Y} \iff s' \cdot s'' \in \Sem{X \Since_I Y} \iff s' \in \Sem{X \Since_{I-d} Y}.
  \]

  \fbox{$\Longleftarrow$}
  Suppose that $s \in \Sem{X \Since_{I-d} Y} \iff s' \in \Sem{X \Since_{I-d} Y}$ for all $d$ and fix an arbitrary non-empty trace $s''$.
  We distinguish several cases:
  \begin{itemize}
    \item $s'' \in \Sem{X \Since_I Y}$, which implies that $s \cdot s'' \in \Sem{X \Since_I Y}$ and $s' \cdot s'' \in \Sem{X \Since_I Y}$.
    \item The previous case does not apply and there is an event in $s''$ that does not satisfy $X$.
      Then $s \cdot s'' \not\in \Sem{X \Since_I Y}$ and $s' \cdot s'' \not\in \Sem{X \Since_I Y}$.
    \item Neither case applies.
      Then $s \cdot s'' \in \Sem{X \Since_I Y}$ iff $s \in \Sem{X \Since_{I-d} Y}$, where $d$ is the sum of the delays of $s''$, and similarly for $s'$.
  \end{itemize}
\end{proof}
\begin{customlemma}{\ref{lemma:information formulae}}\leanok{Main.lean}{36}{main_lemma_16}
  The HyperpTPTL formula $\BE{\phi}$ fulfills the conditions of Lemma~\ref{lemma:build:extract}.
\end{customlemma}
\begin{proof}
  Recall the conditions:
  \begin{enumerate}[(1)]
    \item $S \subseteq [s]_{\phi}$ implies $S \hypermodels \BE{\phi}$.
    \item $s \in S$ and $S \hypermodels \BE{\phi}$ imply $S \subseteq [s]_{\phi}$.
    \item $s \deepbisim_\phi s^{+d}$ for every $d \in \Re$, where $s^{+d}$ denotes the copy of $s$ where the first delay has been incremented by $d$.
  \end{enumerate}
  By Lemma~\ref{lemma:deepbisim:alt}, the condition $S \subseteq [s]_\phi$ is equivalent to $\forall s' \in S.\;\Equiv(\phi;s,s')$, and the condition $s \deepbisim_\phi s^{+d}$ is equivalent to $\Equiv(\phi;s,s^{+d})$.
  We proceed by structural induction on $\phi$.
  \begin{itemize}
    \item $\phi \equiv \True$ or $\phi \equiv P$: Trivial, as both $S \hypermodels \BE{\phi}$ and $\Equiv(\phi;s,s')$ hold for all $S$, $s$, and $s'$.
    \item $\phi \equiv \lnot \phi'$: Note that $S \hypermodels \BE{\phi}$ iff $S \hypermodels \BE{\phi'}$ and $\Equiv(\phi;s,s')$ iff $\Equiv(\phi';s,s')$, so this case follows directly from the induction hypothesis.
    \item $\phi \equiv \phi_1 \lor \phi_2$: Similarly, as $S \hypermodels \BE{\phi}$ iff both $S \hypermodels \BE{\phi_1}$ and $S \hypermodels \BE{\phi_2}$, and $\Equiv(\phi;s,s')$ iff both $\Equiv(\phi_1;s,s')$ and $\Equiv(\phi_2;s,s')$.
    \item $\phi \equiv \Prev_I \phi'$:
      \begin{enumerate}[(1)]
        \item Suppose that $\forall s' \in S.\;\Equiv(\phi;s,s')$, i.e., \[\forall s' \in S.\; \Eval(\phi';s) \bisim_{\Prev_I X} \Eval(\phi';s') \land \Equiv(\phi';s,s').\]
          Using the induction hypothesis, it follows that $S \hypermodels \BE{\phi'}$.
          It remains to show that $S \hypermodels \tforall \pi,\pi'.~\{\phi'\}_{\pi} \Imp \{\phi'\}_{\pi'}$ if $I \neq \emptyset$.
          To this end, fix two arbitrary, non-empty traces $r,r'$ in $S$ such that $r \in \Sem{\phi'}$.
          (As traces in $S$ are aligned, $\phi'$ is necessarily evaluated at the last time-point of each trace.)
          We have
          \[
            \Eval(\phi';r) \bisim_{\Prev_I X} \Eval(\phi';s) \bisim_{\Prev_I X} \Eval(\phi';r')
          \]
          because it is an equivalence relation.
          Using Lemma~\ref{lemma:foo:prev}, this implies $\Eval(\phi';r) \in \Sem{X} \iff \Eval(\phi';r') \in \Sem{X}$, i.e., $r \in \Sem{\phi'} \iff r' \in \Sem{\phi'}$.
        \item Suppose that $s \in S$ and $S \hypermodels \BE{\phi}$.
          Hence $S \hypermodels \BE{\phi'}$ by construction of $\BE{\phi}$, and from the induction hypothesis we get that $\forall s' \in S.\;\Equiv(\phi';s,s')$.
          It remains to show that $\forall s' \in S.\;\Eval(\phi';s) \bisim_{\Prev_I X} \Eval(\phi';s')$.
          Fix an arbitrary $s' \in S$.
          Using Lemma~\ref{lemma:foo:empty}, it suffices to consider the case where $I \neq \emptyset$, in which case we know that $S \hypermodels \tforall \pi,\pi'.~\{\phi'\}_{\pi} \Imp \{\phi'\}_{\pi'}$.
          We instantiate the quantifiers twice, once with $s$ and $s'$, and once with $s'$ and $s$, to conclude that $s \in \Sem{\phi'} \iff s' \in \Sem{\phi'}$ and hence $\Eval(\phi';s) \in \Sem{X} \iff \Eval(\phi';s') \in \Sem{X}$.
          Now apply Lemma~\ref{lemma:foo:prev}.
        \item We obtain $\Equiv(\phi';s,s^{+d})$ from the induction hypothesis and it remains to prove that $\Eval(\phi';s) \bisim_{\Prev_I X} \Eval(\phi';s^{+d})$.
          Fix an arbitrary non-empty trace $s''$.
          The difference between any pair of timestamps in $\Eval(\phi';s) \cdot s''$ is the same as in $\Eval(\phi';s^{+d}) \cdot s''$ because concatenation is defined such that the timestamps of $s''$ are interpreted as relative increments.
          Moreover, the events of the two concatenations are pointwise equal.
          Therefore, $\Eval(\phi';s) \cdot s'' \in \Sem{\Prev_I X} \iff \Eval(\phi';s^{+d}) \cdot s'' \in \Sem{\Prev_I X}$ because $\Prev_I$'s semantics depends only on the events and timestamp differences.
      \end{enumerate}
    \item $\phi \equiv \phi_1 \Since_I \phi_2$:
      \begin{enumerate}[(1)]
        \item Suppose that $\forall s' \in S.\;\Equiv(\phi;s,s')$, i.e., \[\forall s' \in S.\; \Eval(\phi_1,\phi_2;s) \bisim_{X \Since_I Y} \Eval(\phi_1,\phi_2;s') \land \Equiv(\phi_1;s,s') \land \Equiv(\phi_2;s,s').\]
          Using the induction hypothesis, it follows that $S \hypermodels \BE{\phi_1} \HConj \BE{\phi_2}$.
          It remains to show \[S \hypermodels \tforall \pi,\pi'.~\{\phi_1 \Since_{I-d} \phi_2\}_{\pi} \Imp \{\phi_1 \Since_{I-d} \phi_2\}_{\pi'}\] for all  $d \in \Re$ if $I \neq \emptyset$.
          To this end, fix $d$ and two arbitrary, non-empty traces $r,r'$ in $S$ such that $r \in \Sem{\phi_1 \Since_{I-d} \phi_2}$.
          We have
          \[
            \Eval(\phi_1,\phi_2;r) \bisim_{X \Since_I Y} \Eval(\phi_1,\phi_2;s) \bisim_{X \Since_I Y} \Eval(\phi_1,\phi_2;r').
          \]
          Using Lemma~\ref{lemma:foo:since}, this implies $\Eval(\phi_1,\phi_2;r) \in \Sem{X \Since_{I-d} Y} \iff \Eval(\phi_1,\phi_2;r') \in \Sem{X \Since_{I-d} Y}$, i.e., $r \in \Sem{\phi_1 \Since_{I-d} \phi_2} \iff r' \in \Sem{\phi_1 \Since_{I-d} \phi_2}$.
        \item Suppose that $s \in S$ and $S \hypermodels \BE{\phi}$.
          Hence $S \hypermodels \BE{\phi_1} \HConj \BE{\phi_2}$ by construction of $\BE{\phi}$, and from the induction hypothesis we get that $\forall s' \in S.\;\Equiv(\phi_1;s,s') \land \Equiv(\phi_2;s,s')$.
          It remains to show that $\forall s' \in S.\;\Eval(\phi_1,\phi_2;s) \bisim_{X \Since_I Y} \Eval(\phi_1,\phi_2;s')$.
          Fix an arbitrary $s' \in S$.
          Using Lemma~\ref{lemma:foo:empty}, it suffices to consider the case where $I \neq \emptyset$, in which case we know that \[S \hypermodels \tforall \pi,\pi'.~\{\phi_1 \Since_{I-d} \phi_2\}_{\pi} \Imp \{\phi_1 \Since_{I-d} \phi_2\}_{\pi'}\]
          for all $d \in \Re$.
          (Strictly speaking, $\BE{\phi}$ guarantees this only up to some upper bound of $d$.
          However, as discussed in Section~\ref{sec:information}, the upper bound is chosen such that the interval $I - d$ is either empty or the same for all larger $d$.)
          We again instantiate the quantifiers twice to conclude that $s \in \Sem{\phi_1 \Since_{I-d} \phi_2} \iff s' \in \Sem{\phi_1 \Since_{I-d} \phi_2}$ and hence $\Eval(\phi_1,\phi_2;s) \in \Sem{X \Since_{I-d} Y} \iff \Eval(\phi_1,\phi_2;s') \in \Sem{X \Since_{I-d} Y}$ for all $d$.
          Now apply Lemma~\ref{lemma:foo:since}.
        \item We obtain $\Equiv(\phi_1;s,s^{+d})$ and $\Equiv(\phi_2;s,s^{+d})$ from the induction hypotheses and it remains to prove that $\Eval(\phi_1,\phi_2;s) \bisim_{X \Since_I Y} \Eval(\phi_1,\phi_2;s^{+d})$.
          Fix an arbitrary non-empty trace $s''$.
          The difference between any pair of timestamps in $\Eval(\phi_1,\phi_2;s) \cdot s''$ is the same as in $\Eval(\phi_1,\phi_2;s^{+d}) \cdot s''$ because concatenation is defined such that the timestamps of $s''$ are interpreted as relative increments.
          Moreover, the events of the two concatenations are pointwise equal.
          Therefore, $\Eval(\phi_1,\phi_2;s) \cdot s'' \in \Sem{X \Since_I Y} \iff \Eval(\phi_1,\phi_2;s^{+d}) \cdot s'' \in \Sem{X \Since_I Y}$ because $\Since_I$'s semantics depends only on the events and timestamp differences.
      \end{enumerate}
  \end{itemize}
\end{proof}

\subsection{Rewriting of information formulae for $\HSince$}
\label{sec:proof:lemma:rewriting}

We first introduce some notation and an auxiliary lemma, then define the formula $\Phi_{\alpha,\beta,I}$, and finally prove Lemma~\ref{lemma:rewriting}.
Throughout, fix pMTL formulae $\alpha$ and $\beta$ (playing the roles of $\phi_1$ and $\phi_2$), a non-empty interval $I$, a non-empty set $S$ of aligned traces, an assignment $\Pi$ with $\pi,\pi' \in \dom\Pi$, a valuation $v$, and the evaluation instant $u = \tms(S)$.
We use the following abbreviations, for $\varpi \in \{\pi,\pi'\}$:
\[
  W_\varpi \equiv \HLift{\alpha \Since_{[0,0]} \beta}{\varpi},\qquad
  \Lambda_\varpi \equiv \HNeg\HLift{\Once_{[0,0]} \lnot\alpha}{\varpi}.
\]
At an instant $y$, $W_\varpi$ holds iff the trace $\Pi(\varpi)$ has a time-point $j$ with timestamp $y$ satisfying $\beta$ such that all later time-points with timestamp $y$ satisfy $\alpha$; $\HNeg\Lambda_\varpi$ holds iff $\Pi(\varpi)$ has a time-point with timestamp $y$ violating $\alpha$.
Note that $W_\varpi$ and $\HNeg\Lambda_\varpi$ are strict and $\Lambda_\varpi$ is lax.

\paragraph*{Witnesses.}
We say that $y \leq u$ is a \emph{witness} of $\varpi$ if $W_\varpi$ holds at $y$ and $\Lambda_\varpi$ holds at all instants in $(y,u]$, and we let $D_\varpi = \{u - y \mid y \text{ is a witness of } \varpi\} \subseteq \Re$, the set of witness \emph{distances}.
Directly from the semantics of $\Since$, for every interval $K$:
\begin{equation}\label{eq:witness}
  \Tuple{S,\Pi,u,v} \hypermodels \HLift{\alpha \Since_K \beta}{\varpi} \iff D_\varpi \cap K \neq \emptyset.
\end{equation}
Similarly, if $y$ is a witness of $\varpi$, then $\HLift{\alpha \Since_K \beta}{\varpi}$ holds at $y$ iff $D_\varpi \cap ((u-y) + K) \neq \emptyset$, since the time-points of $\Pi(\varpi)$ between a witness and $y$ satisfy $\alpha$.
Hence, the conjunction in Lemma~\ref{lemma:rewriting} fails iff there is a $d$ with $D_\pi \cap (I-d) \neq \emptyset = D_{\pi'} \cap (I-d)$; we then say that the \emph{window} $I-d$ \emph{separates} $\pi$ from $\pi'$.

\paragraph*{Windows.}
Write $\langle$ and $\rangle$ for the left and right brackets of $I$, and let $a = \inf I$.
If $I$ is bounded, let $b = \sup I$ and $L = b - a$.
A~direct computation shows that the non-empty windows $I - d$, $d \geq 0$, are exactly:
\begin{compactitem}
  \item the \emph{translates} $W_j = \langle j, j + L\rangle$ for $j \in [0,a]$ (taking $d = a - j$), and
  \item the \emph{prefixes} $[0,k]$ for $k \in [0,L)$ if $\rangle = {]}$, or $[0,k)$ for $k \in (0,L)$ if $\rangle = {)}$ (taking $d = b - k \in (a,b]$),
\end{compactitem}
if $I$ is bounded, and $\langle j,\infty)$ for $j \in [0,a]$ together with $[0,\infty)$ if $I$ is unbounded.

\begin{figure}[p]
  \centering
  \tikzset{
    dpi/.style={fill=blue!70!black, circle, inner sep=1.7pt},
    dpj/.style={draw=orange!85!black, very thick, cross out, inner sep=2.2pt},
    win/.style={line width=3.5pt, green!55!black, opacity=0.45},
    axis/.style={->, thin, gray!70!black},
    cond/.style={font=\scriptsize, align=left, draw=gray!50, rounded corners=2pt, inner sep=3pt}}
  \begin{subfigure}[t]{0.49\textwidth}
    \centering
    \begin{tikzpicture}[x=0.62cm, y=0.62cm]
      \draw[axis] (-0.2,0) -- (8.6,0);
      \draw[gray] (0,0.08) -- (0,-0.08) node[below=-1pt]{\tiny $0$};
      \draw[dashed, gray] (5,-0.25) -- (5,1.7) node[above, font=\scriptsize, text=black]{$a = 5$};
      \draw[dashed, gray] (7,-0.25) -- (7,1.7) node[above, font=\scriptsize, text=black]{$a{+}L = 7$};
      \node[dpj, label={[font=\scriptsize]below:$w{=}4$}] at (4,0) {};
      \node[dpi, label={[font=\scriptsize, xshift=3pt]below:$x{=}5$}] at (5,0) {};
      \node[dpj, label={[font=\scriptsize]below:$q{=}7$}] at (7,0) {};
      \draw[win] (4.5,0.75) -- (6.5,0.75);
      \node[left, font=\scriptsize] at (4.4,0.75) {$W_{4.5}$};
      \draw[decorate, decoration={brace, mirror, amplitude=3pt}] (4.05,-0.95) -- (6,-0.95)
        node[midway, below=2pt, font=\scriptsize]{$w + J = (4,6]$};
      \node[cond] at (4.2,-2.9) {(i)~$x \leq a + L$: $5 \leq 7$ \textcolor{green!45!black}{\checkmark}\\
        (ii)~$w < a$: $4 < 5$ \textcolor{green!45!black}{\checkmark}; $D' \cap (4,6] = \emptyset$ \textcolor{green!45!black}{\checkmark}\\
        $\Rightarrow$ separating $j \in (4,5)$, e.g., $W_{4.5} = [4.5,6.5]$};
    \end{tikzpicture}
    \caption{$I = [5,7]$, $D = \{5\}$, $D' = \{4,7\}$}
  \end{subfigure}\hfill
  \begin{subfigure}[t]{0.49\textwidth}
    \centering
    \begin{tikzpicture}[x=0.62cm, y=0.62cm]
      \draw[axis] (-0.2,0) -- (8.6,0);
      \draw[gray] (0,0.08) -- (0,-0.08) node[below=-1pt]{\tiny $0$};
      \draw[dashed, gray] (3,-0.25) -- (3,1.5) node[above, font=\scriptsize, text=black]{$a$};
      \draw[dashed, gray] (5,-0.25) -- (5,1.5) node[above, font=\scriptsize, text=black]{$a{+}L$};
      \node[dpj, label={[font=\scriptsize]below:$w$}] at (1.5,0) {};
      \node[dpi, label={[font=\scriptsize]below:$x$}] at (4,0) {};
      \node[dpj, label={[font=\scriptsize]below:$q$}] at (6.6,0) {};
      \draw[win] (2.6,0.75) -- (4.6,0.75);
      \node[left, font=\scriptsize] at (2.5,0.75) {$W_j$};
      \draw[decorate, decoration={brace, mirror, amplitude=3pt}] (1.6,-0.85) -- (3.5,-0.85)
        node[midway, below=2pt, font=\scriptsize]{$w + J$};
      \node[cond] at (4.2,-2.5) {(i)~$x \leq a + L$ \textcolor{green!45!black}{\checkmark}\\
        (ii)~$w < a$ \textcolor{green!45!black}{\checkmark}; $D' \cap (w + J) = \emptyset$ \textcolor{green!45!black}{\checkmark}};
    \end{tikzpicture}
    \caption{Conditions (i) and (ii)}
  \end{subfigure}

  \vspace{2mm}
  \begin{subfigure}[t]{0.49\textwidth}
    \centering
    \begin{tikzpicture}[x=0.62cm, y=0.62cm]
      \draw[axis] (-0.2,0) -- (8.6,0);
      \draw[gray] (0,0.08) -- (0,-0.08) node[below=-1pt]{\tiny $0$};
      \draw[dashed, gray] (2,-0.25) -- (2,1.5) node[above, font=\scriptsize, text=black]{$L$};
      \draw[dashed, gray] (3,-0.25) -- (3,1.5) node[above, font=\scriptsize, text=black]{$a$};
      \draw[dashed, gray] (5,-0.25) -- (5,1.5) node[above, font=\scriptsize, text=black]{$a{+}L$};
      \node[dpi, label={[font=\scriptsize]below:$x$}] at (4,0) {};
      \node[dpj, label={[font=\scriptsize]below:$q$}] at (6.6,0) {};
      \draw[win] (2.5,0.75) -- (4.5,0.75);
      \node[above, font=\scriptsize, fill=white, inner sep=1pt] at (3.9,0.9) {$W_j$};
      \draw[decorate, decoration={brace, mirror, amplitude=3pt}] (0,-0.85) -- (2,-0.85)
        node[midway, below=2pt, font=\scriptsize]{$[0,L]$};
      \node[cond] at (4.2,-2.75) {no point of $D'$ below $x$ (no $w$)\\
        (i)~$x \leq a + L$ \textcolor{green!45!black}{\checkmark}\\
        (iii)~$D' \cap [0,L] = \emptyset$ \textcolor{green!45!black}{\checkmark}};
    \end{tikzpicture}
    \caption{Conditions (i) and (iii): no $w$}
  \end{subfigure}\hfill
  \begin{subfigure}[t]{0.49\textwidth}
    \centering
    \begin{tikzpicture}[x=0.62cm, y=0.62cm]
      \draw[axis] (-0.2,0) -- (8.6,0);
      \draw[gray] (0,0.08) -- (0,-0.08) node[below=-1pt]{\tiny $0$};
      \draw[dashed, gray] (3,-0.25) -- (3,1.5) node[above, font=\scriptsize, text=black]{$L$};
      \node[dpi, label={[font=\scriptsize]below:$\min D$}] at (1.2,0) {};
      \node[dpi] at (4.5,0) {};
      \node[dpj] at (2.4,0) {};
      \node[dpj] at (6,0) {};
      \draw[win] (0,0.75) -- (1.8,0.75);
      \node[above, font=\scriptsize] at (0.9,0.8) {$[0,k]$};
      \draw[decorate, decoration={brace, mirror, amplitude=3pt}] (0,-1.0) -- (1.2,-1.0)
        node[midway, below=2pt, font=\scriptsize]{$[0,\min D]$};
      \node[cond] at (4.2,-2.6) {$\min D < L$ \textcolor{green!45!black}{\checkmark}\\
        $D' \cap [0,\min D] = \emptyset$ \textcolor{green!45!black}{\checkmark}};
    \end{tikzpicture}
    \caption{A separating prefix window}
  \end{subfigure}

  \vspace{2mm}
  \begin{subfigure}[t]{0.49\textwidth}
    \centering
    \begin{tikzpicture}[x=0.62cm, y=0.62cm]
      \draw[axis] (-0.2,0) -- (8.6,0);
      \draw[gray] (0,0.08) -- (0,-0.08) node[below=-1pt]{\tiny $0$};
      \draw[dashed, gray] (3.5,-0.25) -- (3.5,1.5) node[above, font=\scriptsize, text=black]{$a$};
      \node[dpj] at (0.8,0) {};
      \node[dpj, label={[font=\scriptsize]below:$\max D'$}] at (2.2,0) {};
      \node[dpi] at (1.5,0) {};
      \node[dpi, label={[font=\scriptsize]below:$\max D$}] at (5.5,0) {};
      \draw[win] (2.9,0.75) -- (8.3,0.75);
      \draw[green!40!black, -{Stealth[length=4pt]}] (8.3,0.75) -- (8.6,0.75);
      \node[above, font=\scriptsize, fill=white, inner sep=1pt] at (6,0.9) {$\langle j,\infty)$};
      \node[cond] at (4.2,-2.3) {$\max D' < \max D$ \textcolor{green!45!black}{\checkmark}\\
        $\max D' < a$ \textcolor{green!45!black}{\checkmark} $\;\Rightarrow\;$ $j \in (\max D', a]$};
    \end{tikzpicture}
    \caption{Unbounded $I = [a,\infty)$}
  \end{subfigure}\hfill
  \begin{subfigure}[t]{0.49\textwidth}
    \centering
    \begin{tikzpicture}[x=0.62cm, y=0.62cm]
      \node[dpi] at (0,0.8) {}; \node[right, font=\scriptsize] at (0.3,0.8) {$D$: witness distances of $\pi$};
      \node[dpj] at (0,0) {}; \node[right, font=\scriptsize] at (0.3,0) {$D'$: witness distances of $\pi'$};
      \draw[win] (-0.3,-0.8) -- (0.3,-0.8); \node[right, font=\scriptsize] at (0.3,-0.8) {a separating window};
      \node[right, font=\scriptsize] at (-0.3,-1.6) {\textcolor{green!45!black}{\checkmark}: condition of Lemma~\ref{lemma:separation} satisfied};
    \end{tikzpicture}
  \end{subfigure}
  \caption{Intuition for Lemma~\ref{lemma:separation}. Points are distances from the current instant (larger means further in the past). A window is separating if it contains a point of $D$ but none of $D'$.
  (a)~The separating translates of $[5,7]$ start strictly between two consecutive time-points, so they cannot be found by inspecting windows anchored at time-points.
  (b)~A separating translate $W_j = \langle j, j+L\rangle$ with $j \leq a$ must contain $x$ and fit strictly between $w$ (the closest point of $D'$ below $x$) and $q$ (the next one above). Such a $j$ exists iff $x$ is not too far ($x \leq a + L$), $w$ is not too far ($w < a$), and the gap after $w$ is long enough ($D' \cap (w + J) = \emptyset$). All three conditions only compare $x$ and $w$ with the current instant or $w$ with nearby points of the same trace, which makes them expressible with a single register.
  (c)~If $D'$ has no point below $x$, the window only has to avoid the first point $q$ of $D'$ above $x$; since $j \geq 0$, this is possible iff $D'$ has no point in $[0,L]$ (condition~(iii)), e.g., by sliding $W_j$ left until it ends before $q$.
  (d)~Prefix windows $[0,k]$, $k < L$, separate iff $\min D < L$ and no point of $D'$ lies in $[0, \min D]$.
  (e)~For unbounded $I$, the windows are the half-lines $\langle j,\infty)$, $j \leq a$ (and $[0,\infty)$): such a window must start above all points of $D'$ and below $\max D$, which is possible iff $\max D' < \max D$ and $\max D' < a$ (or $D' = \emptyset \neq D$).}
  \label{fig:separation}
\end{figure}

\paragraph*{Separation.}
Figure~\ref{fig:separation} illustrates the following lemma.
\begin{lemma}\label{lemma:separation}\leanok{Separation.lean}{122}{sepWindow_bounded}
  Let $D,D' \subseteq \Re_{\geq 0}$ be finite and $I = \langle a, b\rangle$ bounded with $L = b - a$.
  Some translate $W_j$, $j \in [0,a]$, satisfies $D \cap W_j \neq \emptyset = D' \cap W_j$ iff there is some $x \in D \setminus D'$ such that, writing $w = \max(D' \cap [0,x))$ if this set is non-empty:
  \begin{compactenum}[(i)]
    \item $x \leq a + L$ if $\rangle = {]}$, and $x < a+L$ if $\rangle = {)}$; moreover, $x > 0$ if $\langle = {(}$;
    \item if $w$ exists: $w < a$ if $\langle = {[}$, and $w \leq a$ if $\langle = {(}$; moreover, $D' \cap (w + J) = \emptyset$, where $w + J$ denotes the interval $\{w + z \mid z \in J\}$ and $J = (0,L)$ if both brackets are open, and $J = (0,L]$ otherwise;
    \item if $w$ does not exist: $D' \cap [0,L] = \emptyset$ if $\rangle = {]}$, and $D' \cap [0,L) = \emptyset$ if $\rangle = {)}$.
  \end{compactenum}
  Moreover, some prefix window separates $D$ from $D'$ iff $D \neq \emptyset$, $\min D < L$, and $D' \cap [0,\min D] = \emptyset$.
  If $I = \langle a,\infty)$ is unbounded, some window separates $D$ from $D'$ iff $D \neq \emptyset$ and either $D' = \emptyset$, or $\max D' < \max D$ and $\max D' < a$ (if $\langle = {[}$), resp.\ $\max D' \leq a$ (if $\langle = {(}$).
\end{lemma}
\begin{proof}
  Bounded case, translates.
  Fix $x \in D \setminus D'$, let $w$ be as above (or $-\infty$), and let $q = \min(D' \cap (x,\infty))$ (or $+\infty$).
  Since $D'$ has no element in $(w,q)$, a translate $W_j$ containing $x$ avoids $D'$ iff $W_j \subseteq (w,q)$ (for an open bracket, the bound may coincide with $w$ or $q$).
  Thus, the separating translates containing $x$ are given by the constraints
  \[
    j \mathrel{\triangleright_1} w,\quad j \mathrel{\triangleright_2} x - L,\quad j \geq 0,\qquad
    j \mathrel{\triangleleft_1} x,\quad j \mathrel{\triangleleft_2} q - L,\quad j \leq a,
  \]
  where $\triangleright_1$ is $>$ if $\langle = {[}$ and $\geq$ otherwise, $\triangleright_2$ is $\geq$ if $\rangle = {]}$ and $>$ otherwise, $\triangleleft_1$ is $\leq$ if $\langle = {[}$ and $<$ otherwise, and $\triangleleft_2$ is $<$ if $\rangle = {]}$ and $\leq$ otherwise.
  Such a $j$ exists iff every lower bound is compatible with every upper bound.
  The pairs $(w,x)$, $(x-L,x)$, $(x-L,q-L)$, $(0,a)$, and $(0,x)$ for $\langle = {[}$ are always compatible because $w < x < q$ and $L \geq 0$ (with $L > 0$ if a bracket is open).
  The remaining pairs yield exactly the conditions (i)--(iii): $(x-L,a)$ and $(0,x)$ yield~(i); $(w,a)$ and $(w,q-L)$ yield~(ii) (note that $q$ is the least element of $D'$ above $w$, so $q - w > L$ iff $D' \cap (w,w+L] = \emptyset$, and $q - w \geq L$ iff $D' \cap (w,w+L) = \emptyset$); and $(0,q-L)$ yields~(iii) if $w$ does not exist, while it follows from $(w,q-L)$ and $w \geq 0$ otherwise.
  As every separating translate contains some $x \in D \setminus D'$, the claim follows.

  Bounded case, prefixes.
  A prefix $[0,k]$ (resp.\ $[0,k)$) is hit by $D$ iff $\min D \leq k$ (resp.\ $< k$), and avoided by $D'$ iff $k < \min D'$ (resp.\ $k \leq \min D'$).
  A~suitable $k$ in the allowed range exists iff $\min D < L$ and $\min D < \min D'$.

  Unbounded case.
  A window $\langle j,\infty)$ is hit by a non-empty $D$ iff $\max D \geq j$ (resp.\ $>j$), and $[0,\infty)$ is hit iff $D \neq \emptyset$.
  The claim follows by choosing $j$ in $(\max D', \min(\max D, a)]$ (resp.\ $[\max D',\min(\max D,a))$).
\end{proof}

\begin{figure}[p]
  \centering
  \tikzset{
    wp/.style={fill=blue!70!black, circle, inner sep=1.7pt},
    wq/.style={draw=orange!85!black, very thick, cross out, inner sep=2.2pt},
    lane/.style={thin, gray!70!black, ->},
    nogo/.style={orange!25},
    scan/.style={-{Stealth[length=4pt]}, thick, green!45!black},
    dist/.style={{Bar[width=4pt]}-{Bar[width=4pt]}, thin}}
  \begin{subfigure}{0.49\textwidth}
    \centering
    \begin{tikzpicture}[x=0.6cm, y=0.6cm]
      \draw[lane] (0,1) node[left, font=\scriptsize]{$\pi$} -- (8.3,1);
      \draw[lane] (0,0) node[left, font=\scriptsize]{$\pi'$} -- (8.3,0);
      \draw[dashed, gray] (8,-0.3) -- (8,1.35) node[above, font=\scriptsize, text=black]{$u$};
      \fill[nogo] (2.5,-0.18) rectangle (5.5,0.18);
      \fill[orange!45] (3.5,-0.18) rectangle (5.5,0.18);
      \node[wp, label={[font=\scriptsize]above:$y_x$}] at (2.5,1) {};
      \node[wq, label={[font=\scriptsize]below:$y_w$}] at (5.5,0) {};
      \node[wq] at (0.8,0) {};
      \draw[scan] (8,1.85) -- node[above, font=\scriptsize]{$\mathit{Ph}_1 \HSince$} (5.5,1.85);
      \draw[scan] (5.5,1.85) -- node[above, font=\scriptsize]{$\mathit{Ph}_2 \HSinceS$} (2.5,1.85);
      \draw[dist] (5.5,-0.75) -- node[below=-1pt, font=\tiny]{(ii) $w < a$} (8,-0.75);
      \draw[dist] (2.5,-1.45) -- node[below=-1pt, font=\tiny]{(i) $x \leq a{+}L$} (8,-1.45);
      \draw[dist] (3.5,-0.75) -- node[below=-1pt, font=\tiny]{(ii) gap $w{+}J$} (5.5,-0.75);
      \node[font=\tiny, text=gray!40!black] at (4.0,0.42) {no $W_{\pi'}$ in $(y_x,y_w)$};
    \end{tikzpicture}
    \caption{$\mathit{Sep}_A$ (translates, $w$ exists)}
  \end{subfigure}\hfill
  \begin{subfigure}{0.49\textwidth}
    \centering
    \begin{tikzpicture}[x=0.6cm, y=0.6cm]
      \draw[lane] (0,1) node[left, font=\scriptsize]{$\pi$} -- (8.3,1);
      \draw[lane] (0,0) node[left, font=\scriptsize]{$\pi'$} -- (8.3,0);
      \draw[dashed, gray] (8,-0.3) -- (8,1.35) node[above, font=\scriptsize, text=black]{$u$};
      \fill[nogo] (4,-0.18) rectangle (8,0.18);
      \node[wp, label={[font=\scriptsize]above:$y_x$}] at (5.2,1) {};
      \node[wq] at (1.5,0) {};
      \draw[scan] (8,1.85) -- node[above, font=\scriptsize]{$\mathit{Ph}_2 \HSince$} (5.2,1.85);
      \draw[dist] (5.2,-0.75) -- node[below=-1pt, font=\tiny]{(i) $x \leq a{+}L$} (8,-0.75);
      \draw[dist] (4,-1.45) -- node[below=-1pt, font=\tiny]{(iii) no $W_{\pi'}$ at distance in $K$} (8,-1.45);
      \node[font=\tiny, text=gray!40!black] at (6.6,0.42) {no $W_{\pi'}$};
    \end{tikzpicture}
    \caption{$\mathit{Sep}_B$ (translates, no $w$)}
  \end{subfigure}

  \vspace{2mm}
  \begin{subfigure}{0.49\textwidth}
    \centering
    \begin{tikzpicture}[x=0.6cm, y=0.6cm]
      \draw[lane] (0,1) node[left, font=\scriptsize]{$\pi$} -- (8.3,1);
      \draw[lane] (0,0) node[left, font=\scriptsize]{$\pi'$} -- (8.3,0);
      \draw[dashed, gray] (8,-0.3) -- (8,1.35) node[above, font=\scriptsize, text=black]{$u$};
      \fill[nogo] (6,-0.18) rectangle (8,0.18);
      \node[wp, label={[font=\scriptsize]above:$y_x$}] at (6,1) {};
      \node[wq] at (4.8,0) {};
      \draw[scan] (8,1.85) -- node[above, font=\scriptsize]{$\mathit{Ph}_2 \HSince$} (6,1.85);
      \draw[dist] (6,-0.75) -- node[below=-1pt, font=\tiny]{$\min D \leq x < L$} (8,-0.75);
      \node[font=\tiny, text=gray!40!black] at (7,0.42) {no $W_{\pi'}$};
    \end{tikzpicture}
    \caption{$\mathit{Sep}_P$ (prefixes)}
  \end{subfigure}\hfill
  \begin{subfigure}{0.49\textwidth}
    \centering
    \begin{tikzpicture}[x=0.6cm, y=0.6cm]
      \draw[lane] (0,1) node[left, font=\scriptsize]{$\pi$} -- (8.3,1);
      \draw[lane] (0,0) node[left, font=\scriptsize]{$\pi'$} -- (8.3,0);
      \draw[dashed, gray] (8,-0.3) -- (8,1.35) node[above, font=\scriptsize, text=black]{$u$};
      \fill[nogo] (0,-0.18) rectangle (8,0.18);
      \node[wp] at (3,1) {};
      \node[font=\tiny, text=gray!40!black] at (4,0.42) {no live $W_{\pi'}$};
      \node[font=\scriptsize, align=center] at (4,-0.85) {$\HLift{\alpha \Since \beta}{\pi} \HConj \HNeg \HLift{\alpha \Since \beta}{\pi'}$ at $u$};
    \end{tikzpicture}
    \caption{$\mathit{Sep}_{U,1}$ (unbounded, $D' = \emptyset$)}
  \end{subfigure}

  \vspace{2mm}
  \begin{subfigure}{0.49\textwidth}
    \centering
    \begin{tikzpicture}[x=0.6cm, y=0.6cm]
      \draw[lane] (0,1) node[left, font=\scriptsize]{$\pi$} -- (8.3,1);
      \draw[lane] (0,0) node[left, font=\scriptsize]{$\pi'$} -- (8.3,0);
      \draw[dashed, gray] (8,-0.3) -- (8,1.35) node[above, font=\scriptsize, text=black]{$u$};
      \fill[nogo] (0,-0.18) rectangle (4.8,0.18);
      \node[wp, label={[font=\scriptsize]above:$y_x$}] at (2.6,1) {};
      \node[wq, label={[font=\scriptsize]below:$y_w$}] at (4.8,0) {};
      \node[wq] at (6.5,0) {};
      \draw[scan] (8,1.85) -- node[above, font=\scriptsize]{$\mathit{Ph}_1 \HSince$} (4.8,1.85);
      \draw[scan] (4.8,1.85) -- node[above, font=\scriptsize]{$\Lambda_\pi \HSinceS$} (2.6,1.85);
      \draw[dist] (4.8,-1.1) -- node[below=-1pt, font=\tiny]{$\max D' < a$} (8,-1.1);
      \draw[dist] (2.6,-1.75) -- node[below=-1pt, font=\tiny]{$\max D > \max D'$} (8,-1.75);
      \node[font=\tiny, text=gray!40!black] at (2.4,0.42) {no older $W_{\pi'}$};
    \end{tikzpicture}
    \caption{$\mathit{Sep}_{U,2}$ (unbounded, $D' \neq \emptyset$)}
  \end{subfigure}\hfill
  \begin{subfigure}{0.49\textwidth}
    \centering
    \begin{tikzpicture}[x=0.6cm, y=0.6cm]
      \node[wp] at (0,1.2) {}; \node[right, font=\scriptsize] at (0.3,1.2) {live witness of $\alpha \Since \beta$ on $\pi$ ($W_\pi$)};
      \node[wq] at (0,0.4) {}; \node[right, font=\scriptsize] at (0.3,0.4) {live witness on $\pi'$ ($W_{\pi'}$)};
      \fill[nogo] (-0.3,-0.58) rectangle (0.3,-0.22); \node[right, font=\scriptsize] at (0.3,-0.4) {no witness of $\pi'$};
      \draw[scan] (0.3,-1.2) -- (-0.3,-1.2); \node[right, font=\scriptsize] at (0.3,-1.2) {backward scan of $\HSince$ / $\HSinceS$};
    \end{tikzpicture}
  \end{subfigure}
  \caption{Intuition for the formulae of Lemma~\ref{lemma:rewriting}, which implement the conditions of Lemma~\ref{lemma:separation} (Figure~\ref{fig:separation}) on instants: time flows to the right, $u$ is the evaluation instant (frozen in $t$), and the distance of an instant $y$ is $u - y$.
  The formulae scan backwards from $u$ with the outer since operators, and $t'.\,(t - t' \bowtie c)$ checks the distance of the instant reached; $\Lambda_\pi$ and $\Lambda_{\pi'}$ (in $\mathit{Ph}_1$, $\mathit{Ph}_2$) ensure that the scanned witnesses stay live.
  (a)~$\mathit{Sep}_A$ first finds the witness $y_w$ of $\pi'$ ($w = u - y_w$, with $w < a$), checks with the lift $\HLift{\alpha \Since_J \beta}{\pi'}$ at $y_w$ that no witness of $\pi'$ lies in the gap $w + J$, and then finds an older witness $y_x$ of $\pi$ ($x = u - y_x \leq a + L$) with no witness of $\pi'$ in between ($\mathit{Ph}_2$); the second disjunct of $\mathit{Next}$ and the $\mathit{Brk}$ branch of $\mathit{Rest}$ handle a violation of $\alpha$ on $\pi'$, after which older witnesses of $\pi'$ are dead.
  (b)~$\mathit{Sep}_B$ finds a witness $y_x$ of $\pi$ with no witness of $\pi'$ in $[y_x, u]$, and checks at $u$ that $\pi'$ has no witness at distance in $K = [0,L\rangle$.
  (c)~$\mathit{Sep}_P$ is like $\mathit{Sep}_B$ with the distance constraint $u - y_x < L$.
  (d)~$\mathit{Sep}_{U,1}$ only needs the lifts of $\alpha \Since \beta$ at $u$.
  (e)~$\mathit{Sep}_{U,2}$ finds the oldest live witness $y_w$ of $\pi'$ (no older one: $\lnot(\alpha \Since_{(0,\infty)} \beta)$ at $y_w$), checks $u - y_w < a$, and finds an older witness of $\pi$ by $\Lambda_\pi \HSinceS W_\pi$ (a lift at $y_w$ would require $\pi$ to have a time-point at $y_w$).}
  \label{fig:sep formulae}
\end{figure}

\paragraph*{Formulae.}
Let $\Tilde{R}_X \equiv t'.\;(t - t' \leq a + L)$ if $\rangle = {]}$, and $\Tilde{R}_X \equiv t'.\;(t - t' < a + L)$ otherwise, conjoined with $t'.\;(t - t' > 0)$ if $\langle = {(}$.
Let $\Tilde{R}_W \equiv t'.\;(t - t' < a)$ if $\langle = {[}$ and $\Tilde{R}_W \equiv t'.\;(t-t' \leq a)$ otherwise.
In these formulae, $t$ will always be frozen at the evaluation instant $u$, so that $t - t'$ is the distance of the current instant from $u$.
Define:
\begin{align*}
  \mathit{Ph}_1 &\equiv \Lambda_\pi \HConj \Lambda_{\pi'}, \qquad
  \mathit{Ph}_2 \equiv \Lambda_\pi \HConj \Lambda_{\pi'} \HConj \HNeg W_{\pi'}, \qquad
  \mathit{Brk} \equiv \HNeg\Lambda_{\pi'} \HConj \Lambda_\pi \HConj \HNeg W_{\pi'},\\
  X &\equiv W_\pi \HConj \HNeg W_{\pi'} \HConj \Tilde{R}_X, \qquad
  X' \equiv W_\pi \HConj \Tilde{R}_X, \qquad
  \mathit{Rest} \equiv X \HDisj \bigl(\mathit{Brk} \HConj (\Lambda_\pi \HSinceS X')\bigr),\\
  \mathit{Sep}_A &\equiv t.\;\mathit{Ph}_1 \HSince \Bigl(\HLift{(\alpha \Since_{[0,0]} \beta) \land \lnot(\alpha \Since_J \beta)}{\pi'} \HConj \Lambda_\pi \HConj \Tilde{R}_W \HConj \mathit{Next}\Bigr),\\
  \mathit{Next} &\equiv \bigl(\Lambda_{\pi'} \HConj (\mathit{Ph}_2 \HSinceS \mathit{Rest})\bigr) \HDisj \bigl(\HNeg\Lambda_{\pi'} \HConj (\Lambda_\pi \HSinceS X')\bigr),\\
  \mathit{Sep}_B &\equiv t.\;\HNeg \HLift{\alpha \Since_{K} \beta}{\pi'} \HConj \bigl(\mathit{Ph}_2 \HSince \mathit{Rest}\bigr),\\
  \mathit{Sep}_P &\equiv t.\;\mathit{Ph}_2 \HSince \bigl(X_P \HDisj (\mathit{Brk} \HConj (\Lambda_\pi \HSinceS X_P'))\bigr),\\
  X_P &\equiv W_\pi \HConj \HNeg W_{\pi'} \HConj t'.\;(t-t' < L), \qquad X'_P \equiv W_\pi \HConj t'.\;(t-t' < L),
\end{align*}
where $K = [0,L]$ if $\rangle = {]}$, and $K = [0,L)$ otherwise.
For bounded $I$, let $\Phi_{\alpha,\beta,I} \equiv \tforall \pi,\pi'.\;\HNeg(\mathit{Sep}_A \HDisj \mathit{Sep}_B \HDisj \mathit{Sep}_P)$.
For unbounded $I = \langle a,\infty)$, let $\Phi_{\alpha,\beta,I} \equiv \tforall \pi,\pi'.\;\HNeg(\mathit{Sep}_{U,1} \HDisj \mathit{Sep}_{U,2})$, where
\begin{align*}
  \mathit{Sep}_{U,1} &\equiv \HLift{\alpha \Since \beta}{\pi} \HConj \HNeg \HLift{\alpha \Since \beta}{\pi'},\\
  \mathit{Sep}_{U,2} &\equiv t.\;\mathit{Ph}_1 \HSince \bigl(\Tilde{R}_W \HConj \HLift{(\alpha \Since_{[0,0]} \beta) \land \lnot(\alpha\Since_{(0,\infty)}\beta)}{\pi'} \HConj \Lambda_\pi \HConj (\Lambda_\pi \HSinceS W_\pi)\bigr).
\end{align*}
Here, $\mathit{Sep}_{U,2}$ checks the existence of a witness of $\pi$ strictly before the last witness $w$ of $\pi'$ by $\Lambda_\pi \HSinceS W_\pi$, and not by a lift $\HLift{\alpha \Since_{(0,\infty)} \beta}{\pi}$ evaluated at the instant of $w$: the trace $\pi$ need not have a time-point at that instant (e.g., for $I = [6,\infty)$, $D_\pi = \{10\}$, and $D_{\pi'} = \{5\}$, where $\pi$ has time-points only at $0$ and $10$, the window $[6,\infty)$ separates $\pi$ from $\pi'$).
Figure~\ref{fig:sep formulae} illustrates how these formulae work.

\begin{customlemma}{\ref{lemma:rewriting}}\leanok{Size.lean}{230}{lemma_rewriting}%
  For any non-empty interval $I$ and pMTL formulae $\alpha$ and $\beta$, there exists an anchored 1-HyperpTPTL formula $\Phi_{\alpha,\beta,I}$
  such that
  \begin{align*}
    \Phi_{\alpha,\beta,I} && \text{and} && \tforall \pi,\pi'.\bm{\bigwedge}_{d \in I \cup [0,\inf I]} \left(\{\alpha \Since_{I-d} \beta\}_\pi \HImp \{\alpha \Since_{I-d} \beta\}_{\pi'}\right)
  \end{align*}
  are semantically equivalent and $|\Phi_{\alpha,\beta,I}| = O(|\alpha| + |\beta| + |I|)$.
\end{customlemma}
\begin{proof}%
Let $\Phi_{\alpha,\beta,I}$ be defined as above.

\emph{Syntactic properties.}
The formula has size $O(|\alpha| + |\beta| + |I|)$.
It is in 1-HyperpTPTL: $t$ is frozen once, at the evaluation instant, and every occurrence of $t'$ is of the form $t'.\;C$; all lifts are closed.
It is anchored: all left operands of $\HSince$ and $\HSinceS$ ($\mathit{Ph}_1$, $\mathit{Ph}_2$, $\Lambda_\pi$) are conjunctions of negated lifts, hence lax, and all right operands are conjunctions with a lift, or disjunctions thereof, hence strict.

\emph{Equivalence.}
Since $\{\alpha \Since_{I-d}\beta\}_\pi \HImp \{\alpha \Since_{I-d}\beta\}_{\pi'}$ holds for all $d \in I \cup [0,\inf I]$ iff no window separates $D_\pi$ from $D_{\pi'}$ (windows for other $d$ are empty), it suffices to show that the body of $\Phi_{\alpha,\beta,I}$ fails iff some window separates $D_\pi$ from $D_{\pi'}$.
We show that $\mathit{Sep}_A \HDisj \mathit{Sep}_B$ holds iff some translate separates $D_\pi$ from $D_{\pi'}$, using the characterization of Lemma~\ref{lemma:separation}; the arguments for $\mathit{Sep}_P$ (prefixes) and $\mathit{Sep}_{U,1} \HDisj \mathit{Sep}_{U,2}$ (unbounded case) are analogous and simpler.
By anchoring, all outer temporal operators only need to be considered at instants of $T(\Pi)$.

($\Leftarrow$) Let $x$ and $w$ be as in Lemma~\ref{lemma:separation}, and let $y_x = u - x$ and $y_w = u - w$ be the corresponding instants.
Suppose first that $w$ exists.
Then $y_w$ is a witness of $\pi'$, so $\Lambda_{\pi'}$ holds on $(y_w,u]$; $y_x$ is a witness of $\pi$, so $\Lambda_\pi$ holds on $(y_x,u] \supseteq (y_w,u]$.
Hence $\mathit{Ph}_1$ holds on $(y_w,u]$, and at $y_w$, $W_{\pi'}$, $\Lambda_\pi$, and $\Tilde{R}_W$ hold by~(ii); by~(ii) and~\eqref{eq:witness} applied at $y_w$, $\lnot(\alpha \Since_J \beta)$ holds there in $\pi'$.
It remains to show that $\mathit{Next}$ holds at $y_w$.
If $\HNeg\Lambda_{\pi'}$ holds at $y_w$, the second disjunct of $\mathit{Next}$ holds with witness $y_x$.
Otherwise, $\Lambda_{\pi'}$ holds at $y_w$, and we show that $\mathit{Ph}_2 \HSinceS \mathit{Rest}$ holds there.
If $\Lambda_{\pi'}$ holds on $(y_x,y_w)$, then no instant in $(y_x,y_w)$ satisfies $W_{\pi'}$, since such an instant would be a witness of $\pi'$ between $w$ and $x$, contradicting the maximality of $w$; similarly, $W_{\pi'}$ fails at $y_x$ as $x \notin D'$, and $X$ holds at $y_x$ by~(i).
Otherwise, let $g \in (y_x,y_w)$ be the greatest instant at which $\HNeg\Lambda_{\pi'}$ holds; then $\mathit{Ph}_2$ holds on $(g,y_w)$ as before, and either $g$ is not a witness (so $\mathit{Brk}$ holds at $g$, and $\Lambda_\pi \HSinceS X'$ holds at $g$ with witness $y_x$), or $W_{\pi'}$ holds at $g$, in which case $g$ is a witness of $\pi'$ between $w$ and $x$, a contradiction.
The case where $w$ does not exist is identical, using~(iii) for the first conjunct of $\mathit{Sep}_B$ and a non-strict $\HSince$ starting at $u$.

($\Rightarrow$) Suppose that $\mathit{Sep}_A$ holds, and let $y_w$ be the instant at which the right operand of the outer $\HSince$ holds, and $y_x$ the instant at which $X$ or $X'$ holds.
As $\mathit{Ph}_1$ holds on $(y_w,u]$ and $W_{\pi'}$ holds at $y_w$, $y_w$ is a witness of $\pi'$; let $w = u - y_w$.
As $\Lambda_\pi$ holds on $(y_x,u]$ and $W_\pi$ holds at $y_x$, $y_x$ is a witness of $\pi$; let $x = u - y_x \in D_\pi$, with $x > w$.
If the second disjunct of $\mathit{Next}$ holds, then $\HNeg\Lambda_{\pi'}$ holds at $y_w$, so no instant before $y_w$ is a witness of $\pi'$; hence $x \notin D_{\pi'}$ and $w = \max(D_{\pi'} \cap [0,x))$.
Otherwise $\Lambda_{\pi'}$ holds at $y_w$.
If $X$ holds at $y_x$, then $\Lambda_{\pi'}$ holds on $(y_x,u]$ and $W_{\pi'}$ fails at $y_x$ and on $(y_x,y_w)$, so $x \notin D_{\pi'}$ and $w = \max(D_{\pi'} \cap [0,x))$.
If the $\mathit{Brk}$ branch is taken at $g \in (y_x,y_w)$, then no instant at or before $g$ is a witness of $\pi'$, as $\Lambda_{\pi'}$ fails at $g$, and no instant in $[g,y_w)$ is one either, so again $x \notin D_{\pi'}$ and $w = \max(D_{\pi'} \cap [0,x))$.
Conditions~(i) and~(ii) follow from $\Tilde{R}_X$, $\Tilde{R}_W$, and the lift at $y_w$ together with~\eqref{eq:witness}.
For $\mathit{Sep}_B$, the same reasoning shows that $D_{\pi'} \cap [0,x) = \emptyset$, and~(iii) follows from the first conjunct by~\eqref{eq:witness}.
\end{proof}

\paragraph*{Discrete time.}\leanok{UpperBound.lean}{264}{nat_window_iff}
Over discrete time, all distances are natural numbers, so $\HLift{\alpha \Since_K \beta}{\varpi}$ only depends on $K \cap \Nat$, and only $d \in \Nat$ is relevant (Appendix~\ref{apx:discrete}).
We may assume $I \cap \Nat \neq \emptyset$, as otherwise $I$ is equivalent to $\emptyset$ (Appendix~\ref{apx:discrete}).
Let $a_I = \min(I \cap \Nat)$ and $b_I = \sup(I \cap \Nat) \in \Nat \cup \{\infty\}$, so that $I \cap \Nat = \{n \in \Nat \mid a_I \leq n \leq b_I\}$, and let $\hat{I} = [a_I, b_I + 1)$, where $\infty + 1 = \infty$.
Then $\hat{I}$ is a non-empty real interval, and we use $\Phi_{\alpha,\beta,\hat{I}}$ without change.
Indeed, for every real $d \geq 0$, we have $(\hat{I} - d) \cap \Nat = (I - \lfloor d \rfloor) \cap \Nat$, and $d \in \hat{I} \cup [0,a_I]$ iff $\lfloor d \rfloor \in \Nat \cap (I \cup [0,\inf I])$.
Hence, over discrete time, the conjunction of Lemma~\ref{lemma:rewriting} for $\hat{I}$ and the conjunction over $d \in \Nat \cap (I \cup [0,\inf I])$ have the same conjuncts up to equivalence.
As $\Phi_{\alpha,\beta,\hat{I}}$ is anchored, it does not matter that instants range over $\Nat$ rather than $\Re$ (see the remark after Definition~\ref{def:anchored}).

\subsection{Decision Procedure}\label{subsecappx:decision procedure}

Let $\chi$ be an anchored HyperpTPTL formula of the form $\texists \pi_1,\dots,\pi_n.\; \tforall \pi'_1,\dots,\pi'_m.\; \chi'$, where $n \geq 1$, $m \geq 1$, and $\chi'$ does not contain any trace quantifiers.
Recall that $\tilde\chi$ is constructed as
\[
  \tilde\chi \equiv \bigwedge_{i_1 \in \{1,\dots,n\}} \cdots \bigwedge_{i_m \in \{1,\dots,n\}} \chi'[\pi'_1 \mapsto \pi_{i_1}]\cdots[\pi'_m \mapsto \pi_{i_m}].
\]
For an assignment $\Pi$ to variables $X$, let $\img(\Pi)$ denote the image of $X$ under $\Pi$.

\begin{lemma}\label{lemma:hsat ea iff}\leanok{Main.lean}{77}{main_lemma_A8}
  Suppose that $\chi'$ is anchored. Then
  $\chi$ is satisfiable iff there exists an assignment $\Pi$ of traces to the trace variables $\{\pi_1,\dots,\pi_n\}$ such that $\img(\Pi)$ is aligned and \[\Tuple{\img(\Pi),\Pi,\tms(\img(\Pi)),\emptyset} \hypermodels \tilde\chi.\]
\end{lemma}
\begin{proof}
  \fbox{$\Longrightarrow$}
  Suppose that $S \hypermodels \chi$. Since $n \geq 1$, then $S$ is non-empty, and since $\chi$ is anchored, $S$ is aligned.
  Then there exists an assignment $\Pi$ to $\pi_1,\dots,\pi_n$ such that $\Tuple{S,\Pi,\tms(S),\emptyset} \hypermodels \tforall \pi'_1,\dots,\pi'_m.\;\chi'$.
  Note that $\img(\Pi) \subseteq S$, so $\img(\Pi)$ is aligned, and it is easy to see that $\Tuple{\img(\Pi),\Pi,\tms(S),\emptyset} \hypermodels \tforall \pi'_1,\dots,\pi'_m.\;\chi'$.
  Since $\img(\Pi) = \{\Pi(\pi_1),\dots,\Pi(\pi_n)\}$ is a finite set, we can expand the universal quantifiers by enumerating all possible assignments of traces to the variables $\pi'_j$, and each available trace is identified by one of the variables $\pi_i$.
  This expansion results in $\tilde\chi$ and hence $\Tuple{\img(\Pi),\Pi,\tms(S),\emptyset} \hypermodels \tilde\chi$.
  It remains to replace $\tms(S)$ by $\tms(\img(\Pi))$.
  If some trace in $\img(\Pi)$ is non-empty, then $\tms(\img(\Pi)) = \tms(S)$, as $S$ is aligned.
  Otherwise, $T(\Pi) = \emptyset$ and $\tms(\img(\Pi)) = 0$, which may differ from $\tms(S)$.
  In this case, we show that the truth value of $\tilde\chi$, which is anchored as $\chi'$ is, does not depend on the evaluation instant.
  More precisely, we show by induction that for every subformula $\psi$ of $\tilde\chi$, all instants $u,u' \in \Re$, and all valuations $v,v'$ with the same domain such that $v'(x) = v(x) + (u'-u)$ for all $x$ in this domain,
  \[
    \Tuple{\img(\Pi),\Pi,u,v} \hypermodels \psi \iff \Tuple{\img(\Pi),\Pi,u',v'} \hypermodels \psi.
  \]
  Lifts are false at every instant, as all traces in $\img(\Pi)$ are empty.
  For $\psi \equiv \psi_1 \HSince \psi_2$ and $\psi \equiv \psi_1 \HSinceS \psi_2$, the right operand $\psi_2$ is strict, as $\tilde\chi$ is anchored; hence $\psi_2$ is false at every instant not in $T(\Pi) = \emptyset$, i.e., at every instant, and $\psi$ is false at both $u$ and $u'$.
  For $\psi \equiv t.\;\psi_1$, the valuations $v(t \mapsto u)$ and $v'(t \mapsto u')$ again satisfy the hypothesis, and the claim follows from the induction hypothesis for $\psi_1$.
  For a constraint $t_1 - t_2 \leq c$, we have $v'(t_1) - v'(t_2) = v(t_1) - v(t_2)$.
  The cases $\HTrue$, $\HNeg$, and $\HDisj$ follow directly from the induction hypothesis.
  Applying this to $\psi \equiv \tilde\chi$ and $v = v' = \emptyset$ yields $\Tuple{\img(\Pi),\Pi,\tms(\img(\Pi)),\emptyset} \hypermodels \tilde\chi$ in this case as well.

  \fbox{$\Longleftarrow$}
  Suppose that $\Tuple{\img(\Pi),\Pi,\tms(\img(\Pi)),\emptyset} \hypermodels \tilde\chi$.
  By the converse of the expansion argument given above, we have $\Tuple{\img(\Pi),\Pi,\tms(\img(\Pi)),\emptyset} \hypermodels \tforall \pi'_1,\dots,\pi'_m.\;\chi'$
  and hence $\Tuple{\img(\Pi),\emptyset,\tms(\img(\Pi)),\emptyset} \hypermodels \chi$, i.e., $\img(\Pi) \hypermodels \chi$.
\end{proof}

\begin{definition}\label{def:interleaving}\leanok{Interleave.lean}{107}{IsInter}
  Let $S$ be a \emph{finite}, non-empty set of aligned traces.
  Let $\{\pi_1,\dots,\pi_n\}$ be a set of trace variables, and let $\Pi$ map each $\pi_k$ to an element of $S$ such that $\img(\Pi) = S$.
  We say that the non-empty trace $s$ is an \emph{interleaving} of $\Pi$ if the following conditions are satisfied.
  \begin{compactenum}[(1)]
    \item If some trace in $S$ is non-empty, the last timestamp of $s$ is equal to $\tms(S)$.
    \item A~time-point $i$ of $s$ satisfies $\Tick \in \ev_i$ iff either $i = |s| - 1$ or $\tms_{i+1} > \tms_i$.
    \item No time-point carrying $\Tick$ carries any $\TP_{\pi_k}$.
    \item For every $k$, $\Pi(\pi_k)$ is equal to the \emph{projection} $\mathrm{prj}_{\pi_k}(s)$, i.e., the trace consisting of exactly those time-points $i$ of $s$ with $\TP_{\pi_k} \in \ev_i$, in the same order, with event sets $\{P \mid P_{\pi_k} \in \ev_i, P \neq \TP\}$ and with the same \emph{absolute} timestamps $\tms_i$ as in $s$ (the delays of the projection are recomputed accordingly).
  \end{compactenum}
\end{definition}

We write $C(s)$ for the set of timestamps of the time-points of $s$ carrying $\Tick$; by~(2), each timestamp of $s$ occurs exactly once in $C(s)$, as the timestamp of the last time-point of its block of equal timestamps.
By~(4), $T(\Pi) \subseteq C(s)$.

Recall the pTPTL formula $\zeta$:
\[
  \zeta \equiv \Tick \land \HAlways\bigl(\Prev \Tick \Iff t.\;\Prev(t'.\;\lnot(t-t' \leq 0))\bigr) \land \HAlways\bigl(\Tick \Imp \textstyle\bigwedge_{k} \lnot \TP_{\pi_k}\bigr) \land
  \bigwedge_{i\in\{1,\dots,n\}} (\Once \TP_{\pi_i} \Imp \Once_{[0,0]} \TP_{\pi_i}).
\]

\begin{lemma}\label{lemma:zeta}\leanok{Project.lean}{128}{isInter_of_zeta}
  If the non-empty trace $s$ satisfies $\zeta$ at its last time-point, it is an interleaving of the assignment $\Pi$ mapping each $\pi_k$ to $\mathrm{prj}_{\pi_k}(s)$, and $\img(\Pi)$ is aligned.
\end{lemma}
\begin{proof}
  Condition~(4) holds by definition of $\Pi$, and condition~(3) by the third conjunct of $\zeta$.
  The last conjunct of $\zeta$ states that every $\pi_k$ that has a time-point in $s$ has one with the last timestamp of $s$; hence all non-empty projections end with the last timestamp of $s$, which shows that $\img(\Pi)$ is aligned and condition~(1).
  For condition~(2), the last time-point carries $\Tick$ by the first conjunct.
  For $i < |s|-1$, the second conjunct, evaluated at $i+1$, states that $\Tick \in \ev_i$ iff $\tms_{i+1} - \tms_i > 0$.
\end{proof}

\newcommand*{\Interleave}[1]{#1^\propto}
\begin{lemma}\label{lemma:ex interleaving}\leanok{Construct.lean}{394}{exists_zeta_model}
  For every assignment $\Pi$ of traces to $\pi_1,\dots,\pi_n$ such that $\img(\Pi)$ is aligned, there exists an interleaving $\Interleave{\Pi}$ of $\Pi$ satisfying $\zeta$ at its last time-point.
\end{lemma}
\begin{proof}
  Let $S = \img(\Pi)$ and $s^k = \Pi(\pi_k) = (\ev^k,\tmsdiff^k)$.
  Let $u_1 < \dots < u_m$ be the elements of $T(\Pi)$; if $T(\Pi) = \emptyset$ (all traces are empty), let $m = 1$ and $u_1 = 0$.
  For every $\ell \in \{1,\dots,m\}$ and $k \in \{1,\dots,n\}$, let $\Interleave{\Pi}_{\ell,k}$ be the sequence of events $(\{\TP_{\pi_k}\} \cup \{P_{\pi_k} \mid P \in \ev^k_x\},0)$ for all time-points $x$ of $s^k$ with timestamp $u_\ell$, in order (possibly empty).
  Define
  \[
    \Interleave{\Pi}_\ell := (\emptyset, u_\ell - u_{\ell-1}) \cdot \Interleave{\Pi}_{\ell,1} \cdots \Interleave{\Pi}_{\ell,n} \cdot (\{\Tick\},0),
  \]
  where $u_0 = 0$, and let $\Interleave{\Pi} = \Interleave{\Pi}_1 \cdots \Interleave{\Pi}_m$.
  All time-points of $\Interleave{\Pi}_\ell$ have timestamp $u_\ell$, and the timestamps strictly increase between blocks.
  Hence condition~(2) holds, conditions~(3) and~(4) hold by construction, and condition~(1) holds because $u_m = \tms(S)$ when some trace is non-empty.
  Finally, $\Interleave{\Pi}$ satisfies $\zeta$: the first three conjuncts follow from~(2) and~(3), and the last one from the alignment of $S$.
  Note that all timestamps of $\Interleave{\Pi}$ belong to $T(\Pi) \cup \{0\}$, so the construction stays within the discrete time domain if the traces in $S$ are discrete.
\end{proof}

\begin{lemma}\label{lemma:red aux}\leanok{Interleave.lean}{190}{sat_rhoT}
  Let $s$ be an interleaving of $\Pi$, and let $\pi \in \dom \Pi$.
  Then $\Tuple{\mathrm{prj}_\pi(s),i,v} \models \phi$ iff $\Tuple{s,j,v} \models \Red_\pi(\phi)$, where $j$ is the time-point of the $i$-th event in $s$ satisfying $\TP_\pi$.
\end{lemma}
\begin{proof}
  Let $s = (\ev,\tmsdiff)$ and $\mathrm{prj}_\pi(s) = (\ev^\pi,\tmsdiff^\pi)$. 
  By induction on the structure of $\phi$, generalizing $i$ and $v$.
  \begin{itemize}
    \item $\phi \equiv \Red_\pi(\phi) \equiv \True$: Trivial.
    \item $\phi \equiv \Red_\pi(\phi) \equiv t_1 - t_2 \leq c$: Trivial.
    \item $\phi \equiv P$, hence $\Red_\pi(\phi) \equiv P_\pi$:
      Observe that $\ev^\pi_i = \{Q \mid Q_\pi \in \ev_j, Q \neq \TP\}$.
      Hence
      \[
        \Tuple{\mathrm{prj}_\pi(s),i,v} \models P \iff P \in \ev^\pi_i \iff P_\pi \in \ev_j \iff \Tuple{s,j,v} \models P_\pi.
      \]
    \item $\phi \equiv \lnot \phi'$, hence $\Red_\pi(\phi) \equiv \lnot \Red_\pi(\phi)$:
      Follows directly from the induction hypothesis (I.H.).
    \item $\phi \equiv \phi_1 \lor \phi_2$, hence $\Red_\pi(\phi) \equiv \Red_\pi(\phi_1) \lor \Red_\pi(\phi_2)$:
      Follows directly from the I.H.
    \item $\phi \equiv \Prev \phi'$, hence $\Red_\pi(\phi) \equiv \Prev \left( \lnot \TP_\pi \Since (\TP_\pi \land \Red_\pi(\phi'))\right)$:
      Suppose that $\Tuple{\mathrm{prj}_\pi(s),i,v} \models \Prev\phi'$.
      Therefore, $i > 0$ and $\Tuple{\mathrm{prj}_\pi(s),i-1,v} \models \phi'$.
      Using the I.H., we have $\Tuple{s,j',v} \models \Red_\pi(\phi')$, where $j'$ is the time-point of the $(i-1)$-th event in $s$ satisfying $\TP_\pi$.
      It follows that $j' < j$, $\Tuple{s,j',v} \models \TP_\pi$, and $\Tuple{s,k,v} \not\models \TP_\pi$ for all $k \in \{j'+1,\dots,j-1\}$.
      Hence $\Tuple{s,j,v} \models \Red_\pi(\phi)$.

      Now consider the converse direction.
      If $\Tuple{s,j,v} \models \Red_\pi(\phi)$, we know that there exists $j' < j$ where $\Tuple{s,j',v} \models \TP_\pi \land \Red_\pi(\phi')$ and, for all $k \in \{j'+1,\dots,j-1\}$, $\Tuple{s,k,v} \not\models \TP_\pi$.
      Clearly, $j'$ is the time-point of the $(i-1)$-th event in $s$ satisfying $\TP_\pi$.
      In particular, $j$ must be positive.
      Using the I.H., we have $\Tuple{s,i-1,v} \models \phi'$, which concludes the case.
    \item $\phi \equiv \phi_1 \Since \phi_2$, hence $\Red_\pi(\phi) \equiv (\TP_\pi \Imp \Red_\pi(\phi_1)) \Since (\TP_\pi \land \Red_\pi(\phi_2))$:
      Let us write $\pi(i)$ for the time-point in $s$ corresponding to time-point $i$ in $\mathrm{prj}_\pi(s)$, which is a monotone map.
      We reason
      \begin{align*}
        &\Tuple{\mathrm{prj}_\pi(s),i,v} \models \phi_1 \Since \phi_2 \\
        \iff &\exists i'\leq i.\; \Tuple{\mathrm{prj}_\pi(s),i',v} \models \phi_2 \land \forall k\in\{i'+1,\dots,i\}.\;\Tuple{\mathrm{prj}_\pi(s),k,v} \models \phi_1 \\
        \stackrel{\text{I.H.}}{\iff} &\exists i'\leq i.\; \Tuple{s,\pi(i'),v} \models \Red_\pi(\phi_2) \land \forall k\in\{i'+1,\dots,i\}.\;\Tuple{s,\pi(k),v} \models \Red_\pi(\phi_1) \\
        \stackrel{(*)}{\iff} &\begin{aligned}[t]
          \exists i'\leq i.\; &\Tuple{s,\pi(i'),v} \models \Red_\pi(\phi_2) \land{}\\
          &\forall k\in\{\pi(i')+1,\dots,j\}.\;\Tuple{s,k,v} \models (\TP_\pi \Imp \Red_\pi(\phi_1))
        \end{aligned}\\
        \stackrel{(**)}{\iff} &\begin{aligned}[t]
          \exists j'\leq j.\; &\Tuple{s,j',v} \models (\TP_\pi \land \Red_\pi(\phi_2)) \land{}\\
          &\forall k\in\{j'+1,\dots,j\}.\;\Tuple{s,k,v} \models (\TP_\pi \Imp \Red_\pi(\phi_1))
        \end{aligned}\\
        \iff &\Tuple{s,j,v} \models (\TP_\pi \Imp \Red_\pi(\phi_1)) \Since (\TP_\pi \land \Red_\pi(\phi_2)).
      \end{align*}
      The step $(*)$ is justified because the set $\{\pi(k) \mid k \in \{i'+1,\dots,i\}\}$ is equal to the set $\{k \in \{\pi(i')+1,\dots,\pi(i)\} \mid \TP_\pi \in \ev_k\}$, and $\pi(i) = j$.
      The step $(**)$ is similarly justified because $\TP_\pi \in \ev_{\pi(i')}$, and whenever $\TP_\pi \in \ev_{j'}$ for some $j' \leq \pi(i)$, there exists some $i' \leq i$ with $j' = \pi(i')$.
    \item $\phi \equiv t.\;\phi'$, hence $\Red_\pi(\phi) \equiv t.\;\Red_\pi(\phi)$:
      Follows from the I.H.\ by instantiating $v$ with $v(t \mapsto \tms_j)$, where $\tms_j = \tms^\pi_i$ by choice of $j$.
  \end{itemize}
\end{proof}

\begin{lemma}\label{lemma:red}\leanok{Interleave.lean}{197}{hsat_rho}
  Let $\chi$ be an anchored HyperpTPTL formula over trace variables $\{\pi_1,\dots,\pi_n\}$ and without trace quantifiers.
  Let $s = (\ev,\tmsdiff)$ be an interleaving of $\Pi$.
  Moreover, assume that $\Tuple{s,i,v} \models \Tick$.
  Then $\Tuple{\img(\Pi),\Pi,\tms_i,v} \hypermodels \chi$ iff $\Tuple{s,i,v} \models \Red(\chi)$.
\end{lemma}
\begin{proof}
  We write $\Tuple{\Pi,u,v} \hypermodels \chi$ instead of $\Tuple{\img(\Pi),\Pi,u,v} \hypermodels \chi$ for brevity.
  For $u \in C(s)$, let $c(u)$ denote the unique time-point of $s$ with timestamp $u$ carrying $\Tick$; $c$ is strictly monotone, and $c(\tms_i) = i$.
  The proof is by induction on the structure of $\chi$, generalizing $i$ and $v$; subformulae of anchored formulae are anchored.
  \begin{itemize}
    \item $\HTrue$, constraints, $\HNeg$, $\HDisj$: immediate (using the I.H.).
    \item $\chi \equiv t.\;\chi'$: follows from the I.H.\ with $v(t \mapsto \tms_i)$.
    \item $\chi \equiv \HLift{\phi}{\pi}$, hence $\Red(\chi) \equiv \Prev\bigl((\lnot \TP_\pi \land \lnot \Tick) \Since (\TP_\pi \land \Red_\pi(\phi))\bigr)$.
      By conditions~(2)--(4), the time-points of $\Pi(\pi)$ with timestamp $\tms_i$ correspond exactly to the time-points $j < i$ of $s$ that carry $\TP_\pi$ and such that no time-point in $[j,i)$ carries $\Tick$ (by~(2), these are exactly the time-points with timestamp $\tms_i$ other than $i$, and by~(3), $i$ itself carries no $\TP_\pi$).
      Suppose $\Tuple{\Pi,\tms_i,v} \hypermodels \HLift{\phi}{\pi}$, and let $x$ be the greatest time-point of $\Pi(\pi)$ with timestamp $\tms_i$, with $\Tuple{\Pi(\pi),x,v} \models \phi$.
      Let $j$ be the corresponding time-point of $s$.
      By Lemma~\ref{lemma:red aux}, $\Tuple{s,j,v} \models \Red_\pi(\phi)$, and by maximality of $x$, no time-point in $(j,i-1]$ carries $\TP_\pi$ or $\Tick$.
      Hence $\Tuple{s,i-1,v} \models (\lnot \TP_\pi \land \lnot \Tick) \Since (\TP_\pi \land \Red_\pi(\phi))$, and so $\Tuple{s,i,v} \models \Red(\chi)$.
      Conversely, if $\Tuple{s,i,v} \models \Red(\chi)$, then $i > 0$ and there is $j \leq i-1$ with $\TP_\pi \in \ev_j$ and $\Tuple{s,j,v} \models \Red_\pi(\phi)$ such that no time-point in $(j,i-1]$ carries $\TP_\pi$ or $\Tick$.
      By~(3), $j$ carries no $\Tick$ either, so $j$ has timestamp $\tms_i$ and corresponds to the greatest time-point of $\Pi(\pi)$ with timestamp $\tms_i$.
      By Lemma~\ref{lemma:red aux}, $\Tuple{\Pi,\tms_i,v} \hypermodels \HLift{\phi}{\pi}$.
    \item $\chi \equiv \chi_1 \HSince \chi_2$, hence $\Red(\chi) \equiv (\Tick \Imp \Red(\chi_1)) \Since (\Tick \land \Red(\chi_2))$.
      Since $\chi$ is anchored and $T(\Pi) \subseteq C(s)$, we have $\Tuple{\Pi,\tms_i,v} \hypermodels \chi$ iff there is $u \in C(s)$ with $u \leq \tms_i$ such that $\Tuple{\Pi,u,v} \hypermodels \chi_2$ and $\Tuple{\Pi,u',v} \hypermodels \chi_1$ for all $u' \in C(s) \cap (u,\tms_i]$.
      By the I.H.\ applied at the $\Tick$ time-points $c(u)$ and $c(u')$, and since $c$ is a monotone bijection between $C(s) \cap [0,\tms_i]$ and the $\Tick$ time-points $\leq i$, this is equivalent to the existence of a $\Tick$ time-point $j \leq i$ with $\Tuple{s,j,v} \models \Red(\chi_2)$ such that all $\Tick$ time-points in $(j,i]$ satisfy $\Red(\chi_1)$, i.e., to $\Tuple{s,i,v} \models \Red(\chi)$.
    \item $\chi \equiv \chi_1 \HSinceS \chi_2$, hence $\Red(\chi) \equiv \Prev\bigl((\Tick \Imp \Red(\chi_1)) \Since (\Tick \land \Red(\chi_2))\bigr)$.
      As in the previous case, but with $u < \tms_i$ and $u' \in C(s) \cap (u,\tms_i)$: the $\Tick$ time-points $< i$ are exactly those $\leq i - 1$, and they correspond to the instants of $C(s)$ strictly below $\tms_i$.
  \end{itemize}
\end{proof}

\begin{customlemma}{\ref{lemma:HyperpTPTL}}\leanok{Main.lean}{85}{main_lemma_19}
  Suppose that $\chi'$ is anchored. Then $\chi$ is satisfiable iff $\zeta \land \Red(\tilde\chi)$ is satisfiable. %
  The size of $\zeta \land \Red(\tilde\chi)$ is $O(n^m \cdot |\chi'|)$.
\end{customlemma}
\begin{proof}
  \fbox{$\Longrightarrow$}
  Suppose that $S \hypermodels \chi$.
  Using Lemma~\ref{lemma:hsat ea iff}, obtain an assignment $\Pi$ of traces to the trace variables $\{\pi_1,\dots,\pi_n\}$ such that $\img(\Pi)$ is aligned and \[\Tuple{\img(\Pi),\Pi,\tms(\img(\Pi)),\emptyset} \hypermodels \tilde\chi.\]
  We use Lemma~\ref{lemma:ex interleaving} to obtain an interleaving $\Interleave{\Pi} = (\ev,\tmsdiff)$ of $\Pi$ satisfying $\zeta$.
  Let $\ell = |\Interleave{\Pi}|-1$; note that $\tms_\ell = \tms(\img(\Pi))$ and $\Tuple{\Interleave{\Pi},\ell,\emptyset} \models \Tick$ by Definition~\ref{def:interleaving} and the construction.
  As $\tilde\chi$ is anchored, Lemma~\ref{lemma:red} yields
  \[
    \Tuple{\img(\Pi),\Pi,\tms(\img(\Pi)),\emptyset} \hypermodels \tilde\chi
    \iff \Tuple{\Interleave{\Pi},\ell,\emptyset} \models \Red(\tilde\chi),
  \]
  i.e., $\Interleave{\Pi}$ satisfies $\zeta \land \Red(\tilde\chi)$.

  \fbox{$\Longleftarrow$}
  Suppose that $s$ is a non-empty trace such that $\Tuple{s,|s|-1,\emptyset} \models \zeta \land \Red(\tilde\chi)$.
  Let $\Pi$ map $\pi_i$ to $\mathrm{prj}_{\pi_i}(s)$ for all $i \in \{1,\dots,n\}$.
  By Lemma~\ref{lemma:zeta}, $s$ is an interleaving of $\Pi$ and $\img(\Pi)$ is aligned.
  Moreover, $\Tuple{s,|s|-1,\emptyset} \models \Tick$, so Lemma~\ref{lemma:red} yields $\Tuple{\img(\Pi),\Pi,\tms_{|s|-1},\emptyset} \hypermodels \tilde\chi$.
  If some trace in $\img(\Pi)$ is non-empty, then $\tms_{|s|-1} = \tms(\img(\Pi))$ by condition~(1).
  Otherwise, $T(\Pi) = \emptyset$, and since $\tilde\chi$ is anchored, its truth value does not depend on the evaluation instant, as shown in the proof of Lemma~\ref{lemma:hsat ea iff}.
  In both cases, $\Tuple{\img(\Pi),\Pi,\tms(\img(\Pi)),\emptyset} \hypermodels \tilde\chi$.
  Hence $\Tuple{\img(\Pi),\Pi,\tms(\img(\Pi)),\emptyset} \hypermodels \tilde\chi$ and $\chi$ is satisfiable using Lemmata~\ref{lemma:red} and~\ref{lemma:hsat ea iff}.

  Size:
  The size of $\tilde\chi$ is $O(n^m \cdot |\chi'|)$.
  Observe that the translation $\Red$ adds a bounded number of symbols for every symbol in $\tilde\chi$.
  Hence, $|\Red(\tilde\chi)|$ is linear in $|\tilde\chi|$ and hence also in $O(n^m \cdot |\chi'|)$.
  Finally, the size of $\zeta$ is $O(n)$.
  Note that the whole argument only uses the timestamps occurring in the traces; it therefore applies verbatim to the discrete time domain.
\end{proof}

\begin{customlemma}{\ref{lemma:1-HyperpTPTL}}\leanok{Main.lean}{91}{main_lemma_20}
  If $\chi$ is in 1-HyperpTPTL, then $\zeta \land \Red(\tilde\chi)$ is in 1-pTPTL.
\end{customlemma}
\begin{proof}
  Assume that $\chi$ is in 1-HyperpTPTL, i.e., it only uses $t$ and $t'$ and satisfies restrictions (a) and (b) from Section~\ref{sec:information}.
  Clearly, $\tilde{\chi}$ is in 1-HyperpTPTL too, since it is a conjunction of copies of $\chi'$ with renamed trace variables.
  The translation $\Red$ does not introduce time variables or constraints, and it maps each subformula $t'.\;(C \land \psi)$ to $t'.\;(C \land \Red(\psi))$, where $\Red(\psi)$ has no free time variables since $\psi$ has none; hence restrictions (a) and (b) are preserved.
  Finally, $\zeta$ only uses $t$ and $t'$, and its only constraints occur in the subformula $t'.\;\lnot(t - t' \leq 0)$ and within $\Once_{[0,0]}$, hence the conclusion.
\end{proof}

\begin{customcorollary}{\ref{corollary:HyperpLTL sat}}\leanok{Main.lean}{97}{main_corollary_21}
  Suppose that $\chi$ is a HyperpLTL formula without $\HSince$ operators.
  Then $\chi$ is satisfiable iff $\Red'(\tilde\chi)$, which is a pLTL formula of size $O(n^m \cdot |\chi'|)$, is satisfiable.
\end{customcorollary}
\begin{proof}
  A~straightforward induction shows that $\Red'(\tilde\chi)$ is in pLTL.
  The size bound follows by the same argument as in the proof of Lemma~\ref{lemma:HyperpTPTL}.

  HyperpLTL formulae without $\HSince$ are invariant under transformations of the traces' timestamps, as long as monotonicity per trace and the alignment of all traces are preserved.
  We therefore assume without loss of generality that all timestamps are $0$ in the remainder of this proof.
  (The semantics of HyperpLTL is formally defined on timestamps, not time-points, and hence we continue to assume that traces are equipped with timestamps.)
  We also omit the mapping $v$ of time variables to timestamps in the semantics as it is not used by the fragment.

  It remains to show the correctness of the reduction.
  To this end, we introduce a simplified version of interleaving that omits all constraints related to timestamps, as these are irrelevant for non-metric formulae.
  Specifically, we call $s$ a \emph{non-metric interleaving} of $S$ if it satisfies condition (4) of Definition~\ref{def:interleaving}.
  We now justify the analogues of the preceding lemmata:
  \begin{itemize}
    \item Every non-empty trace is trivially a non-metric interleaving of its projections (cf.\ Lemma~\ref{lemma:zeta}).
    \item For every finite, non-empty set $S$ of traces, there exists a non-metric interleaving $\Interleave{S}$.
      This is a direct corollary of Lemma~\ref{lemma:ex interleaving}.
    \item Let $s$ be a non-metric interleaving of $S$.
      Then $\Tuple{\mathrm{prj}_\pi(s),i} \models \phi$ iff $\Tuple{s,j} \models \Red_\pi(\phi)$, where $j$ is the time-point of the $i$-th event in $s$ satisfying $\TP_\pi$.
      All relevant cases of Lemma~\ref{lemma:red aux}'s proof (i.e., excluding atoms $t_1 - t_2 \leq c$ and $t \leq c$, as well as freeze quantifiers) apply to non-metric interleavings as they only depend on the relationship between the interleaving and its projections.
    \item 
      Let $\chi$ be a HyperpLTL formula over trace variables $\{\pi_1,\dots,\pi_n\}$ and without trace quantifiers or $\HSince$.
      Let $s$ be a non-metric interleaving of $\img(\Pi)$.
      Then $\Tuple{\img(\Pi),\Pi,0} \hypermodels \chi$ iff $\Tuple{s,|s|-1} \models \Red'(\chi)$.
      The only case out of the relevant ones that is different from Lemma~\ref{lemma:red} is the one for $\chi \equiv \HLift{\phi}{\pi}$.

      Suppose that $\Tuple{\img(\Pi),\Pi,0} \hypermodels \HLift{\phi}{\pi}$.
      Thus $\Tuple{\Pi(\pi),j} \models \phi$, where $j$ is the last time-point in $\Pi(\pi)$.
      We have $\Pi(\pi) = \mathrm{prj}_\pi(s)$ because $s$ is a (non-metric) interleaving of $\img(\Pi)$.
      Using the modified Lemma~\ref{lemma:red aux} from above, it follows that $\Tuple{s,k} \models \Red_\pi(\phi)$, where $k$ is the $j$-th time-point in $s$ satisfying $\TP_\pi$.
      Moreover, there cannot be any $k' > k$ satisfying $\TP_\pi$, because this would contradict the fact that $j$ was maximal in $\mathrm{prj}_\pi(s)$.
      It follows that $\Tuple{s,|s|-1} \models \lnot \TP_\pi \Since (\TP_\pi \land \Red_\pi(\phi))$.

      For the converse direction, suppose that $\Tuple{s,|s|-1} \models \Red'(\chi)$.
      Hence we obtain $j \leq |s|-1$ where $\Tuple{s,j} \models \TP_\pi \land \Red_\pi(\phi)$, and there is no time-point $k \in \{j+1,\dots,|s|-1\}$ satisfying $\TP_\pi$.
      Let $\pi(j)$ be the time-point in $\mathrm{prj}_\pi(s) = \Pi(\pi)$ corresponding to $j$.
      By a similar argument as before, this is the last time-point in the projection, and $\Tuple{\mathrm{prj}_\pi(s),\pi(j)} \models \phi$.
      Hence $\Tuple{\img(\Pi),\Pi,0} \hypermodels \HLift{\phi}{\pi}$.
  \end{itemize}

  The rest of the proof is similar to Lemma~\ref{lemma:HyperpTPTL}, except that $\zeta$ and $\Tick$ are not needed for non-metric interleavings.
\end{proof}

\begin{customtheorem}{\ref{theorem:pc pmtl real-time lb}}\leanpart{Main.lean}{68}{main_lower_bound}
  \textsc{PolicyChange} for pMTL over real-time semantics has a non-primitive recursive lower bound; more precisely, it is $\mathbf{F}_{\omega^\omega}$-hard.
\end{customtheorem}
\begin{proof}
  We reduce satisfiability of pMTL over finite timed words (a formula $\phi$ is satisfiable iff $\Tuple{s,|s|-1} \models \phi$ for some non-empty trace $s$), which is non-primitive recursive~\cite{journal/lmcs/ouaknine2007decidability} and $\mathbf{F}_{\omega^\omega}$-hard~\cite{journals/toct/Schmitz16}.
  Given $\phi$, let $A$ be a fresh predicate and map $\phi$ to the pair $\Tuple{\Prev\phi,\Prev(\phi \land A)}$, where $\Prev = \Prev_{[0,\infty)}$ and $\phi \land A$ abbreviates $\lnot(\lnot\phi \lor \lnot A)$.
  This is a polynomial-time reduction; we show that \textsc{PolicyChange} returns 1 on $\Tuple{\Prev\phi,\Prev(\phi \land A)}$ iff $\phi$ is unsatisfiable.
  By Lemma~\ref{lemma:pc:iff}, it suffices to show that $\phi$ is unsatisfiable iff $s \deepbisim_{\Prev\phi} s'$ implies $s \deepbisim_{\Prev(\phi\land A)} s'$ for all $s,s'$.
  Recall from the proof of Lemma~\ref{lemma:deepbisim:indep} that $s \deepbisim_\psi s'$ iff $\Eval_v(s) \bisim_{o_v} \Eval_v(s')$ for all vertices $v$ of $\psi$'s syntax tree, and that this condition is trivially true for non-temporal vertices (Lemma~\ref{lemma:deepbisim:alt}).
  The temporal vertices of $\Prev(\phi \land A)$ are those of $\phi$ (with the same subformulae, hence the same conditions as in $\Prev\phi$) and the root, labeled by $\Prev X$.
  As in the proof of Lemma~\ref{lemma:foo:prev}, the condition at the root holds iff $s$ and $s'$ are both empty, or both non-empty and agree on $\phi \land A$ at their last time-point, or one is empty and the other violates $\phi \land A$ at its last time-point.

  If $\phi$ is unsatisfiable, so is $\phi \land A$, and the condition at the root is always true.
  Hence $s \deepbisim_{\Prev\phi} s'$ (which implies the conditions for all vertices of $\phi$) implies $s \deepbisim_{\Prev(\phi \land A)} s'$.

  If $\Tuple{s,|s|-1} \models \phi$, let $s'$ be obtained from $s$ by adding $A$ to the last event set.
  Since $A$ does not occur in $\phi$, every subformula of $\phi$ has the same truth value on $s$ and $s'$ at every time-point.
  Hence, for every vertex $v$ of $\Prev\phi$, the traces $\Eval_v(s)$ and $\Eval_v(s')$ only differ in the predicate $A$ at their last time-point; as the operator $o_v$ does not mention $A$, we have $\Eval_v(s)\cdot s'' \in \Sem{o_v}$ iff $\Eval_v(s')\cdot s'' \in \Sem{o_v}$ for all $s''$, i.e., $s \deepbisim_{\Prev\phi} s'$.
  However, $\phi \land A$ is false at the last time-point of $s$ and true at the last time-point of $s'$, so the condition at the root of $\Prev(\phi \land A)$ fails and $s \not\deepbisim_{\Prev(\phi \land A)} s'$.

  The reduction is independent of the time domain and maps pLTL formulae to pairs of pLTL formulae; we reuse it in Theorem~\ref{thm:pc lb discrete}.
\end{proof}

\subsection{Discrete-Time Semantics}
\label{apx:discrete}

Over discrete time, all delays, timestamps, and instants are natural numbers, and all continuations $s''$ in the definitions of $\bisim$ and $\deepbisim$ have natural delays.
We indicate the changes to the proofs of Section~\ref{sec:continuous}.
\begin{compactitem}
  \item Lemmas~\ref{lemma:deepbisim:indep}, \ref{lemma:pc:iff}, \ref{lemma:build:extract}, \ref{lemma:deepbisim:conv}, and~\ref{lemma:deepbisim:alt} do not depend on the time domain (in Lemma~\ref{lemma:build:extract}, condition~(3) is required for $d \in \Nat$).
  \item In Lemma~\ref{lemma:foo:since}, $d$ ranges over $\Nat$: in the direction ($\Rightarrow$), the continuation $(\{X\},d)$ is only available for $d \in \Nat$, and in the direction ($\Leftarrow$), the sum $d$ of the delays of $s''$ is a natural number.
  Accordingly, the information formula for $\phi_1 \Since_I \phi_2$ uses the conjunction over $d \in \Nat \cap (I \cup [0,\inf I])$, and Lemma~\ref{lemma:information formulae} holds with the same proof.
  Likewise, in Lemma~\ref{lemma:foo:prev}, the continuation $(\{X\},d)$ with $d \in I$ is only available if $I \cap \Nat \neq \emptyset$.
  If $I \cap \Nat = \emptyset$ (e.g., $I = (0,1)$), then $I$ is equivalent to $\emptyset$ over discrete time: $\Sem{\Prev_I X} = \Sem{X \Since_I Y} = \emptyset$, so all states of both testers are bisimilar, as in Lemma~\ref{lemma:foo:empty}.
  Accordingly, the information formulae for $\Prev_I \phi$ and $\phi_1 \Since_I \phi_2$ are those for $I = \emptyset$ whenever $I \cap \Nat = \emptyset$.
  \item The rewriting of Lemma~\ref{lemma:rewriting} is applied to the real interval $\hat{I}$ instead of $I$, as described at the end of its proof; the resulting formulae are anchored, in 1-HyperpTPTL, and of linear size.
  \item Lemmas~\ref{lemma:HyperpTPTL} and~\ref{lemma:1-HyperpTPTL} hold for discrete time, as the interleaving of Lemma~\ref{lemma:ex interleaving} only uses the timestamps of the encoded traces and $0$.
\end{compactitem}

\begin{customtheorem}{\ref{thm:pc ub discrete}}\leanpart{Main.lean}{60}{main_upper_bound_pltl}\leanpart{Main.lean}{52}{main_upper_bound}
  \textsc{PolicyChange} for pLTL is in PSPACE, and \textsc{PolicyChange} for pMTL with discrete-time semantics is in EXPSPACE.
\end{customtheorem}
\begin{proof}
  (ii) By the above, Algorithm~\ref{alg:upper-bound}, instantiated with the discrete-time information formulae, is correct when $\sattptl$ decides satisfiability of pTPTL over finite timed words with natural timestamps.
  The formula $\zeta \land \Red(\tilde\chi)$ has size polynomial in $|\phi_1| + |\phi_2|$, with constants encoded in binary.
  Satisfiability of TPTL over (infinite) timed state sequences with natural, weakly monotonic timestamps is in EXPSPACE~\cite{journals/jacm/AlurH94}.
  We reduce our problem to it as follows.
  First, by time reversal (see the proof of Theorem~\ref{thm:upper}), satisfiability of a past-time formula at the last position of a finite word reduces to satisfiability of the mirrored future-time formula at the first position of a finite word; the reversed timestamps $\tms_{k-1} - \tms_{k-1-i}$ are again natural and weakly monotonic.
  Second, a finite word is encoded as an infinite timed state sequence by appending infinitely many positions labeled with a fresh predicate $\mathit{end}$ and with timestamps increasing by $1$, and by relativizing every temporal operator to positions not labeled with $\mathit{end}$ (e.g., $\psi_1 \UUntil \psi_2$ becomes $\psi_1 \UUntil (\lnot \mathit{end} \land \psi_2)$, and the next operator requires $\lnot \mathit{end}$ at the next position), and by requiring that $\lnot \mathit{end}$ holds initially and that $\mathit{end}$, once true, remains true~\cite{conf/ijcai/DeGiacomo13}.
  Both steps are polynomial, hence \textsc{PolicyChange} for discrete-time pMTL is in EXPSPACE.

  (i) For pLTL formulae, the information formulae are in HyperpLTL, without $\HSince$, and of quadratic size (Lemma~\ref{cor:information formulae HyperpLTL}); no rewriting is needed.
  The formula $\chi$ computed by Algorithm~\ref{alg:upper-bound} is in the $\texists^*\tforall^*$ fragment with $n = m = 2$, so by Corollary~\ref{corollary:HyperpLTL sat}, it is satisfiable iff the pLTL formula $\Red'(\tilde\chi)$ of polynomial size is satisfiable over finite traces.
  Satisfiability of LTL over finite traces is in PSPACE~\cite{conf/ijcai/DeGiacomo13}, and pLTL over finite traces reduces to it by time reversal.
  Since pLTL formulae do not depend on timestamps, this argument applies to both time domains.
\end{proof}

\begin{customtheorem}{\ref{thm:pc lb discrete}}\leanpart{Main.lean}{68}{main_lower_bound}
  \textsc{PolicyChange} for pLTL is PSPACE-hard, and \textsc{PolicyChange} for pMTL with discrete-time semantics is EXPSPACE-hard.
\end{customtheorem}
\begin{proof}
  The reduction in the proof of Theorem~\ref{theorem:pc pmtl real-time lb} does not depend on the time domain, and it maps pLTL formulae to pairs of pLTL formulae.
  Satisfiability of pLTL over finite traces is PSPACE-hard (by time reversal from LTL over finite traces~\cite{conf/ijcai/DeGiacomo13}).
  Satisfiability of MTL with discrete time is EXPSPACE-hard~\cite{journals/iandc/AlurH93}; the hardness proof encodes accepting computations of exponential-space Turing machines, which are finite, so it also applies to finite words, and by time reversal to pMTL.
\end{proof}

\section{Policy Change for Non-Treelike Monitors}

\begin{customlemma}{\ref{lemma:pcstar:iff}}\leanok{Main.lean}{71}{main_pcstar_iff}
  For every generic monitor $\Monitor$, there exists a function implementing policy change from $\phi_1$ to $\phi_2$ in the sense of Definition~\ref{def:pcstar} iff ${\sim_{\phi_1}} \subseteq {\sim_{\phi_2}}$.
\end{customlemma}
\begin{proof}
  For a correct monitor $\Monitor^\phi$ and traces $s,s'$, Proposition~\ref{prop:monitor:simps} yields $\Out^\phi(\Step^\phi(\Init^\phi,s),s'') = \Out^\phi(\Init^\phi,s \cdot s'')$, hence $\Step^\phi(\Init^\phi,s) \bisim_{\Monitor^\phi} \Step^\phi(\Init^\phi,s')$ iff $s \sim_\phi s'$.
  With this fact in place of Lemma~\ref{lemma:deepbisim:conv}, the proof of Lemma~\ref{lemma:pc:iff} applies verbatim, replacing $\deepbisim$ by $\bisim$.
\end{proof}

\begin{customtheorem}{\ref{thm:pcstar pltl}}\leanpart{Main.lean}{74}{main_pcstar_lower_bound}
  \textsc{PolicyChange}* is PSPACE-complete for pLTL and EXPSPACE-complete for pMTL with discrete-time semantics.
\end{customtheorem}
\begin{proof}
  \emph{Monitors.}
  For a pLTL formula $\phi$, let $\Monitor_\phi$ be the deterministic monitor whose state after reading $s$ records whether $s$ is empty and, if not, the truth values of all subformulae of $\phi$ at the last time-point of $s$.
  The step and output functions are computable in polynomial time from the previous state and the new event, by the usual recursive characterization of $\Prev$ and $\Since$.
  For a discrete-time pMTL formula $\phi$, let $M$ be the largest finite interval bound in $\phi$ (so $M \leq 2^{|\phi|}$).
  The state of $\Monitor_\phi$ additionally records, for each subformula $\alpha \Since_I \beta$, the set of distances in $\{0,\dots,M\}$ of the witnesses of $\alpha \Since \beta$ (time-points satisfying $\beta$ followed only by time-points satisfying $\alpha$) and whether there is a witness at distance $> M$.
  On a new event with delay $\tmsdiff$, these sets are shifted by $\tmsdiff$ (distances beyond $M$ are merged), emptied if $\alpha$ fails at the new time-point, and extended with $0$ if $\beta$ holds; then $\alpha \Since_I \beta$ holds iff the set intersects $I$ (or $I$ is unbounded and there is a witness at distance $> M \geq \inf I$).
  The state consists of $O(|\phi| \cdot M)$ bits, and the step function is computable in time polynomial in the size of the state.
  Moreover, all delays $\tmsdiff > M+1$ have the same effect as $M+1$, so we may restrict the input alphabet to $2^{\Props'} \times \{0,\dots,M+1\}$, where $\Props'$ is the set of predicates occurring in $\phi_1$ or $\phi_2$ (other predicates do not affect the monitors).

  \emph{Upper bounds.}
  Let $\Monitor_1 = \Monitor_{\phi_1}$ and $\Monitor_2 = \Monitor_{\phi_2}$, and let $N$ bound the number of bits of their states.
  By Lemma~\ref{lemma:pcstar:iff} and the first part of its proof, \textsc{PolicyChange}* outputs 0 iff there are $s,s'$ such that $q_1 \bisim q'_1$ and $q_2 \not\bisim q'_2$, where $q_k$ and $q'_k$ are the states reached by $\Monitor_k$ on $s$ and $s'$.
  Two states $q,q'$ of $\Monitor_k$ are not bisimilar iff some non-empty continuation leads to different outputs; such a continuation can be guessed letter by letter while storing the two current states, so non-bisimilarity is decidable in nondeterministic space $O(N)$, and bisimilarity in deterministic space $O(N^2)$ by Savitch's theorem.
  The prefixes $s$ and $s'$ can be guessed letter by letter as well, storing $q_1,q_2$ (resp.\ $q'_1,q'_2$).
  Overall, we obtain a nondeterministic algorithm using space polynomial in $N$, hence a deterministic one by Savitch's theorem.
  As $N$ is polynomial in $|\phi_1|+|\phi_2|$ for pLTL and exponential for discrete-time pMTL, we obtain PSPACE and EXPSPACE upper bounds, respectively.

  \emph{Lower bounds.}
  We reuse the reduction $\phi \mapsto \Tuple{\Prev\phi,\Prev(\phi\land A)}$.
  If $\phi$ is unsatisfiable, then $\Sem{\Prev(\phi \land A)} = \emptyset$, so ${\sim_{\Prev(\phi\land A)}}$ relates all traces and \textsc{PolicyChange}* outputs 1.
  If $\Tuple{s,|s|-1} \models \phi$, let $s'$ be $s$ with $A$ added to the last event set.
  For every $s''$, the traces $s \cdot s''$ and $s'\cdot s''$ agree on all $A$-free formulae at all time-points, so $s \sim_{\Prev\phi} s'$.
  For $s'' = (\emptyset,0)$, however, $s \cdot s'' \notin \Sem{\Prev(\phi\land A)}$ and $s'\cdot s'' \in \Sem{\Prev(\phi\land A)}$, so $s \not\sim_{\Prev(\phi \land A)} s'$, and \textsc{PolicyChange}* outputs 0.
  Hardness follows as in Theorem~\ref{thm:pc lb discrete}.
\end{proof}

\section{Lean Formalization}
\label{apx:lean}

All definitions needed to state our main results, and the reductions underlying Theorems~\ref{thm:upper}, \ref{theorem:pc pmtl real-time lb}, \ref{thm:pc ub discrete}, and~\ref{thm:pc lb discrete} and Section~\ref{sec:general}, are formalized in Lean~4 with Mathlib (available at \url{https://github.com/francois-hublet/policy-change-lean}), for both time domains.
The definitions and the statements are collected in a single file, \texttt{Spec.lean}; the file \texttt{Main.lean} proves all statements, using only Lean's standard axioms and no \texttt{sorry}.
The reduction of Algorithm~\ref{alg:upper-bound} is an executable Lean function, which is proved to map pairs of pMTL formulae to 1-pTPTL formulae of quadratic size that are unsatisfiable iff policy change is possible.
For pLTL inputs, a second executable reduction (information formulae without rewriting, $\Red'$ instead of $\Red$, and no $\zeta$; Lemma~\ref{cor:information formulae HyperpLTL} and Corollary~\ref{corollary:HyperpLTL sat}) is proved to map pairs of pLTL formulae to pLTL formulae of quadratic size with the same property, for both time domains.
The complexity results we rely on (Sections~\ref{sec:continuous}--\ref{sec:general}) are not formalized.
In Lean, the outer level of HyperpTPTL is evaluated at instants in $\mathbb{R}$ rather than $\Re$.
For anchored formulae, the two semantics coincide at instants in $\Re$, as the right operands of $\HSince$ and $\HSinceS$ only hold at instants in $T(\Pi) \subseteq \Re$; in particular, \texttt{main\_lemma\_A8} does not need the anchoring hypothesis of Lemma~\ref{lemma:hsat ea iff}.

Throughout the paper, a badge \leanfull{} in the margin indicates a result or definition that is fully formalized, and \leanpartial{} a result whose reduction part (but not the complexity bound obtained from it) is formalized; both badges and the declaration names below them are links to the Lean sources.

\begin{center}
\small
\renewcommand{\arraystretch}{1.15}
\begin{tabular}{@{}p{0.39\textwidth}@{\hspace{6pt}}l@{\hspace{6pt}}l@{\hspace{4pt}}l@{}}
  \textbf{Paper} & \textbf{Lean declaration} & \textbf{File} & \\
  \hline
  pMTL syntax and semantics (Sec.~\ref{sec:prelim}) & \leanref{Spec.lean}{72}{PMTL} & \texttt{Spec} & \leanfull \\
  Online monitors, correctness & \leanref{Spec.lean}{146}{Monitor.Correct} & \texttt{Spec} & \leanfull \\
  Temporal testers, monitor trees (Sec.~\ref{sec:treelike}) & \leanref{Spec.lean}{191}{MonTree} & \texttt{Spec} & \leanfull \\
  Treelike monitors (Def.~\ref{def:treelike}) & \leanref{Spec.lean}{274}{TreelikeMonitor} & \texttt{Spec} & \leanfull \\
  Policy change, \textsc{PolicyChange} (Def.~\ref{def:pc}) & \leanref{Spec.lean}{295}{PolicyChange} & \texttt{Spec} & \leanfull \\
  Lemma~\ref{lemma:deepbisim:indep} & \leanref{Main.lean}{17}{main_deep_indep} & \texttt{Main} & \leanfull \\
  Lemma~\ref{lemma:pc:iff} & \leanref{Main.lean}{20}{main_pc_iff} & \texttt{Main} & \leanfull \\
  Lemma~\ref{lemma:build:extract} & \leanref{Main.lean}{26}{main_lemma_12} & \texttt{Main} & \leanfull \\
  Lemmas~\ref{lemma:deepbisim:alt}, \ref{lemma:foo:prev}, \ref{lemma:foo:since} & \leanref{Info.lean}{300}{deepEquiv_iff_infoSem} & \texttt{Info} & \leanfull \\
  Lemma~\ref{lemma:foo:empty} & \leanref{Info.lean}{247}{opEquiv_empty} & \texttt{Info} & \leanfull \\
  pTPTL, 1-pTPTL & \leanref{Spec.lean}{372}{PTPTL.IsOneVar} & \texttt{Spec} & \leanfull \\
  HyperpTPTL, $\texists^*\tforall^*$ fragment & \leanref{Quantified.lean}{52}{EASatisfiable} & \texttt{Quantified} & \leanfull \\
  Lemma~\ref{lemma:pmtl to ptptl} & \leanref{Hyper.lean}{48}{sat_emb} & \texttt{Hyper} & \leanfull \\
  Information formulae & \leanref{UpperBound.lean}{410}{hsat_BE} & \texttt{UpperBound} & \leanfull \\
  Lemma~\ref{lemma:information formulae} & \leanref{Main.lean}{36}{main_lemma_16} & \texttt{Main} & \leanfull \\
  Lemma~\ref{cor:information formulae HyperpLTL} & \leanref{Main.lean}{43}{main_lemma_17} & \texttt{Main} & \leanfull \\
  Lemma~\ref{lemma:separation} & \leanref{Separation.lean}{122}{sepWindow_bounded} & \texttt{Separation} & \leanfull \\
  Lemma~\ref{lemma:rewriting} & \leanref{Size.lean}{230}{lemma_rewriting} & \texttt{Size} & \leanfull \\
  Anchored formulae (Def.~\ref{def:anchored}) & \leanref{Hyper.lean}{139}{Anchored} & \texttt{Hyper} & \leanfull \\
  Interleavings, $\rho$, $\zeta$ (Lemmas~\ref{lemma:zeta}--\ref{lemma:red}) & \leanref{Interleave.lean}{197}{hsat_rho} & \texttt{Interleave} & \leanfull \\
  Lemma~\ref{lemma:hsat ea iff} & \leanref{Main.lean}{77}{main_lemma_A8} & \texttt{Main} & \leanfull \\
  Lemma~\ref{lemma:HyperpTPTL} & \leanref{Main.lean}{85}{main_lemma_19} & \texttt{Main} & \leanfull \\
  Lemma~\ref{lemma:1-HyperpTPTL} & \leanref{Main.lean}{91}{main_lemma_20} & \texttt{Main} & \leanfull \\
  Corollary~\ref{corollary:HyperpLTL sat} & \leanref{Main.lean}{97}{main_corollary_21} & \texttt{Main} & \leanfull \\
  Algorithm~\ref{alg:upper-bound} & \leanref{Algorithm.lean}{266}{reduce} & \texttt{Algorithm} & \leanfull \\
  Algorithm~\ref{alg:upper-bound} for pLTL ($\Red'$, no $\zeta$) & \leanref{Algorithm.lean}{311}{reducePLTL} & \texttt{Algorithm} & \leanfull \\
  Thms.~\ref{thm:upper}, \ref{thm:pc ub discrete}: correctness, 1-pTPTL & \leanref{Main.lean}{52}{main_upper_bound} & \texttt{Main} & \leanpartial \\
  Thms.~\ref{thm:upper}, \ref{thm:pc ub discrete}: polynomial size & \leanref{Main.lean}{56}{main_upper_bound_size} & \texttt{Main} & \leanpartial \\
  Thm.~\ref{thm:pc ub discrete}\,(i): correctness, pLTL & \leanref{Main.lean}{60}{main_upper_bound_pltl} & \texttt{Main} & \leanpartial \\
  Thm.~\ref{thm:pc ub discrete}\,(i): polynomial size & \leanref{Main.lean}{64}{main_upper_bound_pltl_size} & \texttt{Main} & \leanpartial \\
  Thms.~\ref{theorem:pc pmtl real-time lb}, \ref{thm:pc lb discrete}: reduction & \leanref{Main.lean}{68}{main_lower_bound} & \texttt{Main} & \leanpartial \\
  Lemma~\ref{lemma:pcstar:iff} & \leanref{Main.lean}{71}{main_pcstar_iff} & \texttt{Main} & \leanfull \\
  Thm.~\ref{thm:pcstar pltl}: lower-bound reduction & \leanref{Main.lean}{74}{main_pcstar_lower_bound} & \texttt{Main} & \leanpartial \\
\end{tabular}
\end{center}

\bibliography{policy_change}

\begin{thebibliography}{10}

\bibitem{journals/iandc/AlurH93}
Rajeev Alur and Thomas~A. Henzinger.
\newblock Real-time logics: Complexity and expressiveness.
\newblock {\em Inf. Comput.}, 104(1):35--77, 1993.
\newblock \href {https://doi.org/10.1006/inco.1993.1025}
  {\path{doi:10.1006/inco.1993.1025}}.

\bibitem{journals/jacm/AlurH94}
Rajeev Alur and Thomas~A. Henzinger.
\newblock A really temporal logic.
\newblock {\em J. {ACM}}, 41(1):181--204, 1994.
\newblock \href {https://doi.org/10.1145/174644.174651}
  {\path{doi:10.1145/174644.174651}}.

\bibitem{journals/fmsd/BasinBKT19}
David Basin, Bhargav~Nagaraja Bhatt, Sr{\dj}an Krsti{\'c}, and Dmitriy Traytel.
\newblock Almost event-rate independent monitoring.
\newblock {\em Formal Methods Syst. Des.}, 54(3):449--478, 2019.
\newblock \href {https://doi.org/10.1007/s10703-018-00328-3}
  {\path{doi:10.1007/s10703-018-00328-3}}.

\bibitem{journals/jacm/BasinKMZ15}
David Basin, Felix Klaedtke, Samuel M{\"{u}}ller, and Eugen Z{\u{a}}linescu.
\newblock Monitoring metric first-order temporal properties.
\newblock {\em J. {ACM}}, 62(2):15:1--15:45, 2015.
\newblock \href {https://doi.org/10.1145/2699444} {\path{doi:10.1145/2699444}}.

\bibitem{conf/isola/BonakdarpourDP18}
Borzoo Bonakdarpour, Jyotirmoy~V. Deshmukh, and Miroslav Pajic.
\newblock Opportunities and challenges in monitoring cyber-physical systems
  security.
\newblock In Tiziana Margaria and Bernhard Steffen, editors, {\em 8th Int.
  Symp. on Leveraging Applications of Formal Methods, Verification and
  Validation ({ISoLA} 2018), Part {IV}}, volume 11247 of {\em LNCS}, pages
  9--18. Springer, 2018.
\newblock \href {https://doi.org/10.1007/978-3-030-03427-6\_2}
  {\path{doi:10.1007/978-3-030-03427-6\_2}}.

\bibitem{conf/seams/CarwehlVRG23}
Marc Carwehl, Thomas Vogel, Gena{\'{\i}}na Nunes~Rodrigues, and Lars Grunske.
\newblock Runtime verification of self-adaptive systems with changing
  requirements.
\newblock In {\em 18th {IEEE/ACM} Symp. on Software Engineering for Adaptive
  and Self-Managing Systems ({SEAMS} 2023)}, pages 104--114. {IEEE}, 2023.
\newblock \href {https://doi.org/10.1109/SEAMS59076.2023.00024}
  {\path{doi:10.1109/SEAMS59076.2023.00024}}.

\bibitem{conf/post/ClarksonFKMRS14}
Michael~R. Clarkson, Bernd Finkbeiner, Masoud Koleini, Kristopher~K. Micinski,
  Markus~N. Rabe, and C{\'{e}}sar S{\'{a}}nchez.
\newblock Temporal logics for hyperproperties.
\newblock In Mart{\'{\i}}n Abadi and Steve Kremer, editors, {\em 3rd Int. Conf.
  on Principles of Security and Trust ({POST} 2014)}, volume 8414 of {\em
  LNCS}, pages 265--284. Springer, 2014.
\newblock \href {https://doi.org/10.1007/978-3-642-54792-8\_15}
  {\path{doi:10.1007/978-3-642-54792-8\_15}}.

\bibitem{journals/jcs/ClarksonS10}
Michael~R. Clarkson and Fred~B. Schneider.
\newblock Hyperproperties.
\newblock {\em J. Comput. Secur.}, 18(6):1157--1210, 2010.
\newblock \href {https://doi.org/10.3233/JCS-2009-0393}
  {\path{doi:10.3233/JCS-2009-0393}}.

\bibitem{conf/iwssd/FeatherFLP98}
Martin~S. Feather, Stephen Fickas, Axel van Lamsweerde, and Christophe Ponsard.
\newblock Reconciling system requirements and runtime behavior.
\newblock In {\em 9th Int. Workshop on Software Specification and Design
  ({IWSSD} 1998)}, pages 50--59. {IEEE} Computer Society, 1998.
\newblock \href {https://doi.org/10.5555/857205.858297}
  {\path{doi:10.5555/857205.858297}}.

\bibitem{conf/concur/FinkbeinerH16}
Bernd Finkbeiner and Christopher Hahn.
\newblock Deciding hyperproperties.
\newblock In Jos{\'{e}}e Desharnais and Radha Jagadeesan, editors, {\em 27th
  Int. Conf. on Concurrency Theory ({CONCUR} 2016)}, volume~59 of {\em LIPIcs},
  pages 13:1--13:14. Schloss Dagstuhl -- Leibniz-Zentrum f{\"{u}}r Informatik,
  2016.
\newblock \href {https://doi.org/10.4230/LIPIcs.CONCUR.2016.13}
  {\path{doi:10.4230/LIPIcs.CONCUR.2016.13}}.

\bibitem{conf/icse/GhezziGM12}
Carlo Ghezzi, Joel Greenyer, and Valerio Panzica~La Manna.
\newblock Synthesizing dynamically updating controllers from changes in
  scenario-based specifications.
\newblock In Hausi~A. M{\"{u}}ller and Luciano Baresi, editors, {\em 7th Int.
  Symp. on Software Engineering for Adaptive and Self-Managing Systems ({SEAMS}
  2012)}, pages 145--154. {IEEE} Computer Society, 2012.
\newblock \href {https://doi.org/10.1109/SEAMS.2012.6224401}
  {\path{doi:10.1109/SEAMS.2012.6224401}}.

\bibitem{conf/ijcai/DeGiacomo13}
Giuseppe~De Giacomo and Moshe~Y. Vardi.
\newblock Linear temporal logic and linear dynamic logic on finite traces.
\newblock In Francesca Rossi, editor, {\em 23rd Int. Joint Conf. on Artificial
  Intelligence ({IJCAI} 2013)}, pages 854--860. {IJCAI/AAAI}, 2013.

\bibitem{book/Haase10}
Christoph Haase, Jo{\"e}l Ouaknine, and James Worrell.
\newblock {\em On Process-Algebraic Extensions of Metric Temporal Logic}, pages
  283--300.
\newblock Springer London, London, 2010.
\newblock \href {https://doi.org/10.1007/978-1-84882-912-1_13}
  {\path{doi:10.1007/978-1-84882-912-1_13}}.

\bibitem{fmsd/HavelundPU20}
Klaus Havelund, Doron Peled, and Dogan Ulus.
\newblock First-order temporal logic monitoring with {BDDs}.
\newblock {\em Formal Methods Syst. Des.}, 56(1):1--21, 2020.
\newblock \href {https://doi.org/10.1007/s10703-018-00327-4}
  {\path{doi:10.1007/s10703-018-00327-4}}.

\bibitem{journals/ic/HoZJ21}
Hsi{-}Ming Ho, Ruoyu Zhou, and Timothy~M. Jones.
\newblock Timed hyperproperties.
\newblock {\em Inf. Comput.}, 280:104639, 2021.
\newblock \href {https://doi.org/10.1016/j.ic.2020.104639}
  {\path{doi:10.1016/j.ic.2020.104639}}.

\bibitem{journals/rts/Koymans90}
Ron Koymans.
\newblock Specifying real-time properties with metric temporal logic.
\newblock {\em Real Time Syst.}, 2(4):255--299, 1990.
\newblock \href {https://doi.org/10.1007/BF01995674}
  {\path{doi:10.1007/BF01995674}}.

\bibitem{laroussinie2006efficient}
Fran{\c{c}}ois Laroussinie, Nicolas Markey, and Ph~Schnoebelen.
\newblock Efficient timed model checking for discrete-time systems.
\newblock {\em Theoretical Computer Science}, 353(1-3):249--271, 2006.

\bibitem{conf/tacas/LimaHRTY23}
Leonardo Lima, Andrei Herasimau, Martin Raszyk, Dmitriy Traytel, and Simon
  Yuan.
\newblock Explainable online monitoring of metric temporal logic.
\newblock In Sriram Sankaranarayanan and Natasha Sharygina, editors, {\em 29th
  Int. Conf. on Tools and Algorithms for the Construction and Analysis of
  Systems (TACAS 2023), Part II}, volume 13994 of {\em LNCS}, pages 473--491.
  Springer, 2023.
\newblock \href {https://doi.org/10.1007/978-3-031-30820-8\_28}
  {\path{doi:10.1007/978-3-031-30820-8\_28}}.

\bibitem{journal/lmcs/ouaknine2007decidability}
Jo{\"e}l Ouaknine and James Worrell.
\newblock On the decidability and complexity of metric temporal logic over
  finite words.
\newblock {\em Logical Methods in Computer Science}, 3, 2007.

\bibitem{conf/fm/PnueliZ06}
Amir Pnueli and Aleksandr Zaks.
\newblock {PSL} model checking and run-time verification via testers.
\newblock In Jayadev Misra, Tobias Nipkow, and Emil Sekerinski, editors, {\em
  14th Int. Symp. on Formal Methods ({FM} 2006)}, volume 4085 of {\em LNCS},
  pages 573--586. Springer, 2006.
\newblock \href {https://doi.org/10.1007/11813040\_38}
  {\path{doi:10.1007/11813040\_38}}.

\bibitem{conf/atva/RaszykBKT19}
Martin Raszyk, David Basin, Sr{\dj}an Krsti{\'c}, and Dmitriy Traytel.
\newblock Multi-head monitoring of metric temporal logic.
\newblock In Yu{-}Fang Chen, Chih{-}Hong Cheng, and Javier Esparza, editors,
  {\em 17th Int. Symp. on Automated Technology for Verification and Analysis
  (ATVA 2019)}, volume 11781 of {\em LNCS}, pages 151--170. Springer, 2019.
\newblock \href {https://doi.org/10.1007/978-3-030-31784-3\_9}
  {\path{doi:10.1007/978-3-030-31784-3\_9}}.

\bibitem{journals/toct/Schmitz16}
Sylvain Schmitz.
\newblock Complexity hierarchies beyond elementary.
\newblock {\em {ACM} Trans. Comput. Theory}, 8(1):3:1--3:36, 2016.

\bibitem{journals/entcs/ThatiR05}
Prasanna Thati and Grigore Ro{\c{s}}u.
\newblock Monitoring algorithms for metric temporal logic specifications.
\newblock In Klaus Havelund and Grigore Ro{\c{s}}u, editors, {\em 4th Workshop
  on Runtime Verification ({RV@ETAPS} 2004)}, volume 113 of {\em Electronic
  Notes in Theoretical Computer Science}, pages 145--162. Elsevier, 2004.
\newblock \href {https://doi.org/10.1016/j.entcs.2004.01.029}
  {\path{doi:10.1016/j.entcs.2004.01.029}}.

\bibitem{Yuan22}
Simon Yuan.
\newblock Policy change in {MonPoly}.
\newblock Master's thesis, {ETH} Zurich, Switzerland, 2022.

\end{thebibliography}

\end{document}